\documentclass{imsart}

\RequirePackage{amsthm,amsmath,amsfonts,amssymb}
\RequirePackage[numbers,sort&compress]{natbib}
\RequirePackage[colorlinks,citecolor=blue,urlcolor=blue]{hyperref}
\RequirePackage{graphicx}
\RequirePackage{subcaption} 
\RequirePackage{commath} 
\RequirePackage{dsfont}
\RequirePackage{enumerate}

\newcommand{\var}{\mathrm{var}}
\newcommand{\E}{\mathbb{E}}
\newcommand{\p}{\mathbb{P}}
\newcommand{\cov}{\mathrm{cov}}
\newcommand{\1}{\mathds{1}}

\startlocaldefs
\theoremstyle{plain}

\newtheorem{theorem}{Theorem}[section]
\newtheorem{proposition}[theorem]{Proposition}
\newtheorem{lemma}[theorem]{Lemma}
\theoremstyle{remark}

\newtheorem{remark}{Remark}

\endlocaldefs

\begin{document}

\begin{frontmatter}
\title{Time-lagged marginal expected shortfall}
\runtitle{TMES}

\begin{aug}
\author[A]{\fnms{Jiajun}~\snm{Liu}\ead[label=e1]{jiajun.liu@xjtlu.edu.cn}}
\author[A]{\fnms{Xuannan}~\snm{Liu}\ead[label=e2]{xuannan.liu18@student.xjtlu.edu.cn}}
\and
\author[A]{\fnms{Yuwei}~\snm{Zhao}\ead[label=e3]{yuwei.zhao@xjtlu.edu.cn}}
\address[A]{Department of Financial and Actuarial Mathematics,
Xi'an-Jiaotong Liverpool University, Suzhou, China \printead[presep={ ,\ }]{e1,e2,e3}}

\end{aug}

\begin{abstract}
Marginal expected shortfall (MES) is an important measure when assessing and quantifying the contribution of the financial institution to a systemic crisis.
In this paper, we propose {\em time-lagged marginal expected shortfall}
(TMES) as a dynamic extension of the MES, accounting for time lags in assessing systemic risks. A natural estimator for the TMES is proposed, and its asymptotic properties are
studied. To address challenges in constructing confidence intervals for the TMES in practice, we apply the stationary bootstrap method to generate confidence bands for the TMES estimator. Extensive simulation studies were conducted to investigate the asymptotic properties of empirical and bootstrapped TMES. Two practical applications of TMES, supported by real data analyses, effectively demonstrate its ability to account for time lags in risk assessment.

\end{abstract}



\end{frontmatter}


\section{Introduction}

The 2007--2009 financial crisis and the COVID--19 pandemic attract attention from the academic and industry
to the study of systemic risks. The {\em marginal expected shortfall} (MES), one of the popular systemic risk measures, is used to assess how a systemic crisis affects a specific component, i.e., 
\begin{align*}
\mathrm{MES}(u)= \E [X \mid Y > u]\,,
\end{align*}  
where $(X,Y)$ is a bivariate random vector.  The random variable $X$ represents the value of a product in the market and the random variable $Y$ represents the loss of a market
portfolio, or an index of macroeconomic factors. The number $u \in \mathbb{R}$ is a high threshold such that the event $\{Y >
u\}$ stands for the occurrence of a systemic event and $\p(Y > u)$ is extremely
small. In practice, those values mentioned above are usually modeled as time series.
In this paper, we are interested in a
dynamic version of the MES to capture the time effect of systemic
risks. 
Suppose that we observe a time series of bivariate vectors
$(X_t, Y_t)_{t\in \mathbb{Z}}$ and the {\em time-lagged marginal expected shortfall}
(TMES) at lag $h\ge 0$ is given by
\begin{align}\label{eq:tmes0000}
\delta(h) = \lim_{u\to +\infty} \E[X_t \mid Y_{t-h} > u ] \,,
\end{align}
if existed. This type of extremal conditional expectation appears in different forms across other studies. The {\em extremogram} introduced in \cite{davis2009} is defined for a strictly stationary heavy-tailed time series $(X_t)$ and is given by \begin{align}\label{eq:extremogram}
\rho (h) = \lim_{u \to + \infty} \E[ \1(X_t >u) \mid X_{t-h} >u]\,,
\end{align}
if existed, where $\1(\cdot)$ is an indicator function satisfying 
\begin{align*}
  \1(A) = \begin{cases}
    1\,, & \text{the event } A \text{ happens}\,,\\
    0\,, & \text{otherwise}\,.
\end{cases}
\end{align*}
Here $\1 (X_t >u)$ follows a Bernoulli distribution which is not heavy-tailed.  
A {\em causal tail coefficient} introduced in \cite{gnecco2021causal} is defined for a heavy-tailed random vector $(X_i)_{i=1,\ldots,H}$ for some positive integer $H$ and is given by  
\begin{align} \label{eq:causaltail}
\lim_{u\to + \infty} \E[F_{X_i}(X_i) \mid X_j >u]\,, \quad i\neq j\,.  
\end{align}
Given that $X_i$ is continuous with its cumulative distribution function $F_{X_i}$, $F_{X_i}(X_i)$ follows a uniform distribution on the interval $[0,1]$. Therefore, besides the case when both $X_t$ and $Y_{t-h}$ in \eqref{eq:tmes0000} are heavy-tailed, we are also interested in the case when $X_t$ is light-tailed and $Y_{t-h}$ is heavy-tailed.

The essential condition for the existence of the TMES is {\em partial regular
  variation}, a variant of regular variation. We say that the time series $(X_t, Y_t)$ is partially
regularly varying if $(Y_t)$ is regularly varying. Regular variation
is essential for the existence of the limit $\lim_{u \to +\infty}
\mathrm{MES}(u)$ and the extremogram~\eqref{eq:extremogram}; see \cite{Cai2014RS,
  Cai2020, Goegebeur2021, CHEN2022238, davis2009} for example. The classical
definition of regular variation of a time series introduced in
\cite{basrak2009} relies on 
multivariate regular variation and vague
convergence; please refer to \cite{embrechts1997,resnick2007} for more
details. By replacing the vague convergence with the $M_0$-convergence
introduced in \cite{hult2006}, Segers~et~al.~\cite{segers2017}
extended the regular variation concept to a random element taking
values from a general metric space, and consequently, provided a new
definition of regular variation of a time series. Following this idea,
we will use the $M_O$-convergence studied in \cite{lindskog2014} to
replace the $M_0$-convergence for the definition of partial regular
variation.

We are interested in the
estimation problem of the TMES and its applications in practice. In a
series of papers by Cai~et~al.~\cite{Cai2014RS, Cai2020}, the
estimation of the MES is studied under various settings. Estimating the extremogram, as a special
example of TMES, and its applications in the time and frequency
domain are studied in a series of papers by
Mikosch~et~al.~\cite{davis2009, davis2013measures,davis2012142,
  mikosch2014, mikosch2015}. The TMES for a heavy-tailed time series
and its estimation problem are studied in \cite{goegebeur2024marginal}. In this paper, we will
construct an estimator for the TMES and derive its asymptotic
distribution. Since our proposed TMES estimator does not depend on any
parametric model, we propose a bootstrap algorithm to derive the
asymptotic distribution of the TMES estimator and consequently to
generate confidence intervals for the TMES estimator.  

The contribution of this work is threefold. First, we propose TMES as a novel measure and it serves as a dynamic extension of the MES in literature. By incorporating time lags in risk assessment, TMES provides a more nuanced approach to evaluating risks over time. Second, we show the existence of the TMES under the partial regular variation of the time series, allowing TMES to accommodate a variety of marginal distributions, including both light-tailed and heavy-tailed scenarios. Third, a natural estimator for TMES is developed, along with the study of its asymptotic properties, ensuring theoretical soundness and practical applicability. To address challenges in constructing confidence intervals for TMES, we employ the stationary bootstrap method to generate confidence bands. This enhances the reliability of TMES estimates in real-world applications.

The organization of the paper is as follows. In
Section~\ref{sec:tmes}, we introduce the TMES, partial regular variation and the existence
conditions for TMES with theoretical examples. The asymptotic properties of the
empirical TMES as the natural estimator of the TMES are studied in
Section~\ref{sec:empiricaltmes}. In Section~\ref{sec:bootstrap}, we
propose a stationary bootstrap algorithm for the empirical
TMES. According to simulations of two theoretical examples in
Section~\ref{sec:tmes}, the asymptotic properties of the empirical
TMES will be verified in Section~\ref{sec:simulation}. Finally, we
show applications of the TMES to two real data in
Section~\ref{sec:realdataanalyses}. For the sake of clarity, technical proofs of Theorem~\ref{thm:morvdef}, Proposition~\ref{prop:consistency}, Theorem~\ref{thm:clt} and Theorem~\ref{thm:bootclt} are gathered in the Appendix. 

\section{Time-lagged marginal expected shortfall} \label{sec:tmes}
In this section, we define the TMES and provide the essential conditions for its existence.
Let $(X_t, \mathbf{Y}_t)_{t\in \mathbb{Z}}$ be a strictly stationary time series 
taking values from $\mathbb{R}^{1+l}$ for an integer $l\ge 1$, where $X_t$ is a random variable and
$\mathbf{Y}_t$ is a random vector taking values from $\mathbb{R}^l$. Let the
set $A \subset \mathbb{R}^l$ be {\em bounded away from the origin}, i.e. $\inf_{x\in
  A} \|x\| >0$ with $\|\cdot \|$ as a norm in $\mathbb{R}^l$. The {\em time-lagged marginal
  expected shortfall} of $(X_t, \mathbf{Y}_t)$ at lag $h\ge 0$ is given
by 
\begin{align} \label{eq:mes00}
\delta_A (h) = \delta(h) = \lim_{u\to +\infty} \E [X_t \mid u^{-1} \mathbf{Y}_t \in
  A]\,. 
\end{align}
For ease of notation, we will prefer $\delta(h)$, which suppresses the dependence of $\delta_A(h)$ on the set $A$. By choosing $l=1$ and
$A=(1, +\infty)$, $\delta(0)$ agrees with the limit $\lim_{u\to +\infty} \mathrm{MES}(u)$.  

If $X_t$ and $\mathbf{Y}_t$ are independent, $\delta(h) = \E[X_t]$
trivially. We will think about a {\em centered TMES} at lag $h\ge 0$, $\delta_0(h) = \delta(h) - \E[X_t]$ in real data analysis. Moreover,
\begin{align*}
\delta_0 (h) & = \lim_{u\to \infty} \E[X_t -\E[X_t] \mid u^{-1} \mathbf{Y}_{t-h} \in A] \\
& = \lim_{u\to \infty}
       \frac{1}{\p(u^{-1} \mathbf{Y}_{t-h} \in A)} \E\big[ (X_t -\E[X_t]) \1(u^{-1}\mathbf{Y}_{t-h} \in A)
       \big]\\ 
& = \lim_{u\to \infty} \frac{1}{\p(u^{-1} \mathbf{Y}_{t-h} \in A)} \E
  \big[ (X_t -\E[X_t]) (\1(u^{-1}\mathbf{Y}_{t-h} \in A)-\p(u^{-1} \mathbf{Y}_{t-h} \in A) \big]\\ 
& = \lim_{u\to \infty} \frac{1}{\p(u^{-1} \mathbf{Y}_{t-h} \in A)} \cov
  \big( X_t, \1(u^{-1}\mathbf{Y}_{t-h} \in  A) \big)\,. 
\end{align*}
Then $\delta_0(h)$ is
asymptotically a cross-covariance function of $(X_t,\1(u^{-1}\mathbf{Y}_{t-h}\in A)).$

\subsection{Partial regular variation}
The definition of partial regular variation relies on the
$M_O$-convergence. We follow the arguments in \cite{lindskog2014} to
introduce the $M_O$-convergence of measures on a complete and
separable space $(S,d)$ with $d$ as metric and its Borel $\sigma$-algebra is denoted by
$\mathcal{B}$. We fix a closed subset
$\mathbb{C} \in \mathcal{B}$ and $\mathbb{O} = S\setminus \mathbb{C}$. We will denote the
space by $(S,d,\mathbb{C})$. The
sub-$\sigma$-algebra $\mathcal{B}_O$ is given by $\mathcal{B}_O = \{A \cap \mathbb{O}: A\in \mathcal{B} \}$. We suppose that the
space $(S,d,\mathbb{C})$ is equipped with a {\em
  partial scalar multiplication} $<\lambda, x> : [0, +\infty) \times S \to S$ satisfying
that for $\lambda, \lambda_1, \lambda_2 \in [0, +\infty)$ and $x\in S$:
\begin{enumerate}
\item the mapping $<\lambda, x> \mapsto \lambda x$ is continuous;
\item $<1, x> = x$, $<\lambda_1, <\lambda_2, x>> = <\lambda_1 \lambda_2, x>$, $<0, x> \in \mathbb{C}$;
\item $d(x, \mathbb{C}) < d(<\lambda, x>, \mathbb{C})$ if $\lambda>1$ and $x \in
  \mathbb{O}$. 
\end{enumerate}
We define $\mathbb{C}_r =\{x\in S: d(x, \mathbb{C})\le r\}$ and 
a subset $A \in \mathcal{B}$ is said to be {\em bounded away from $\mathbb{C}$} if
there exists $r>0$ such that $A \subset S\setminus\mathbb{C}_r$. Let $M_O(S)$ be the
class of Borel measures on $\mathbb{O}$ whose restrictions to $S\setminus
\mathbb{C}_r$ is finite for every $r>0$ and $\mathcal{C}_O$ be the class of
real-valued, non-negative, bounded and continuous functions on
$\mathbb{O}$ which vanishes on $\mathbb{C}_r$ for some $r>0$. The collection of finite Borel measures on the sub-$\sigma$-algebra $\mathcal{B}_{S \setminus \mathbb{C}_r}=\{A\cap (S \setminus \mathbb{C}_r): A\in \mathcal{B}_O\}$ is denoted by $M_b(S \setminus \mathbb{C}_r)$. The convergence $\mu_n \to \mu$ in $M_b(S\setminus \mathbb{C}_r)$ holds if and only if $\int f \dif \mu_n \to \int f \dif \mu$ for all continuous functions $f\in \mathcal{C}_O$ vanish on $\mathbb{C}_r$.  For
$\mu \in M_O$, let $\mu^{(r)}(\cdot) = \mu(\cdot \setminus \mathbb{C}_r) \in M_b (S\setminus \mathbb{C}_r)$.  According
to the Portmanteau theorem for the $M_O$-convergence (Theorem 2.1
in \cite{lindskog2014}), the following statements are
equivalent: 
\begin{enumerate}
\item $\mu_n \to \mu $ in $M_O$ as $n\to \infty$. 
\item $\int f\dif \mu_n \to \int f\dif \mu$ for each $f \in \mathcal{C}_O$. 
\item $\lim_{n\to \infty} \mu_n(A) = \mu(A)$ for all $A \in \mathcal{B}_O$ bounded away from $\mathbb{C}$ satisfying $\mu (\partial A) =0$.
\item $\mu_n^{(r)} \to \mu^{(r)}$ in $M_b(S \setminus \mathbb{C}_r)$ for all but
  at most countably many $r>0$.
\end{enumerate}

Recall from \cite{bingham1987} that a function $g:(0,\infty) \to (0, \infty)$
is regularly varying with index $\xi>0$, i.e., $g \in \mathcal{R}_{\xi}$, if $\lim_{x\to +\infty} g(\lambda x)/g(x)
= \lambda^{\xi}$ for all $\lambda>0$ and a sequence $(c_n)$ is regularly varying
with index $\xi>0$ if $\lim_{n\to \infty} c_{\lfloor \lambda n \rfloor}/c_n = \lambda^{\xi}$ for
all $\lambda>0$. Here $\lfloor \lambda n \rfloor$ denotes the integer part of $\lambda
n$. Regular variation based on the $M_O$-convergence is defined in
Theorem 3.1 of \cite{lindskog2014}, which is listed below. 
\begin{theorem}\label{thm:morv}
A measure $\nu \in M_O$ is {\em regularly varying} with index $\xi >0$ if
one of the following equivalent conditions is satisfied:
\begin{enumerate}
\item There exist a nonzero $\mu \in M_O$ and a regularly varying sequence
  $(c_n)$ of positive numbers such that $c_n \nu\big(<n, \cdot> \big) \to \mu(\cdot)$ in $M_O$
  as $n\to \infty$. 
\item There exist a nonzero $\mu \in M_O$ and a regularly varying function
  $c$ such that $c(t) \nu(<t, \cdot>) \to \mu(\cdot)$ in $M_O$ as $t\to \infty$.
\item There exist a nonzero $\mu \in M_O$ and a set $A \in \mathcal{B}_O$ satisfying
  $A \subset S\setminus \mathbb{C}_r$ for some $r>0$ such that $\nu(<t, A>)^{-1} \nu(<t, \cdot>) \to
  \mu(\cdot)$ in $M_O$ as $t\to \infty$. 
\end{enumerate}
Each statement 1-3 above implies
that the measure $\mu$ has the homogeneity property $\mu (<\lambda, A>) = \lambda^{-\xi} \mu(A)$
for some $\xi \ge 0$, all $A \in \mathcal{B}_O$ bounded away from $\mathbb{C}$ and $\lambda >0$.  
\end{theorem}

We consider the definitions of partial regular variation on three complete and separable 
spaces $(S,d,\mathbb{C})$, $(S^{(m)},d_m, \mathbb{C}^{(m)})$ and
$(S^{(\infty)}, d_{\infty}, \mathbb{C}^{(\infty)})$. The first space $(S,d,
\mathbb{C})$ is $S= \mathbb{R}^{1+l}$ and 
$\mathbb{C} = \mathbb{R}\times \mathbf{0}^{(l)}$ with $d$ as the Euclidean distance
in $\mathbb{R}^{1+l}$ and $\mathbf{0}^{(l)}$ as a zero vector in $\mathbb{R}^l$. Write
$<\lambda, (x, y_1,\ldots, y_l)>= (x, 
\lambda y_1,\ldots,\lambda y_l)=(x, 
\lambda\mathbf{y}) $ for $\lambda\in [0, +\infty)$ and $(x, \mathbf{y} )\in
\mathbb{R}^{1+l}$ as the partial scalar multiplication. The second space is $S^{(m)}= \mathbb{R}^{(l+1)m}$
for $m\in \mathbb{N}$ equipped with a metric
\begin{align*}
d_m  \big( (x_{1t}, \mathbf{y}_{1t})_{t=1,2,\ldots,m}, (x_{2t},
  \mathbf{y}_{2t})_{t=1,2,\ldots, m} \big) = \sum_{i=1}^m 2^{-i} \max
  \big( d\big((x_{1i}, \mathbf{y}_{1i}), (x_{2i}, \mathbf{y}_{2i}) \big), 1
  \big)\,. 
\end{align*}
We choose
$\mathbb{C}^{(m)} = \prod_{i=1}^m \mathbb{C}$. The partial scalar
multiplication for $S^{(m)}$ is given by
\begin{align*}
<\lambda, (x_{1t}, \mathbf{y}_{1t})_{t=1,2,\ldots,m} > = \big( (x_{j1}, \lambda \mathbf{y}_{j1}), \ldots,
(x_{jm}, \lambda \mathbf{y}_{jm}) \big)\,. 
\end{align*}
The third space $S^{(\infty)}$ is the space of sequences in $\mathbb{R}^{1+l}$,
$\{(x_t, \mathbf{y}_t) \in \mathbb{R}^{1+l}: t\in\mathbb{N}\}$ with a metric 
\begin{align*}
d_{\infty}\big( (x_{1t}, \mathbf{y}_{1t})_{t\in \mathbb{N}}, (x_{2t},
  \mathbf{y}_{2t})_{t\in \mathbb{N}} \big) = \sum_{i\in \mathbb{N}} 2^{-i}\max
  \big( d\big( (x_{1i}, \mathbf{y}_{1i}), (x_{2i}, \mathbf{y}_{2i}) \big), 1
  \big)\,, 
\end{align*}
and we choose $\mathbb{C}^{(\infty)} = \prod_{i=1}^{\infty} \mathbb{C}$. Similar to the partial
scalar multiplication in $S^{(m)}$, the partial scalar multiplication is
defined componentwisely, i.e., $<\lambda, (x_t, \mathbf{y}_t)_{t\in \mathbb{N}}>=
(x_t, \lambda \mathbf{y}_t)_{t \in \mathbb{N}}$.
\begin{remark}
 The formulation of $(S^{(m)},d_m)$ and $(S^{(\infty)},d_{\infty})$ is based on the
 arguments in Chapter 1, \cite{billingsley1999}, which ensures that
 $(S^{(m)}, d_m)$ and $(S^{(\infty)}, d_{\infty})$ are complete and
 separable. Moreover, the spaces $(S^{(m)},d_m)$ and
 $(S^{(\infty)},d_{\infty})$ are equipped with corresponding product topologies.  
\end{remark}

With the corresponding cones, regular variation based on the
$M_O$-convergence can be defined for the three spaces  $(S,d,\mathbb{C})$, $(S^{(m)},d_m, \mathbb{C}^{(m)})$ and
$(S^{(\infty)}, d_{\infty}, \mathbb{C}^{(\infty)})$ according to Theorem~\ref{thm:morv}. Assume that $V \in
\mathcal{R}_{-\xi}$ with $\xi>0$. We say that a random
vector $(X, \mathbf{Y} )$ taking values from $(S, d, \mathbb{C})$ is
{\em partially regularly varying} with the cone $\mathbb{C}$ and index $\xi>0$ if there
exists $\mu \in M_O(S)$ such that as $u\to +\infty$, 
\begin{align} \label{eq:prv00}
\frac{1}{V(u)} \p (u^{-1} (X, \mathbf{Y}) \in \cdot) \to \mu(\cdot)\,, \quad \text{
  in } M_O(S)\,. 
\end{align}
By replacing $(S,d, \mathbb{C})$ with $(S^{(m)}, d_m, \mathbb{C}^{(m)})$
and $(S^{(\infty)}, d_{\infty}, \mathbb{C}^{(\infty)}) $ in the definition of
partial regular variation, we have the following definitions.
\begin{enumerate}[(a)]
\item A random vector $(X_t, \mathbf{Y}_t)_{t=1,2,\ldots,m}$ taking values from
$(S^{(m)}, d_m, \mathbb{C}^{(m)})$ is said to be {\em
  partially regularly varying} with the cone $\mathbb{C}^{(m)}$ and
index $\xi>0$ such that there exists a measure $\mu^{(m)} \in M_O(S^{(m)})$ such
that  
\begin{align} \label{eq:prv01}
\frac{1}{V(u)} \p (u^{-1} (X_t, \mathbf{Y}_t)_{t=1,2,\ldots,m} \in \cdot) \to \mu^{(m)}(\cdot)\,, \quad \text{
  in } M_O(S^{(m)})\,. 
\end{align}
\item A random element $(X_t, \mathbf{Y}_t)_{t\in \mathbb{N}}$ taking values from
$(S^{(\infty)}, d_{\infty}, \mathbb{C}^{(\infty)})$ is {\em partially regularly varying} with the cone
$\mathbb{C}^{(\infty)}$ and index $\xi>0$ if there exists a measure
$\mu^{(\infty)} \in M_O(S^{(\infty)})$ such that 
\begin{align}\label{eq:prv02}
\frac{1}{V(u)} \p (u^{-1} (X_t, \mathbf{Y}_t)_{t\in \mathbb{N}} \in \cdot) \to
  \mu^{(\infty)}(\cdot)\,, \quad \text{ in } M_O(S^{(\infty)})\,.
\end{align}
\end{enumerate}
Similar to Theorem 4.1 in \cite{segers2017}, we can prove that the
two definitions are equivalent.

\begin{theorem} \label{thm:morvdef}
  Let $\xi>0$ and $ V \in \mathcal{R}_{-\xi}$.
 \begin{enumerate}
 \item A random element $(X_t, \mathbf{Y}_t)_{t\in \mathbb{N}}$ is partially regularly varying with the cone
  $\mathbb{C}^{(\infty)}$ and index $\xi$ if and only if for all $m\in \mathbb{N}$,
  the random vector $(X_t, \mathbf{Y}_t)_{t=1,\ldots,m}$ is partially
  regularly varying with the cone $\mathbb{C}^{(m)}$ and index
  $\xi >0$.
\item Moreover, if $(X_t, \mathbf{Y}_t)_{t\in \mathbb{N}}$ is strictly stationary, the condition $\mu^{(\infty)}(
\{(X_t, \mathbf{Y}_t)_{t\in \mathbb{N}}: (X_1, \mathbf{Y}_1) \notin \mathbb{C}\}) >0$ is equivalent to the
condition that $\mu^{(\infty)}$ is non-zero. 
\end{enumerate}
\end{theorem}
We will use \eqref{eq:prv01} for all $m\ge
1$ as the definition for partial regular variation in the rest of the
paper. 
\begin{remark}
Partial regular variation is defined for both a random vector and a time series. 
Both one-component regular variation in \cite{hitz2016} and multi-component regular
variation in \cite{segers2020} are defined for a random vector, not for a time series. One-component regular variation and
multi-component regular variation alter the classical definition of
multivariate regular variation and include multivariate regular variation as an example. Partial regular variation rules out multivariate regular variation. Particularly the marginal distribution of $X_t$ is not necessarily regularly varying when $(X_t, \mathbf{Y}_t)$ is partially regularly varying.  
\end{remark}

\subsection{Existence of TMES}

As shown in \eqref{eq:mes00}, TMES is defined as a limit of
conditional expectations whose existence requires extra
conditions besides partial regular variation.

{\bf Condition (E)}
\begin{enumerate}
\item A strictly stationary time series of random
vectors $(X_t, \mathbf{Y}_t)_{t\in \mathbb{N}}$ taking values from $\mathbb{R}^{1+l}$ for
some $l\ge 1$ is
partially regularly varying with the cone $\mathbb{C} = \mathbb{R} \times
\mathbf{0}^{(l)}$ and index $\xi>0.$ 
\item Let $A \in \mathcal{B}(\mathbb{R}^l)$ satisfy that $\mathbb{R} \times A$ is bounded
 away from $\mathbb{C}$, $\mu^{(1)}$-smooth and $0< \mu^{(1)}(\mathbb{R} \times A) < +\infty$. 
\item For an integer $h\ge 0,$ an integer $r\ge 1$ and a real number
  $\varepsilon>0$, we have 
  \begin{align}  \label{eq:ui}
\sup_{u>0} \E [|X_{h+1}|^{r+\varepsilon}\mid u^{-1}\mathbf{Y}_1 \in A] < \infty\,.
  \end{align}
\end{enumerate}

Suppose that condition {\bf (E)} holds. Partial regular variation
yields that for $B\in \mathcal{B}(\mathbb{R})$ and $h\ge 0$,
\begin{align*}
  & \p (X_{h+1} \in B \mid u^{-1} (X_1, \mathbf{Y}_1) \in \mathbb{R}\times A) \\
  & = \frac{\p((X_{h+1}, \mathbf{Y}_{h+1}) \in B \times \mathbb{R}^l,u^{-1} (X_1,
  \mathbf{Y}_1) \in \mathbb{R}\times A )}{ \p (u^{-1} (X_1, \mathbf{Y}_1) \in \mathbb{R}\times A)}\\  
  & = \frac{(V(u))^{-1} \p((X_{h+1}, \mathbf{Y}_{h+1}) \in B \times \mathbb{R}^l,u^{-1} (X_1,
  \mathbf{Y}_1) \in \mathbb{R}\times A )}{ (V(u))^{-1} \p (u^{-1} (X_1,
    \mathbf{Y}_1) \in \mathbb{R}\times A)}\\
  & \to \frac{\mu^{(h+1)} \big( (B\times \mathbb{R}^l) \times \mathbb{R}^{1+l} \times \cdots \times (\mathbb{R}\times A)
    \big)}{ \mu^{(1)}(\mathbb{R} \times A)}\,, \quad u\to +\infty\,.
\end{align*}
According to the arguments in the proof of Theorem~\ref{thm:morvdef} (see \cite{liu2025b}), $(\mu^{(l)})_{l \ge 1}$ form a consistent family of measures and thus,
the above limit is non-negative and less than or equal to $1$
for all $B\in \mathcal{B}(\mathbb{R})$. Moreover, by taking $B=\mathbb{R}$, 
\[\frac{\mu^{(h+1)} \big( (\mathbb{R} \times \mathbb{R}^l) \times \mathbb{R}^{1+l} \times \cdots \times (\mathbb{R}\times A)
    \big)}{ \mu^{(1)}(\mathbb{R} \times A)} = \frac{\mu^{(1)}(\mathbb{R} \times A)}{ \mu^{(1)}(\mathbb{R} \times
    A)} =1\,.\]
There exists a random variable $\widetilde{X}_h$ whose distribution is
determined by 
\begin{align*}
\p(\widetilde{X}_h \in \cdot) =  \frac{\mu^{(h+1)} \big( (\cdot\times \mathbb{R}^l) \times \mathbb{R}^{1+l} \times \cdots \times (\mathbb{R}\times A)
    \big)}{ \mu^{(1)}(\mathbb{R}\times A)}\,,
\end{align*}
such that
\begin{align}\label{eq:tailprocess}
  X_{h+1}\mid u^{-1} \mathbf{Y}_1 \in A \overset{w}{\to}
  \widetilde{X}_h\,,
\end{align}
where $\overset{w}{\to}$ stands for weak
convergence. To ensure the convergence of moments, we need
a uniform integrability condition \eqref{eq:ui}.

\begin{theorem} \label{thm:existencedeltah}
Given that condition {\bf (E)} holds, the limit 
\begin{align*}
\delta(h) = \lim_{u\to \infty} \E\big[ X_{h+1} \mid u^{-1} \mathbf{Y}_1 \in A
  \big] \,,
\end{align*}
exists for $h\ge 0$.
\end{theorem}
\begin{proof}
Recall that under condition {\bf (E)}, the limit
\begin{align*}
X_{h+1} \mid u^{-1} \mathbf{Y}_1 \in A \overset{w}{\to } \widetilde{X}_h\,,
  \quad u\to +\infty\,.
\end{align*}
According to Skorohod's theorem (see e.g. Theorem 25.6 in \cite{billingsley1995}),
there exist random variables $Z_u$ and $Z$ on a common probability
space $(\Omega, \mathcal{F}, Q)$ such that $Z_u$ has the same distribution as $X_{h+1}
\mid u^{-1} \mathbf{Y}_1 \in A$, $Z$ has the same distribution as $\widetilde{X}_h$
and $Z_u \to Z$ as $u \to \infty$ holds with probability $1$. Denote
$\E_{Q}$ as the expectation under the common probability measure $Q$.

Due to Theorem 25.12 and the corollary on the same page in
\cite{billingsley1995}, the inequality \eqref{eq:ui} ensures that 
\begin{align*}
\E_Q[|Z|^r] < \infty\,, \quad \E_Q [|Z_u|^r] \to \E_Q [|Z|^r]\,, \quad u\to \infty\,, r\ge 1\,.
\end{align*}
This completes the proof. 
\end{proof}

\begin{remark}
Condition {\bf (E)} is sufficient for the existence of $\delta(h)$, but it
is not necessary. As special examples of TMES, the existence of the extremogram \eqref{eq:extremogram} and the causal tail coefficient \eqref{eq:causaltail} is guaranteed by regular variation. For the example in Section~\ref{sec:examplemax}, regular variation implies \eqref{eq:ui}. We find the formula for the TMES of the example in Section~\ref{sec:examplecopula} and thus the sufficient condition for the existence of the TMES becomes trivial. 
\end{remark}
\begin{remark}
The quantity to measure {\em tail adversarial stability} in a regularly varying time series is proposed in \cite{zhang2021, zhang2022} and can be understood as an upper bound of the TMES with properly chosen $(X_t, \mathbf{Y}_t)$. Consequently, condition {\bf (E)} is not sufficient for its existence and it will not be in the scope of our studies.  
\end{remark}

\subsection{Examples} \label{sec:example}

We are interested in two classes of models. The first class of
models consists of regularly varying time series of random vectors that are
classical objects studied in extreme value theory; see
~e.g. \cite{davis2013measures}. We will verify condition {\bf (E)} for
this class of models. The models in the second class are
constructed through a regularly varying time series and a copula. We will provide an alternative to
\eqref{eq:ui} for this class of models.

\subsubsection{Example: pieces of a regularly varying random
  field} \label{sec:examplemax}
If all the finite-dimensional distributions of a real-valued random field
$(Z_{i_1,i_2})_{(i_1, i_2) \in \mathbb{Z}^2}$ are multivariate regularly varying
with index $\xi>3$, the random field $(Z_{i_1, i_2})$ is regularly
varying with index $\xi>3$; see \cite{Wu2020} for more details. We
assume that all the $Z_{i_1, i_2}$'s take non-negative values. Let
$X_t= {Z}_{0,t}$ and $Y_t = {Z}_{1,t}$ for $t\in 
\mathbb{Z}$. The random vector $(X_t,
Y_t, \ldots, X_{t-h}, Y_{t-h})$ is multivariately regularly varying with
the tail index $\xi>3$. The marginal distributions of $X_t$ and
$Y_t$ share the same regular varying distribution. Thus it is easy to verify that the first two conditions in Condition {\bf
  (E)} are satisfied. 

To prove 
that \eqref{eq:ui} holds for $(X_t, Y_t)$, it is sufficient to show that for $0< r
< \xi-1 $ with some $\varepsilon\in (0,0.5)$,
\begin{align} \label{eq:mmamoment}
\sup_{u>0}\E[X_t^{r+\varepsilon} \mid Y_{t-h}>u] < +\infty\,.
\end{align}
For all $u>0$, an application of integration by parts yields
that
\begin{align*}
& \E [X_t^{r+ \varepsilon} \mid Y_{t-h} > u] = - \int_0^{\infty} x^{r+ \varepsilon} \dif \p (X_t > x\mid
                                Y_{t-h}>u)\\
  & \quad = (r+ \varepsilon)\int_0^{\infty} x^{r+\varepsilon -1} \p(X_t\ge x\mid Y_{t-h}>u)\dif x + \big[ x^{r+\varepsilon} \p(X_t\ge x \mid
    Y_{t-h}>u)\big]\Big|_0^{\infty}\\
& \quad = (r+ \varepsilon) \int_0^{\infty} x^{r + \varepsilon -1} \p(X_t> x\mid Y_{t-h}>u)\dif x\,.  
\end{align*}
According to Potter's bound (see Theorem 1.5.6 in \cite{bingham1987}), for $C_0 >1$ and
$\varepsilon >0$,
there exists $x_0>0$ such that for $x>1$ and $u>x_0$, we have
\begin{align*}
\p (X_t>xu \mid Y_{t-h} >u) \le \frac{\p(X_t>xu)}{ \p(Y_{t-h} >u)} 
  =\frac{\p(X_t>xu)}{\p(X_t>u)} \frac{\p(X_t>u)}{\p(Y_{t-h}>u)}\le C_0 x^{-\xi + \varepsilon}\,.
\end{align*}
Since $X_t$ is regularly varying, we have
$\E[X_t^{\xi-1}]<\infty$, which implies that $\sup_{u\in (0,x_0)}
\E[X_t^{r+\varepsilon}\mid Y_{t-h} >u]<\infty$. 
\begin{align*}
  & \sup_{u>x_0} \E [X_t^{r+\varepsilon} \mid Y_{t-h} > u]\\
  & \le \sup_{u>x_0} u^{-r -\varepsilon -1} \int_0^1 x^{r+\varepsilon} \dif x
    +u^{-r -\varepsilon } \int_1^{\infty} x^{r+ \varepsilon-1} \p(X_t> ux \mid Y_{t-h}>u) \dif x \\
  & \le
  x_0^{-r-\varepsilon-1}+ C_0 x_0^{-r-\varepsilon} \int_1^{\infty} x^{r-\xi+ 2\varepsilon-1}\dif x  \le x_0^{-r-\varepsilon-1}+ C_0 x_0^{-r-\varepsilon} \int_1^{\infty} x^{-2+ 2\varepsilon}\dif x\\
  & <\infty\,.  
\end{align*}

\subsubsection{Example: copula-based time series} \label{sec:examplecopula}

We consider a bivariate time series $(X_t, Y_t)_{t \in \mathbb{N}}$ with
continuous marginal distributions. The strictly stationary time series $(Y_t)$ is
regularly varying with tail index $\xi>0$. The vector $(F_{X_t}(X_t),
F_{Y_t}(Y_t))$ forms a copula $C$ where $F_{X_t}$ and $F_{Y_t}$ are the
distribution functions of $X$ and $Y$ respectively. The copula $C$ is
independent of $(Y_t)$. Thus,
$(X_t, Y_t)$ is partially regularly varying with index $\xi> 0$ and the
cone $\mathbb{C}=\mathbb{R} \times \{0\}$. The extremogram of $(Y_t)$ exists and is
denoted by $\rho_Y(h) = \rho(h) = \lim_{u \to \infty} \p(Y_{1+h} >u \mid Y_1
> u)$ for $h\ge 0$. Trivially $\rho(0)=1$. We are interested in the TMES of $(X_t,
Y_t)$ given by  
\[ \delta (h)=\lim_{u \rightarrow \infty}\E [X_{t} \mid Y_{t-h} > u] \,. \]

\begin{proposition} \label{prop:latent}
 Let $(X_t, Y_t)_{t\in \mathbb{N}}$ be the bivariate time series defined
 above. Suppose that $\delta(0)$ exists and $\E[|X|]<+\infty$. Then, 
\begin{align} \label{eq:latent}
\delta(h) = (1-\rho(h))\E [X_1]+\rho(h)\delta (0)\,.
\end{align} 
\end{proposition}
\begin{proof}[Proof of Proposition~\ref{prop:latent}]
Let $V(u)= \p(Y_t >u)$. 
Since $X_t$ is independent of $Y_{t-h}$ for $h>0$ conditionally on $Y_t$,
\begin{align*}
\delta (h) =&\lim_{u \rightarrow \infty}\int_{-\infty}^{\infty} x \dif \p (X_{t} \leq x \mid Y_{t-h} > u)\\
=&\lim_{u \rightarrow \infty}\frac{1}{V(u)} \int_{-\infty}^{\infty} x \dif \p
   (X_{t} \leq x, Y_{t-h} > u)\\
=&\lim_{u \rightarrow \infty}\frac{1}{V(u)} \int_{-\infty}^{\infty} x \dif \big( \p (X_{t} \leq
   x\mid Y_t \le u)\p (Y_{t-h} >u \mid Y_{t}\le u) \p (Y_{t}\le u) \big)\\
& + \lim_{u \rightarrow \infty}  \frac{1}{V(u)}\int_{-\infty}^{\infty} x \dif\big(\p (X_{t} \leq x\mid
  Y_{t} > u)\p (Y_{t-h} > u \mid Y_{t} > u) \p (Y_{t}> u) \big)\\
=&Q_1 + Q_2\,.
\end{align*}
Trivially, $\p (Y_{t}\le u) \to 1$ as $u\to \infty$. Moreover, 
\begin{align*}
  \lim_{u \rightarrow \infty} (V(u))^{-1} \p (Y_{t-h} > u \mid Y_{t} \le u)
  & = \lim_{u \rightarrow \infty}\frac{ (V(u))^{-1} \p (Y_{t-h} > u )-\p(Y_{t-h} >
    u \mid Y_t >u )\big)}{\p(Y_t \le u)}\\
 & = 1 -\rho(h)\,, 
\end{align*}
and 
\begin{align*}
\lim_{u \to \infty} \int_{-\infty}^{\infty} x \, d\, \p(X_t \le x, Y_t \le u) =\E[X]\,.
\end{align*}
Thus, we have $Q_1 = (1- \rho(h))\E[X]$. 
\begin{align*} 
Q_2 =&\lim_{u \rightarrow \infty} (V(u))^{-1} \p(Y_t > u) \p(Y_{t-h} > u \mid Y_{t}>
       u) \int_{-\infty}^{\infty} x \dif \p (X_{t} \leq x \mid Y_{t} >u)\\
=&\rho(h)\delta(0)\,.
\end{align*}
This completes the proof. 
\end{proof}

As $\rho(h)$ decreases,
the difference between $\delta(h)$ and $\E[X_1]$ becomes smaller, that is,
the influence of the systemic risk $\{Y_{t-h}>u\}$ on $X_t$
decays. We further remark that under conditions of Proposition \ref{prop:latent}, it holds that 
$$\delta_0(h)=\rho(h)\delta_0(0),$$ indicating that the $\delta_0(h)$ is proportional to $\rho(h).$

\section{Empirical TMES and its asymptotic properties} \label{sec:empiricaltmes}
The estimation of TMES is naturally applicable across finance, insurance, risk management, and meteorology, which accounts for time lags in assessing systemic risk. In this section, we propose empirical TMES and study its asymptotic properties. Concretely, we fix an integer $h \ge 0$ and a set $A\in \mathcal{B}(\mathbb{R}^l)$. We are interested in the estimation of
$\delta(h)$. According to the definition of partial regular variation, we can
replace $u$ in \eqref{eq:mes00} with an increasing sequence $(a_n)$
satisfying $n\p(a_n^{-1} \mathbf{Y}_t \in A) \to 1$ as $n\to \infty$. Let
$(m_n)$ be a sequence of positive integers such that $m_n \to +\infty$ and
$m_n/n \to 0$ as $n\to \infty$. Suppose that we obtain the observations $(X_t,
\mathbf{Y}_t)_{t=-h+1,\ldots, 0, 1, \ldots, n}$. The {\em empirical TMES} at
lag $h\ge 0$ is given by
\begin{align} \label{eq:tmes}
\widehat{\delta}(h) = \frac{m_n}{n} \sum_{t=1}^n X_t \1 (a_{m_n}^{-1}
  \mathbf{Y}_{t-h} \in A)= \frac{m_n}{n} \sum_{t=1}^n X_t I_{t-h}\,, 
\end{align}
where $I_t = \1 (a_{m_n}^{-1} \mathbf{Y}_t \in A)$. For ease of
notation, we suppress $\widehat{\delta}(h)$'s and $I_t$'s dependence on
$n$. Trivially, we have 
\begin{align*}
\E [\widehat{\delta}(h)] = (m_n \p (a_{m_n}^{-1}
\mathbf{Y}_{t-h} \in A)) \E[ X_{h+1} \mid a_{m_n}^{-1} \mathbf{Y}_1 \in A]
\to \delta(h)\,, \quad n\to \infty\,.
\end{align*}  

For the proofs of the central limit theorem for $\widehat{\delta}(h)$, we
will use the $\alpha$-mixing conditions, which are also used in proving the
asymptotic properties of the extremogram; see e.g. \cite{davis2009,
  davis2013measures}. Let $\mathcal{F}_l^k = \sigma\big( (X_t, \mathbf{Y}_t)_{l\le t
  \le k} \big)$ be a $\sigma$-field generated by $(X_t, \mathbf{Y}_t)_{l\le
  t \le k}$. The {\em $\alpha$-mixing rate} $\alpha(\cdot)$ is given by 
  \begin{align*}
  \alpha(k) = \sup_{B \in \mathcal{F}_{-\infty}^0\,, D \in \mathcal{F}_k^{+\infty}} \big| \p(B \cap D) -
    \p(B) \p(D) \big|\,,\quad k\ge 0\,.
  \end{align*}
{\bf Condition (M)} 

The $\mathbb{R}^{1+l}$-valued strictly stationary time series $(X_t,
\mathbf{Y}_t)$ with generic element $(X, \mathbf{Y})$ satisfies condition {\bf (E)} with $r\ge 4$
  in \eqref{eq:ui}. Let $(m_n)$ and $(r_n)$ be two sequences of
positive integers such 
that $r_n \to \infty$, $m_n \to \infty$, $r_n^2/m_n \to 0$ and $m_n^4r_n^3/n \to 0$ as $n\to
\infty$.
\begin{enumerate}
\item The moments $\E[X] =0$, $\E[X^4] < \infty$.
\item For all $\varepsilon>0$, 
\begin{align}
\label{eq:anticlustering}
\lim_{l\to \infty} \limsup_{n\to \infty} m_n \sum_{s = l+1}^{r_n}\p\big( \| \mathbf{Y}_1 \| >
  \varepsilon a_{m_n}\,, \| \mathbf{Y}_{1+s} \| > \varepsilon a_{m_n} \big) =0\,,
\end{align}
where $\|\cdot \|$ stands for the Euclidean norm in $\mathbb{R}^l$. 
\item The mixing rate $\alpha(h)$ of $(X_t, \mathbf{Y}_t)_{t\in \mathbb{Z}}$ satisfies the conditions 
  \begin{align}
  \label{eq:mixing}
\lim_{n\to \infty} m_n^7\sum_{h=r_n +1}^{\infty}\alpha (h) =0\,, \quad \lim_{n\to \infty} m_n
    \alpha (r_n) =0\,.  
  \end{align}
\end{enumerate}
There exists $\varepsilon>0$ such that $ A \subset \{\mathbf{y}\in \mathbb{R}^{l}: \|\mathbf{y} \| > \varepsilon\}
$. This implies that
\begin{align*}
 m_n \p(\| \mathbf{Y}_1\| > \varepsilon a_{m_n}, \| \mathbf{Y}_{1+s} \| > \varepsilon
  a_{m_n}) \ge m_n \p\big( a_{m_n}^{-1} (\mathbf{Y}_1,
  \mathbf{Y}_{1+s}) \in A \times A \big)\,,
\end{align*}
and consequently,
\begin{align*}
\lim_{l \to \infty } \limsup_{n\to \infty} m_n \sum_{s=l+1}^{r_n} \p
  (a_{m_n}^{-1}(\mathbf{Y}_1, \mathbf{Y}_{1+s}) \in A\times A) =0\,.
\end{align*}
The condition \eqref{eq:mixing} on the mixing rate is satisfied if
$\alpha(h)$ decays exponentially fast to zero as $h\to \infty$. Examples with
this property include regularly varying linear processes, GARCH processes,
max-moving averages, Brown-Resnick processes; see
e.g. \cite{davis2013measures, cho2016,damek2023whittle}. Given that $r\ge 4$ in \eqref{eq:ui}, we have the limit
$\lim_{n\to \infty} \E[X_{h+1}^4 \mid a_{m_n}^{-1} \mathbf{Y}_1 \in A] < \infty$
exists and
\begin{align*}
\E[ |X_h X_{h+l}| \mid a_{m_n}^{-1} ( \mathbf{Y}_{0}, \mathbf{Y}_l) \in A \times A] &\le \frac{2\p(a_{m_n}^{-1}\mathbf{Y}_0 \in A)}{\p(a_{m_n}^{-1} ( \mathbf{Y}_{0}, \mathbf{Y}_l) \in
  A \times A)}\E[ X_h^2  \mid a_{m_n}^{-1}  \mathbf{Y}_{0} \in A] \\
& < \infty\,.
\end{align*}
An application of the dominated convergence theorem ensures that 
\begin{align*}
\tau_h(s) = \lim_{n\to \infty} \E[X_{h+1} X_{h+1+s} \mid a_{m_n}^{-1} ( \mathbf{Y}_{1},
\mathbf{Y}_{s+1}) \in A \times A]\,, \quad s\ge 0\,, 
\end{align*}
exist for all $h,s\ge 0$.

\subsection{Asymptotic properties of empirical TMES}
In this section, we will calculate the asymptotic mean and variance of empirical TMES and will show a pre-asymptotic central limit theorem for
$\widehat{\delta}(h)$ under condition {\bf (M)}. Recall that the extremogram $\rho(s)$ of the sequence $(\mathbf{Y}_t)$ for $s\ge 0$ is given by 
\begin{align*}
    \rho(s) = \lim_{n\to \infty} \p(a_n^{-1} \mathbf{Y}_s \in A\mid a_n^{-1} \mathbf{Y}_0 \in A)\,,
\end{align*}
which exists according to the definition of partial regular variation.

\begin{proposition} \label{prop:consistency}
Assume that $(X_t, \mathbf{Y}_t)_{t\in \mathbb{N}}$ satisfies condition {\bf (M)}.
If moreover
\begin{align}\label{eq:infinitesum}
  \sum_{s\ge 0} \big(|\tau_h(s)| + \rho(s) \big)<+\infty \,,
\end{align}
then as $n\to \infty$,
\begin{align}
\label{eq:consistency1}
\E\big[ \widehat{\delta}(h) \big] & \to  \delta(h)\,,\\  \label{eq:consistency2}
\frac{n}{m_n}\var \big( \widehat{\delta}(h)\big)
 & \to   \Big( \tau_h(0) + 2 \sum _{s=1}^{\infty} \rho(s) \tau_h(s) \Big) =: \sigma_h^2\,.
\end{align}
Moreover, \eqref{eq:consistency1} and \eqref{eq:consistency2} imply that  
\begin{align}\label{eq:consistency3}
\widehat{\delta}(h) \overset{\p}{\to } \delta(h)\,.
\end{align}  
where $\overset{\p}{\to }$ stands for convergence in probability.
\end{proposition}
Proposition~\ref{prop:consistency} shows that $\widehat{\delta}(h)$ is a
consistent estimator of $\delta(h)$ under condition {\bf (M)}. We will show a pre-asymptotic central limit theorem for
$\widehat{\delta}(h)$. 
\begin{theorem}\label{thm:clt}
 Assume that $(X_t, \mathbf{Y}_t)_{t\in \mathbb{N}}$ satisfies condition {\bf
   (M)} and \eqref{eq:infinitesum} holds. Suppose that there exists a
 $ q \geq 6$  such that $\E [|X|^q] < +\infty $ and $n/m_n^{q-1} \to 0$
 as $n \to \infty$. The limit 
  \begin{align}
  \label{eq:clt}
  \sqrt{\frac{n}{m_n}} \big(\widehat{\delta}(h) - \E[\widehat{\delta}(h)]
    \big) \overset{d}{\to }
    Z_h \,, \quad n\to \infty\,,
  \end{align}
 holds where $\overset{d}{\to }$ stands for convergence in distribution and $Z_h$ is a normal random variable with mean $0$ and variance
  $\sigma^2_h$. 
\end{theorem}
The reason why we have a pre-asymptotic central limit theorem for
$\widehat{\delta}(h)$ is because the difference between
$\E[\widehat{\delta}(h)]$ and $\delta(h)$ is not negligible in
general. The proof of Theorem~\ref{thm:clt} follows the framework in
the proof of Theorem 3.1, \cite{davis2009} by using a
big-block-small-block argument. According to
Lemma~\ref{lem:Xtsmallmn}, we will only need to consider the case where
$|X_t|\le m_n$.  

\begin{lemma}\label{lem:Xtsmallmn}
 Assume that $(X_t, \mathbf{Y}_t)$ satisfies condition {\bf (M)}. If
 there exists an integer $ q \geq 2$ such that $\E [|X_t|^q]< \infty$ and
 $n/m_n^{q-1}\to 0$ as $n \to \infty$, then 
\begin{align}
\label{eq:largeneg}
\sqrt{\frac{m_n}{n}} \sum_{t=1}^n |X_t| I_{t-h} \1 (|X_t|>m_n)
  \overset{\p}{\to} 0\,.
\end{align}
\end{lemma}
\begin{proof}
Notice that 
\begin{align*}
& \sqrt{\frac{m_n}{n}} \sum_{t=1}^n \E \big[ |X_t| I_{t-h}
  \1(|X_t|> m_n) \big]\\
& = \sqrt{\frac{m_n}{n}} \sum_{t=1}^n \Big( \cov\big( |X_t|,
  I_{t-h}\1(|X_t|>m_n) \big) + \E[|X_t|] \p(a_{m_n}^{-1}\mathbf{Y}_{t-h} \in A,
  |X_t|>m_n ) \Big)\\ 
& \le \sqrt{\frac{m_n}{n}} \sum_{t=1}^n \sqrt{\var(|X_t|) \p(|X_t|>m_n)}
  + \E[|X_t|] \p( |X_t|>m_n )\,. 
\end{align*}
In the last step, the Cauchy-Schwarz inequality is used. Since
$\E[|X_t|^q] < \infty$, an application of Markov's inequality yields that
$\p(|X_t|>m_n) \le m_n^{-q} \E [|X_t|^q] <\infty$. Since
$n/m_n^{q-1} \to 0 $ as $n\to \infty$, we have \begin{align*}
  \sqrt{\frac{m_n}{n}} \sum_{t=h+1}^n \E
\big[ |X_t| I_{t-h} \1(|X_t|>m_n) \big] = O\big( \sqrt{n/m^{q-1}}
  \big) \to 0\,, n\to \infty\,,
\end{align*}
which implies that \eqref{eq:largeneg} holds. 
\end{proof}
\begin{remark}
To ensure the conditions in Theorem~\ref{thm:clt} are satisfied, we
need to choose carefully the values of $(r_n)$ and $(m_n)$. One possible
choice is to choose $r_n$ as the integer part of $c \log n$ for a
constant $c>0$, i.e. $r_n =
\lfloor c \log n\rfloor$ and $m_n= \lfloor n^{\beta}\rfloor$ with $\beta \in ((q-1)^{-1}, 0.25)$. It is
easy to verify that the limits $r_n^2/m_n \to 0$, $m_n^4 r_n^3/n \to 0$ and
$n/m_n^{q-1} \to 0$ hold as $n \to \infty$.
\end{remark}

\section{Bootstrapping TMES}
\label{sec:bootstrap}

Since the right-hand side of \eqref{eq:consistency2} is in general
unknown, we will provide a non-parametric method to construct
confidence intervals of TMES based on a bootstrap algorithm.  

We will use the stationary bootstrap algorithm proposed in
\cite{politis1994} and modified in \cite{davis2012142}. Fix an integer
$h\ge 0$ and write $Z_t = X_t I_{t-h} - \E[X_t I_{t-h}]$, $t\ge
1$. Suppose that we observe the sample $(Z_t)_{t=1,2,\ldots,n}$ for a
positive integer $n$ and we will generate a bootstrapped sample of $(Z_t)_{t=1,\ldots,n}$, $(Z_t^{\star})_{t=1,\ldots, n}= (Z_{t^{\star}})_{t =1,2,\ldots,n}$, which can be written as 
\begin{align} \label{eq:bootstrapsequence}
\left(Z_{H_1}, \ldots, Z_{H_1 + L_1 -1}, Z_{H_2},
  \ldots, Z_{H_2+ L_2 - 1},
  \ldots, Z_{H_N}, \ldots, Z_{H_N +L_N -1}, \ldots \right)
\end{align}
where $(H_i)$ and $(L_i)$ are two sequences of positive integers independent of  $(Z_t)$. The
segment $\left(Z_{H_i},\ldots, Z_{H_i +L_i -1}\right)$ is the $i$-th block of the
bootstrapped sample. The starting positions of these blocks,
denoted by $(H_i)$, form an iid sequence of random variables uniformly
distributed on the index set $\{1, \ldots, n\}$. The lengths of these
blocks, denoted by $(L_i)$, form an iid sequence of geometrically distributed
random variables with parameter $\theta= \theta_n \in (0,1)$, i.e.,
$\p(L_i=l) = \theta(1 -\theta)^{l-1}$ with $\theta=\theta_n \to 0$ as $n\to \infty$ and
$l=1,2,\ldots$. The number of blocks in the bootstrapped sample is given
by $N=N_n = \inf \{l\ge 1: \sum_{i=1}^l L_i \ge n\}$. If
any element $Z_t$ in \eqref{eq:bootstrapsequence} has an index $t>n$,
we take $Z_t \equiv Z_{t \mod n}$. In what follows, the probability measure generated by the bootstrap
procedure is denoted by $\p^{\star}$, i.e., $\p^{\star} (\cdot) = \p(\cdot \mid
(Z_t))$. The corresponding expectation, variance and covariance are
denoted by $\E^{\star}$, $\var^{\star}$ and $\cov^{\star}$, respectively.

The {\em bootstrapped (empirical) TMES} is given by
\begin{align*}
\delta^{\star} (h) = \frac{m_n}{n} \sum_{t=1}^n Z_t^{\star}\,.
\end{align*}
Due to the stationarity of $(Z_t^{\star})$, we have
\begin{align*}
\E^{\star} [\delta^{\star} (h)] = \frac{m_n}{n} \sum_{t=1}^n
  \E^{\star}[Z_t^{\star} ] =  m_n \E^{\star}[Z_1^{\star}] = m_n \Big(n^{-1} \sum_{t=1}^n
  Z_t\Big)= m_n \overline{Z}_n\,,
\end{align*}
where $\overline{Z}_n = n^{-1} \sum_{t=1}^n Z_t$. 

The following theorem is another main result of the paper, which
proves the consistency of bootstrapped TMES and the central limit
theorem for the bootstrapped TMES. 
\begin{theorem} \label{thm:bootclt}
  Assume that the vector time series $(X_t, \mathbf{Y}_t)_{t\in \mathbb{N}}$ satisfies
  condition {\bf (M)}. If moreover for $q\ge 6$, 
  \begin{align}
  \label{bootcondition2}
 E[|X_t|^q] < \infty\,, \quad n/m_n^{q-1} \to 0\,,
  \end{align}
  and the parameter $\theta$ in
  the stationary bootstrap algorithm satisfies that as $n\to \infty$,
  \begin{align}
\label{bootcondition3}
\theta\to 0\,, \quad n\theta^2 \to \infty\,, 
\end{align}
  we have
  \begin{align}
  \label{eq:bootconsist1}
    \E^{\star} \big[ \delta^{\star}(h) \big] \overset{\p}{\to} 0\,, \quad n\to \infty\,,\\ 
  \label{eq:bootconsist2}
 \limsup_{n\to \infty}\E[ \var^{\star}\big( (n/m_n)^{1/2} \delta^{\star}(h)\big)]
      =\sigma_h^2 \,, \quad n\to \infty\,, 
  \end{align}
where $\sigma_h^2$ is given in \eqref{eq:consistency2}. Moreover, the central limit
theorem holds 
\begin{align}
\label{eq:bootclt}
\sup_x \big| \p^{\star} \big( (n/m_n)^{1/2}  \big( \delta^{\star}(h) - m_n \overline{Z}_n
  \big) \le x  \big) - \Phi(x/\sigma_h)\big| \overset{\p }{\to } 0\,, \quad n\to \infty\,,
\end{align}
where $\Phi$ is the cumulative distribution function of a standard
normal distribution. 
\end{theorem}
Theorem~\ref{thm:bootclt} tells that $\delta^{\star}(h)$ shares the same
asymptotic distribution as $\widehat{\delta}(h)$.

\section{Simulation studies}\label{sec:simulation}

In this section, to verify the asymptotic properties of empirical TMES ($\widehat{\delta}(h)$) and bootstrapped TMES ($\delta^{\star}(h)$) for two classes of models introduced in Section~\ref{sec:example}, we conducted a simulation study. For each model, we simulated sample paths of length n=2,000 to estimate the empirical TMES. Using the stationary bootstrap algorithm with N=300 repetitions, we generated bootstrapped sequences and corresponding bootstrapped TMES values, forming a bootstrapped distribution for $\widehat{\delta}(h)$. Additionally, we simulated N=300 sample paths of each model independently, calculating the empirical TMES for each to form a simulated distribution. According to Theorem~\ref{thm:clt} and Theorem~\ref{thm:bootclt}, both distributions should converge asymptotically and from these distributions, we constructed $90\%$-confidence intervals. To assess their normality, we compared them by using QQ-plots against the standard normal distribution. This methodology allowed us to rigorously examine the convergence properties of our estimators, ensuring that both empirical and bootstrapped TMES values exhibit the same asymptotic behavior. Through this section, we choose $m_n=20$ in estimating $\widehat{\delta}(h)$ and the
parameter $\theta$ for stationary bootstrap is chosen as $\theta = 1/10$. 

\subsection{Parts of a regularly varying random field}

We will use the {\em max-moving averages} model as an example of the
regularly varying random field. Let $(Z_{i_1, i_2})_{(i_1, i_2) \in
  \mathbb{Z}^2}$ be iid Fr\'{e}chet distributed 
with index $\xi > 0$. Define the max-moving averages 
\begin{align} \label{eq:mma}
\widetilde{Z}_{i_1, i_2} = \max_{(s_1, s_2) \in \mathbb{Z}^2} w(s_1, s_2)
  Z_{i_1-s_1, i_2 -s_2}\,, 
\end{align}
where $w(s_1, s_2) = \phi^{|s_1| + |s_2|} \1 \big( |s_1| + |s_2| \le L
\big)$. In Figure~\ref{fig:mmasamplepath}, we present a simulation of
$(X_t^{(1)}, Y_t^{(1)})_{t=1,\ldots,n} = (\widetilde{Z}_{0,t},
\widetilde{Z}_{1,t})_{t=1,\ldots,n}$ with $\phi=0.8$ and $L= +\infty$. The
empirical TMES for $(X_t^{(1)}, Y_t^{(1)})$ is given by 
\begin{align*}
\widehat{\delta}(h) = \frac{m_n}{n} \sum_{t=1}^{n-h} X_{t+h}^{(1)}
  \1(a_{m_n}^{-1} Y_t^{(1)} > 1)\,, \quad h=0,1,\ldots, 9\,.
\end{align*}
As shown in Section~\ref{sec:examplemax}, we have
\begin{align*}
  \delta(h) & = \lim_{n\to \infty} \E[X_{1+h} \mid Y_1 >a_{m_n} ]\\ &=  \lim_{n\to \infty}
  \int_0^{\infty} \p(X_{1+h} \le x \mid Y_1 >a_{m_n}) \dif  x\\ 
  &=\lim_{n\to \infty} \int_0^{\infty} 1- \frac{\p(X_{1+h} \le x) - \p (X_{1+h}\le x
  , Y_1 \le a_{m_n})}{1- \p(Y_1 \le a_{m_n})}\dif  x,
\end{align*}
where for $x,y \ge 0$,
\begin{align*}
  \p (X_{1+h}\le x) \
&= \p(Y_1 \le x) = \exp \Big(- x^{-\xi} \Big( 1 + \sum_{j>0} 4j \phi^{\xi j}
  \Big) \Big)\,, \quad x>0\,,\\ 
  \p(X_t \le x, Y_{t-h} \le y)
& = \exp \Big\{ - \sum_{i_1=-\infty}^{\infty} \sum_{i_2 = -\infty}^{\infty} \max
  \Big(x^{-\xi} \phi^{\xi(|i_1|+|i_2|)}, y^{-\xi} \phi^{\xi(|i_1-1| + |i_2 +h|)}
  \Big) \Big\}\,,
\end{align*}
according to the arguments in \cite{cho2016}. By letting $n$ take a
sufficient large value and replacing the integral with its discretized
version, we obtain the theoretical values of $\delta(h)$ at lags
$h=0,1,\ldots,9$.

The empirical TMES $\widehat{\delta}(h)$ is close to  $\delta(h)$
in Figure~\ref{fig:mmacompare}. Moreover, the theoretical values of $\delta(h)$ stay
in both the bootstrapped and simulated $90\%$-confidence
intervals. This indicates that $\widehat{\delta}(h)$ is consistent. In
Figure~\ref{fig:mmaqq}, we present the QQ-plots of the simulated
distributions for $\widehat{\delta}(0)$ and $\widehat{\delta}(3)$ and the
bootstrapped distribution of $\delta^{\star}(0)$ and $\delta^{\star}(3)$ against the standard normal
distribution. This is an evidence of the asymptotic normality of
$\widehat{\delta}(h)$ and $\delta^{\star}(h)$.

\begin{figure}[htbp]
\centering  
\includegraphics[scale=0.53]{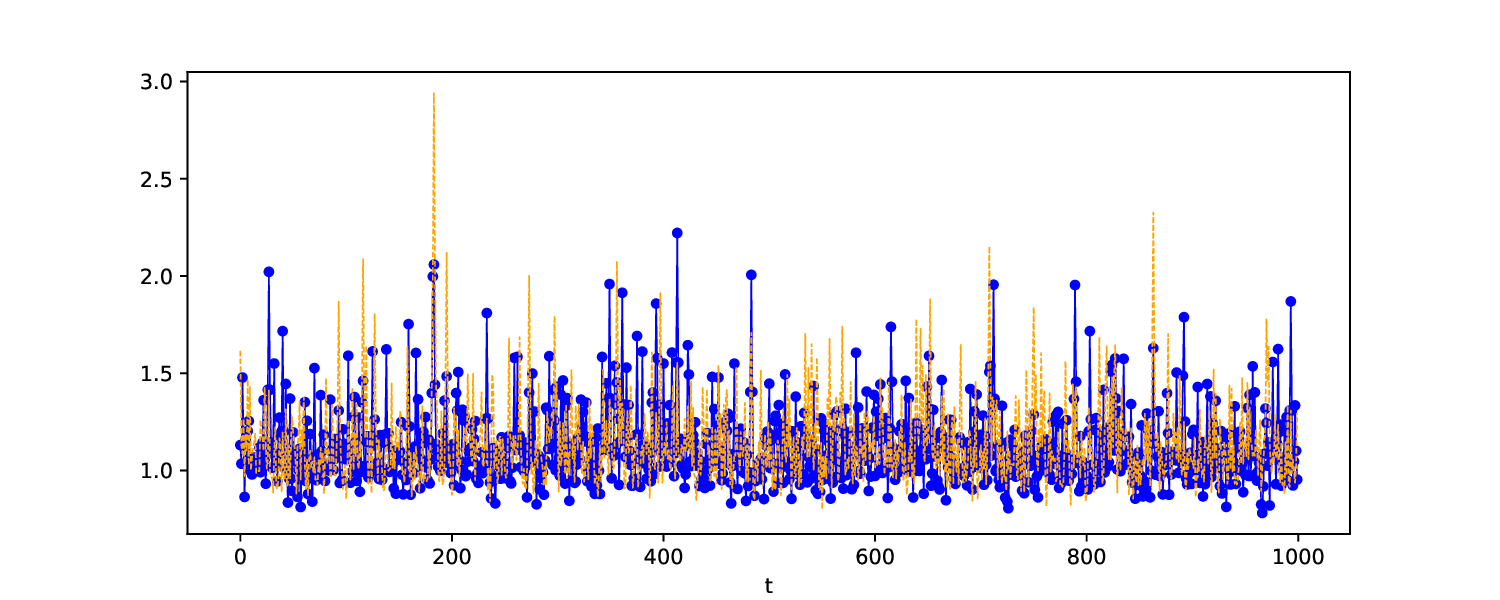}
\caption{A simulation of $(X_t^{(1)}, Y_t^{(1)})$. The
blue solid line represents $(X_t^{(1)})$ and the orange dotted line represents $(Y_t^{(1)})$. }
\label{fig:mmasamplepath}
\end{figure}

\begin{figure}[htbp]
    \includegraphics[width=0.5\linewidth]{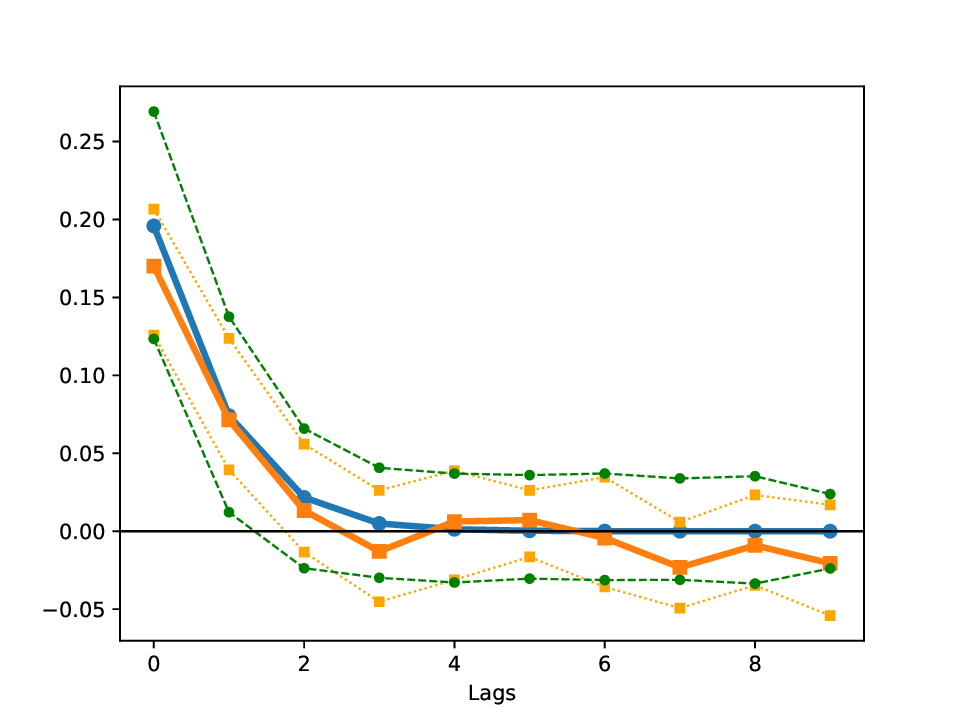}
  \caption{The TMES $\delta(h)$ at lags $h=0,1,\ldots,9$ (blue
    solid line with circle marks) and the empirical TMES $\widehat{\delta}(h)$ at lags
    $h=0,1,\ldots, 9$ (orange solid line with square marks) of
    $(X_t^{(1)}, Y_t^{(1)})$ along with the bootstrapped $90\%$-confidence interval
    (orange dotted line) and the simulated $90\%$-confidence interval
    (green dashed line). }
  \label{fig:mmacompare}
\end{figure}

\begin{figure}[htbp]
  \centering
  \begin{subfigure}[b]{0.49\textwidth}
  \includegraphics[width=\textwidth]{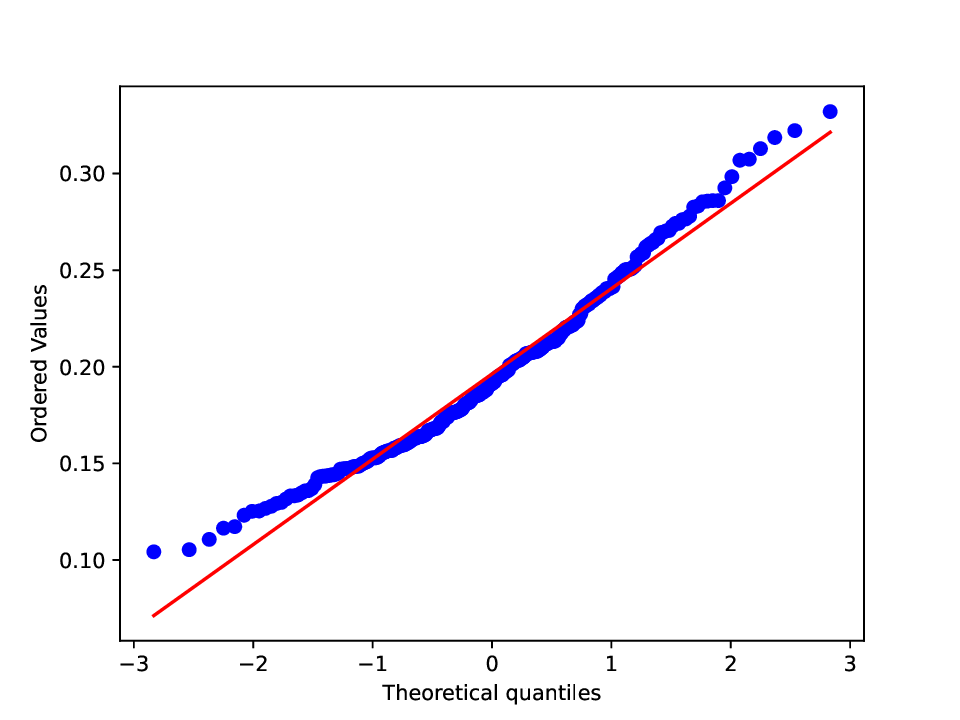}
  \caption{}
  \end{subfigure}
  \hfill
  \begin{subfigure}[b]{0.49\textwidth}
  \includegraphics[width=\textwidth]{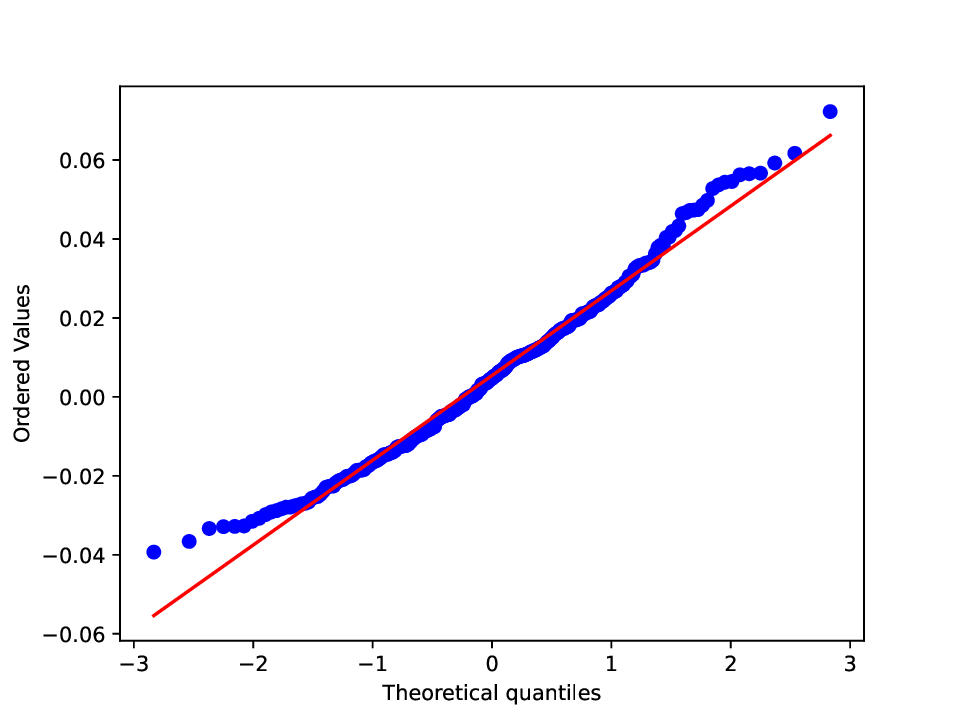}
  \caption{}
  \end{subfigure}
  \vskip\baselineskip
  \begin{subfigure}[b]{0.49\textwidth}
  \includegraphics[width=\textwidth]{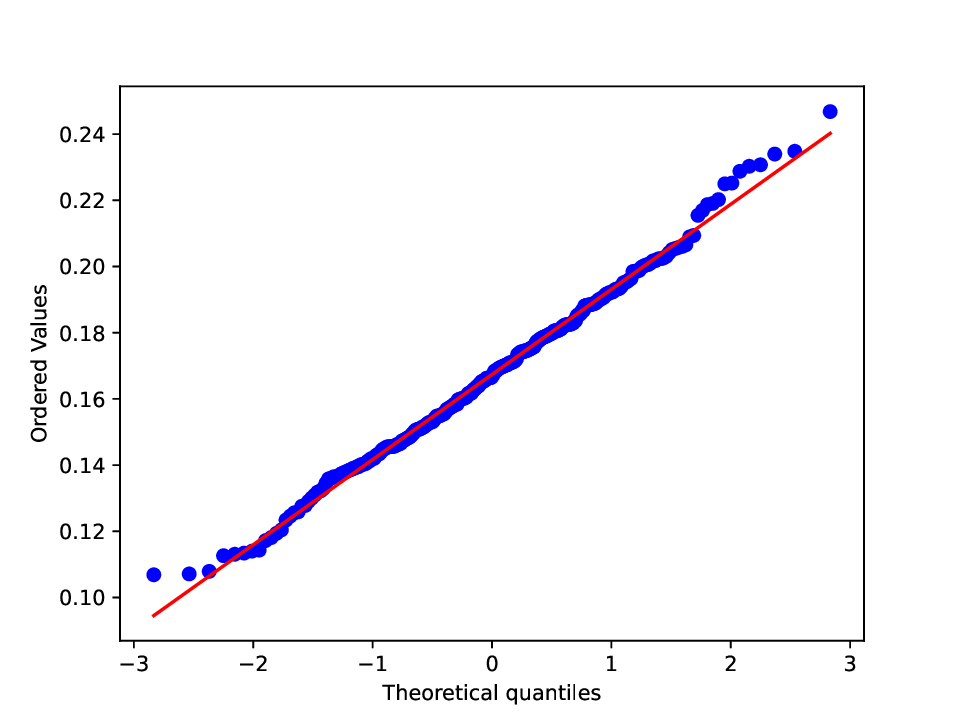}
  \caption{}
  \end{subfigure}
  \hfill
  \begin{subfigure}[b]{0.49\textwidth}
    \includegraphics[width=\textwidth]{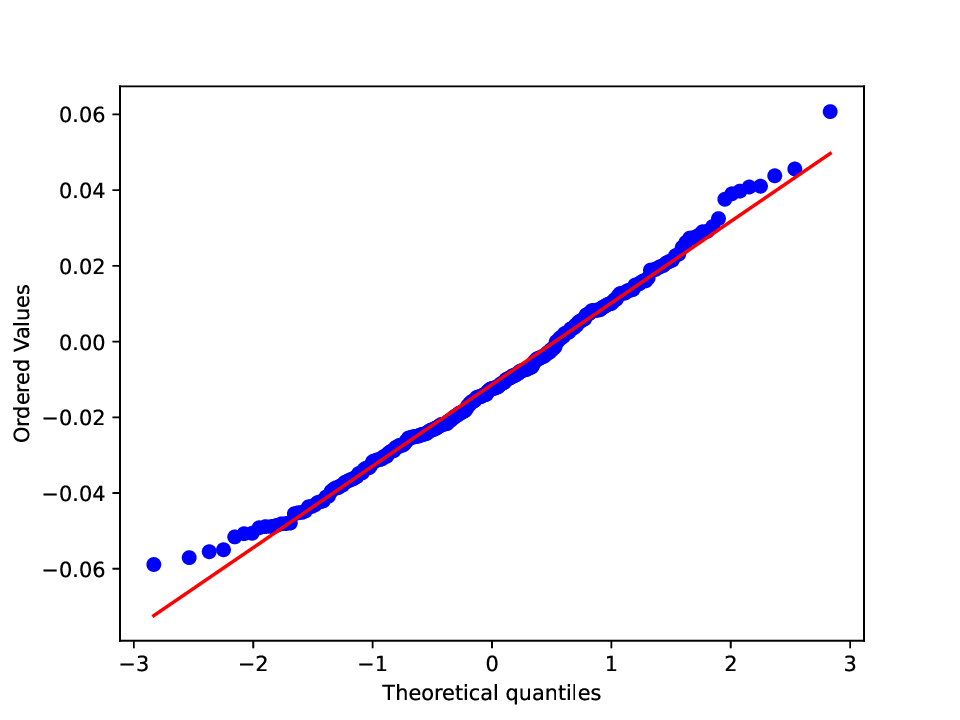}
    \caption{}
  \end{subfigure}
  \caption{ QQ-plots of the bootstrapped and simulation distributions
    against the standard normal distribution for $(X_t^{(1)}, Y_t^{(1)})$. {\bf Up-left}: The
    simulated distribution for $\widehat{\delta}(0)$. {\bf Up-right}: The
    simulated distribution for $\widehat{\delta}(3)$. {\bf Low-left}: The
    bootstrapped distribution of $\delta^{\star}(0)$. {\bf Low-right}: The
    bootstrapped distribution of $\delta^{\star}(3)$.  }
  \label{fig:mmaqq}
\end{figure}

\subsection{Copula-based time series}

We will consider two specific models in this section. In the first
model we choose
\begin{align}\label{eq:armag}
Y_t^{(2)} = \phi Y_{t-1}^{(2)} + Z_t^{(2)} + \theta Z_{t-1}^{(2)} = 0.2 Y_{t-1}^{(2)} + Z_t^{(2)} + 0.8 Z_{t-1}^{(2)}\,,
\end{align}
where $(Z_t^{(2)})$ is an iid sequence of Fr\'{e}chet random variables
with tail index $\xi=6$. The copula of $(X_t^{(2)}, Y_t^{(2)})$ is
chosen as a Gaussian copula with the correlation $0.6$. The marginal
distribution of $X_t^{(2)}$ is the standard normal distribution, which
is light-tailed. The dependence structure of $(X_t^{(2)}, Y_t^{(2)})$
is similar to a Gaussian process. According to
Proposition~\ref{prop:latent}, the theoretical values of $\delta (h)$
are determined by the extremogram $\rho^{(2)}(h) = \lim_{n\to \infty} \p(Y_{h+1}^{(2)} >a_{m_n} \mid Y_1^{(2)}
>a_{m_n})$, $\E[X_1]$ and $\delta(0) = \lim_{n\to \infty} \E[X_1^{(2)} \mid
Y_1^{(2)} >a_{m_n}]$. The formula of $\rho^{(2)}(h)$  is provided in
Appendix B of \cite{mikosch2014}  
\begin{align*}
\rho^{(2)}(h)= \frac{\phi^{\xi(h-1)} (\theta+\phi)^{\xi} + \phi^{\xi h}(\theta + \phi)^{\xi} (1-
  \phi^{\xi})^{-1}}{1+ (\theta + \phi)^{\xi} (1- \phi^{\xi})^{-1}}\,, \quad h\ge 1\,.
\end{align*}
Trivially, $\E[X_1]=0$ and the value of $\delta(0)$ is obtained by
applying a Monte Carlo method of which
the details are omitted. A sample path of $(X_t^{(2)},
Y_t^{(2)})$ is presented in Figure~\ref{fig:armagsamplepath}. We
compare the empirical TMES $\widehat{\delta}(h)$ with $\delta(h)$ at lags
$h=0,1,\ldots, 9$ for $(X_t^{(2)}, Y_t^{(2)})$ in Figure~\ref{fig:armagcompare}. This indicates the
consistency of $\widehat{\delta}(h)$. We present the QQ-plots of the simulated
distributions for $\widehat{\delta}(0)$ and $\widehat{\delta}(3)$ and the
bootstrapped distribution of $\delta^{\star}(0)$ and $\delta^{\star}(3)$ against the standard normal
distribution in Figure~\ref{fig:armagqq}. These QQ-plots shows the
asymptotic normality of $\widehat{\delta} (h)$ and $\delta^{\star}(h)$.

We assume that $(Y_t^{(3)})$ follows a GARCH model 
\begin{align} \label{eq:garcht}
Y_t^{(3)}  = \sigma_t Z_t\,, \quad \sigma_t^2 =0.2 + 0.3 Y^2_{t-1} + 0.3 \sigma^2_{t-1}\,, \quad
  Z_t \overset{iid}{\sim} N(0,1)\,,
\end{align}
and the copula of $(X^{(3)}_t, Y^{(3)}_t)$ is a t-copula with $3$ degrees of freedom and the
correlation $0.7$. The marginal distribution of $X_t^{(3)}$ is the
standard normal distribution. The time series $(X_t^{(3)}, Y_t^{(3)})$
is the second model, whose dependence structure is similar to the
classical regularly varying time series. In Figure~\ref{fig:garchtsamplepath}, a sample path
of $(X_t^{(3)}, Y_t^{(3)})$ is presented. Since there is  no explicit formula
for the extremogram $\rho^{(3)}(h) = \lim_{n\to \infty} \p(Y_{h+1}^{(3)} >a_{m_n} \mid Y_1^{(3)}
>a_{m_n})$ and $\delta(0) = \lim_{n\to \infty} \E[X_t^{(3)}\mid Y_t^{(3)}
>a_{m_n}]$, we use the Monte Carlo method to obtain the values of
$\delta(h)$. We compare the empirical TMES $\widehat{\delta}(h)$ with $\delta(h)$ at lags
$h=0,1,\ldots, 9$ for $(X_t^{(3)}, Y_t^{(3)})$ in
Figure~\ref{fig:garchtcompare}, in which $\delta(h)$ stays in the
bootstrapped and simulated $90\%$-confidence intervals. The consistency of $\widehat{\delta}(h)$
for $(X_t^{(3)}, Y_t^{(3)})$ is verified. By comparing the
bootstrapped and simulated distributions with the standard normal
distribution, we verify the asymptotic normality of $\widehat{\delta}(h)$
and $\delta^{\star}(h)$ for $(X_t^{(3)}, Y_t^{(3)})$ in
Figure~\ref{fig:garchtqq}.

\begin{figure}[htbp]
\centering  
\includegraphics[scale=0.53]{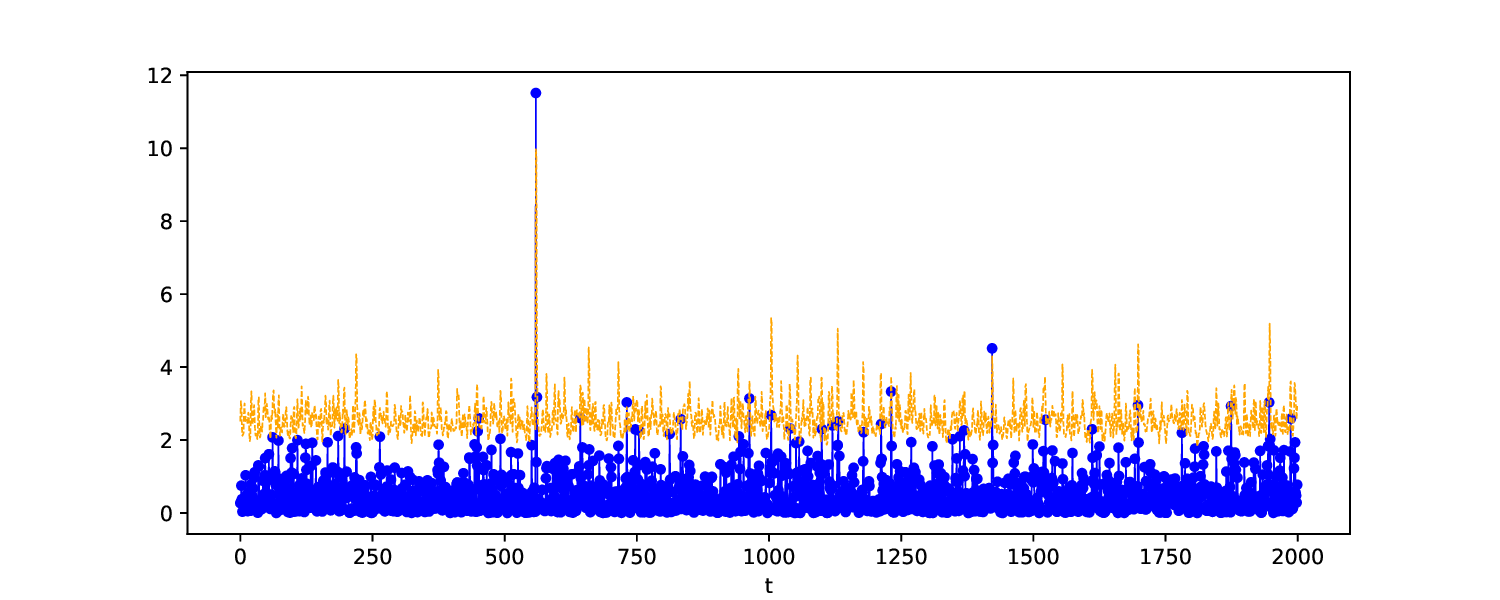}
\caption{A simulation of $(X_t^{(2)}, Y_t^{(2)})$. The
blue solid line represents $(X_t^{(2)})$ and the orange dotted line represents $(Y_t^{(2)})$.}
\label{fig:armagsamplepath}
\end{figure}

\begin{figure}[htbp]
    \includegraphics[width=0.5\linewidth]{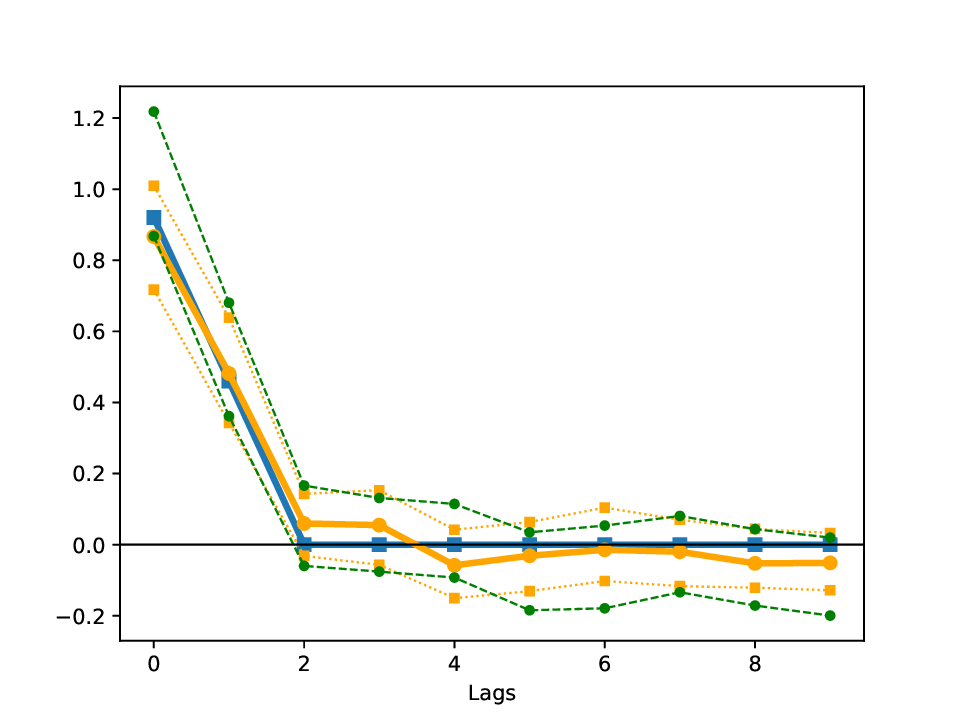}
  \caption{The TMES $\delta(h)$ at lags $h=0,1,\ldots,9$ (blue
    solid line with circle marks) and the empirical TMES $\widehat{\delta}(h)$ at lags
    $h=0,1,\ldots, 9$ (orange solid line with square marks) of
    $(X_t^{(2}, Y_t^{(2)})$ along with the bootstrapped $90\%$-confidence interval
    (orange dotted line) and the simulated $90\%$-confidence interval
    (green dashed line). }
  \label{fig:armagcompare}
\end{figure}

\begin{figure}[htbp]
  \centering
  \begin{subfigure}[b]{0.49\textwidth}
  \includegraphics[width=\textwidth]{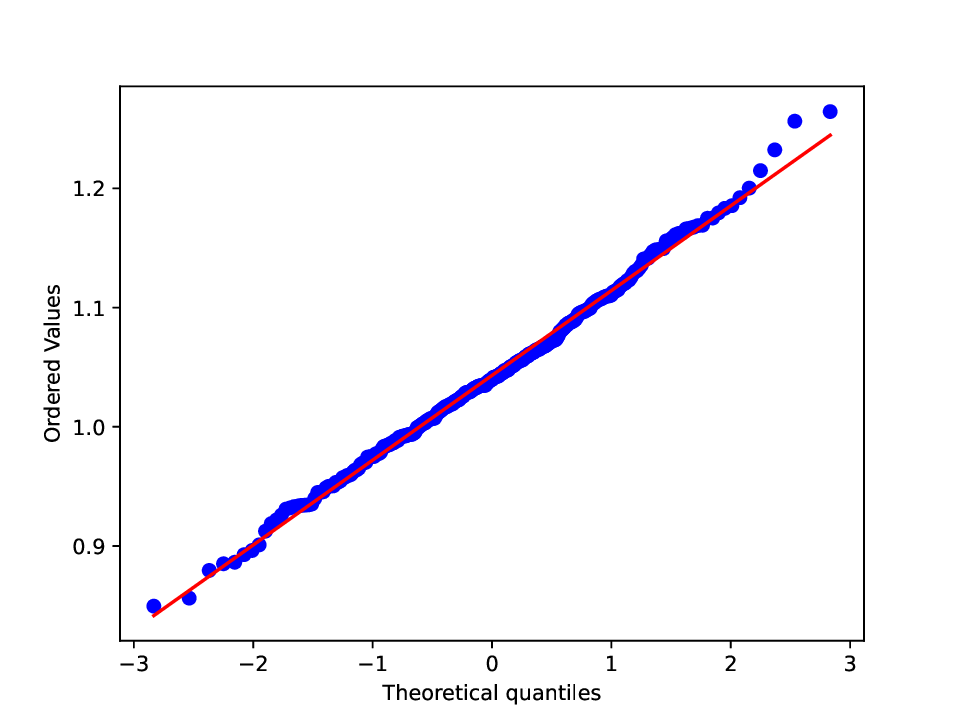}
  \caption{}
  \end{subfigure}
  \hfill
  \begin{subfigure}[b]{0.49\textwidth}
  \includegraphics[width=\textwidth]{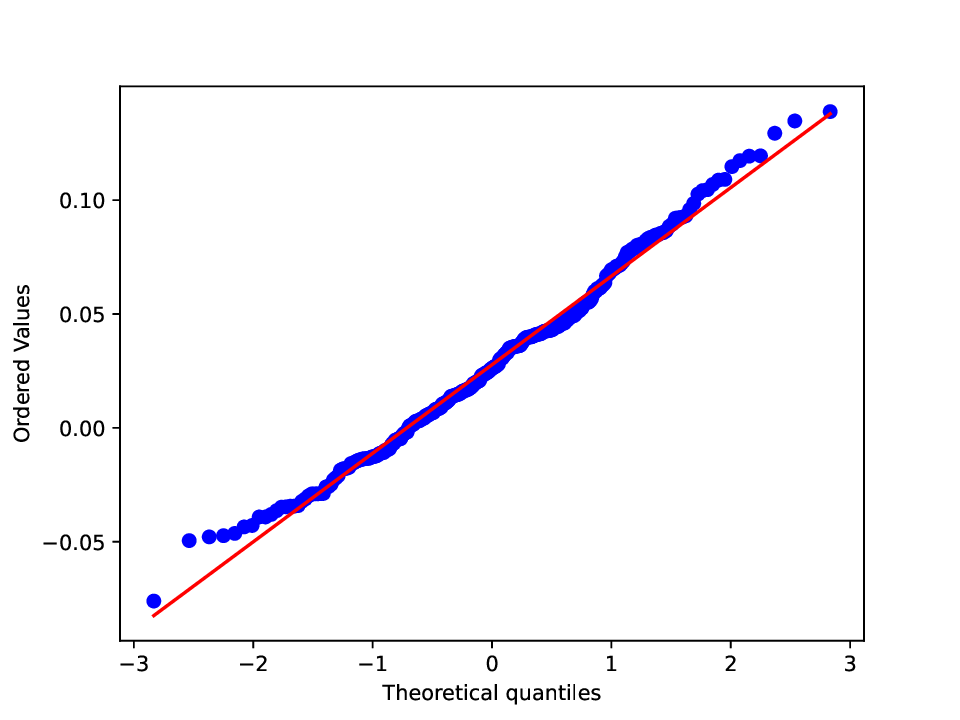}
  \caption{}
  \end{subfigure}
  \vskip\baselineskip
  \begin{subfigure}[b]{0.49\textwidth}
  \includegraphics[width=\textwidth]{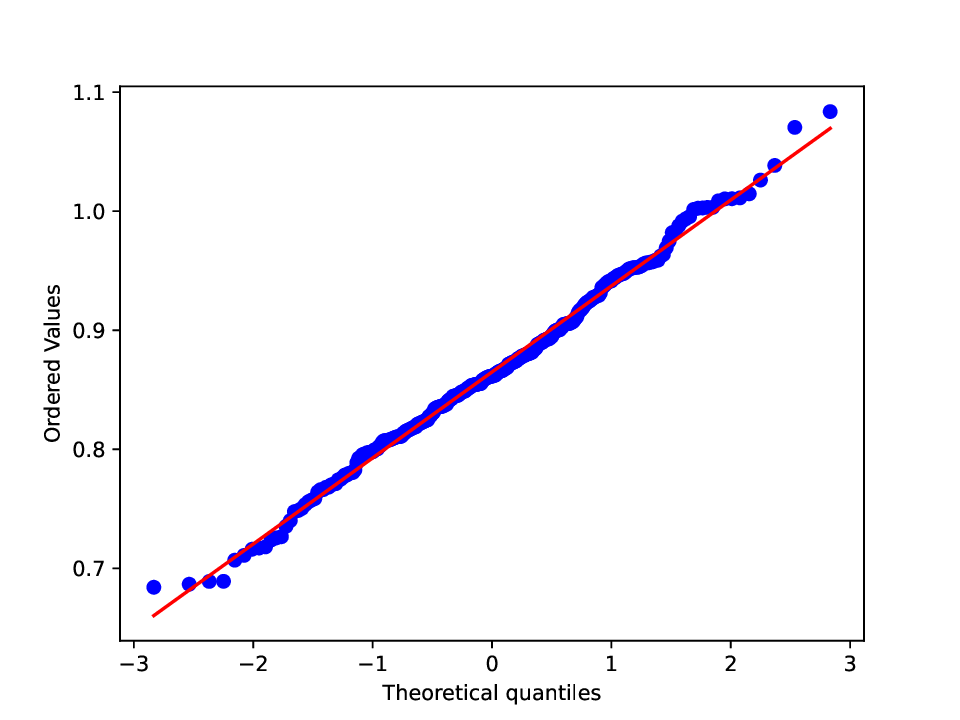}
  \caption{}
  \end{subfigure}
  \hfill
  \begin{subfigure}[b]{0.49\textwidth}
    \includegraphics[width=\textwidth]{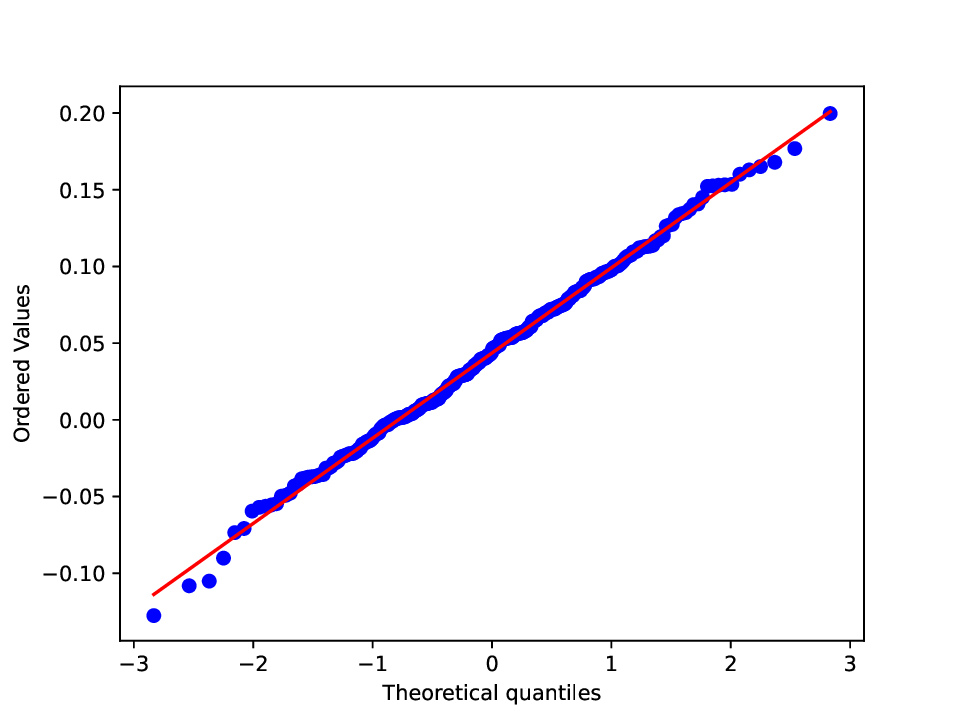}
    \caption{}
  \end{subfigure}
  \caption{QQ-plots of the bootstrapped and simulation distributions
    against the standard normal distribution for $(X_t^{(2)}, Y_t^{(2)})$. {\bf Up-left}: The
    simulated distribution for $\widehat{\delta}(0)$. {\bf Up-right}: The
    simulated distribution for $\widehat{\delta}(3)$. {\bf Low-left}: The
    bootstrapped distribution of $\delta^{\star}(0)$. {\bf Low-right}: The
    bootstrapped distribution of $\delta^{\star}(3)$. }
  \label{fig:armagqq}
\end{figure}

\begin{figure}[htbp]
\centering  
\includegraphics[scale=0.53]{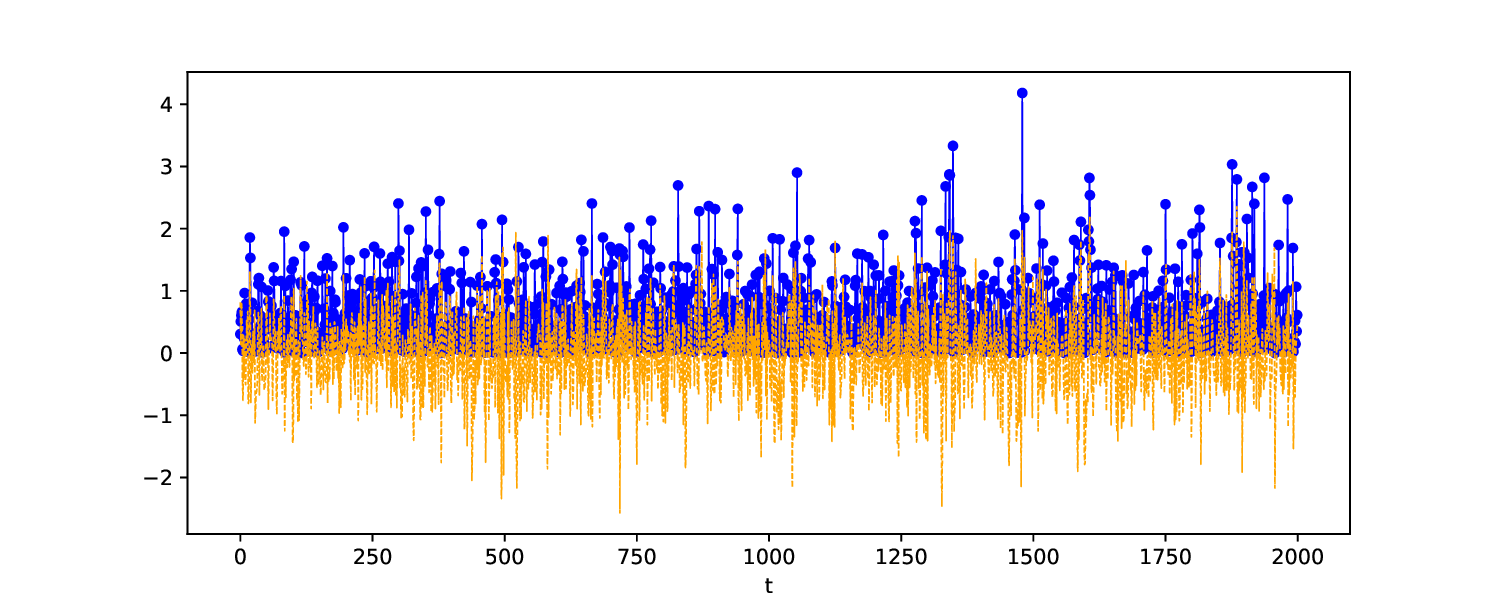}
\caption{A simulation of $(X_t^{(3)}, Y_t^{(3)})$. The
blue solid line represents $(X_t^{(3)})$ and the orange dotted line represents $(Y_t^{(3)})$.}
\label{fig:garchtsamplepath}
\end{figure}

\begin{figure}[htbp]
    \includegraphics[width=0.5\linewidth]{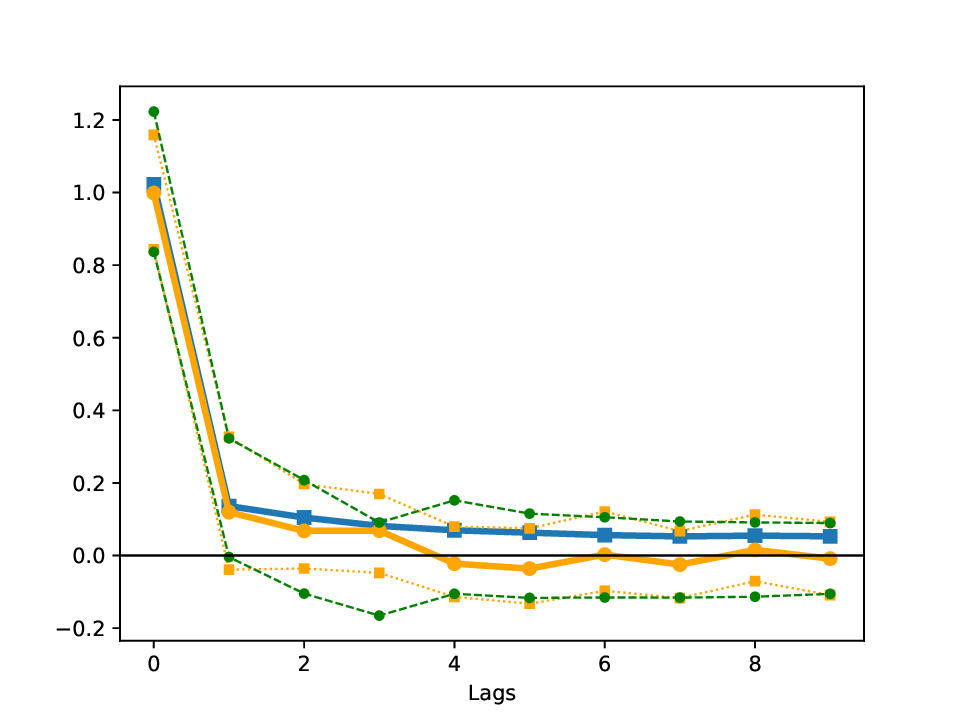}
  \caption{ The TMES $\delta(h)$ at lags $h=0,1,\ldots,9$ (blue
    solid line with circle marks) and the empirical TMES $\widehat{\delta}(h)$ at lags
    $h=0,1,\ldots, 9$ (orange solid line with square marks) of
    $(X_t^{(3}, Y_t^{(3)})$ along with the bootstrapped $90\%$-confidence interval
    (orange dotted line) and the simulated $90\%$-confidence interval
    (green dashed line). }
  \label{fig:garchtcompare}
\end{figure}

\begin{figure}[htbp]
  \centering
  \begin{subfigure}[b]{0.49\textwidth}
  \includegraphics[width=\textwidth]{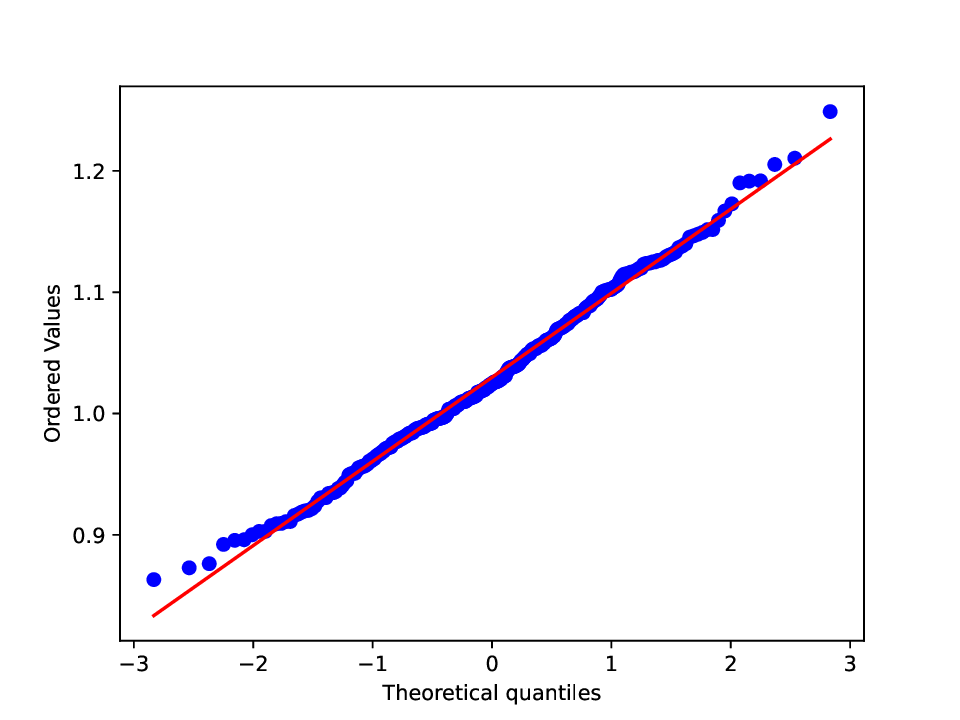}
  \end{subfigure}
  \hfill
  \begin{subfigure}[b]{0.49\textwidth}
  \includegraphics[width=\textwidth]{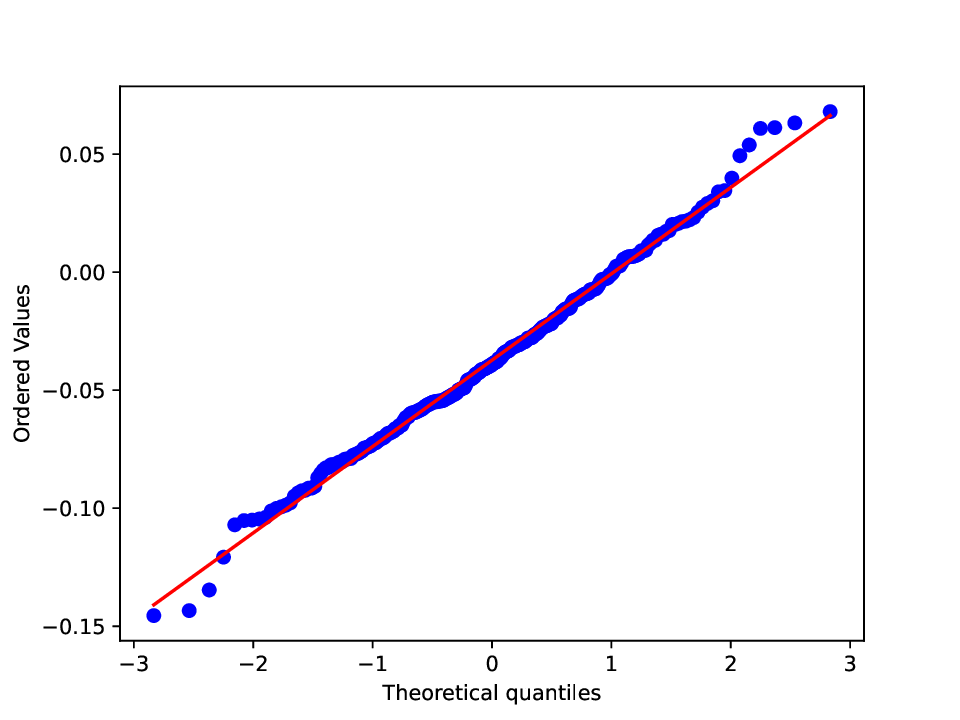}
  \end{subfigure}
  \vskip\baselineskip
  \begin{subfigure}[b]{0.49\textwidth}
  \includegraphics[width=\textwidth]{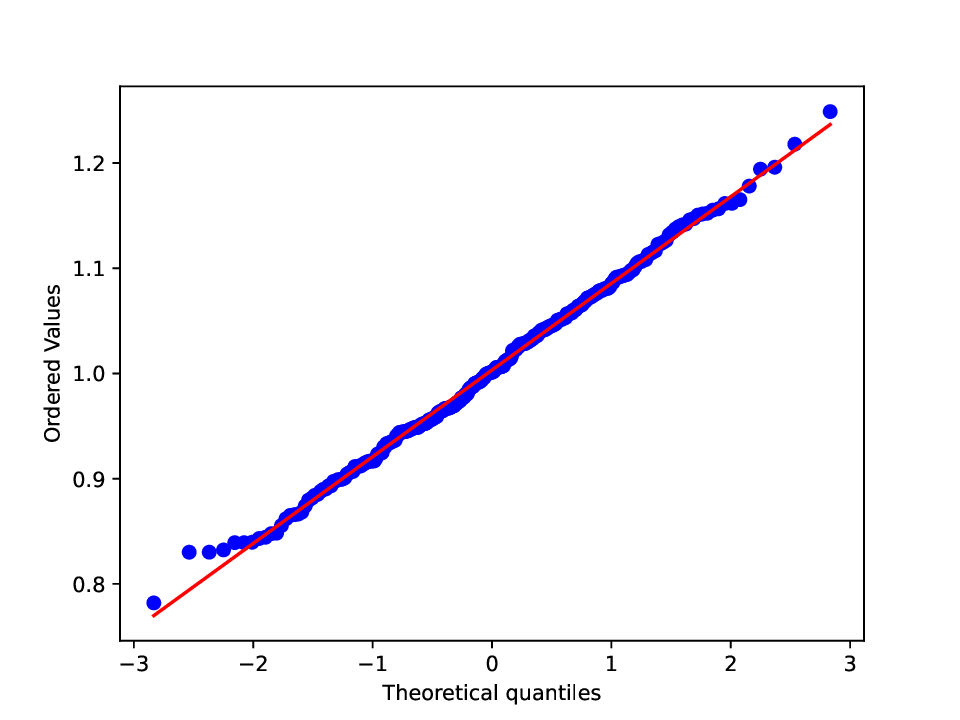}
  \end{subfigure}
  \hfill
  \begin{subfigure}[b]{0.49\textwidth}
    \includegraphics[width=\textwidth]{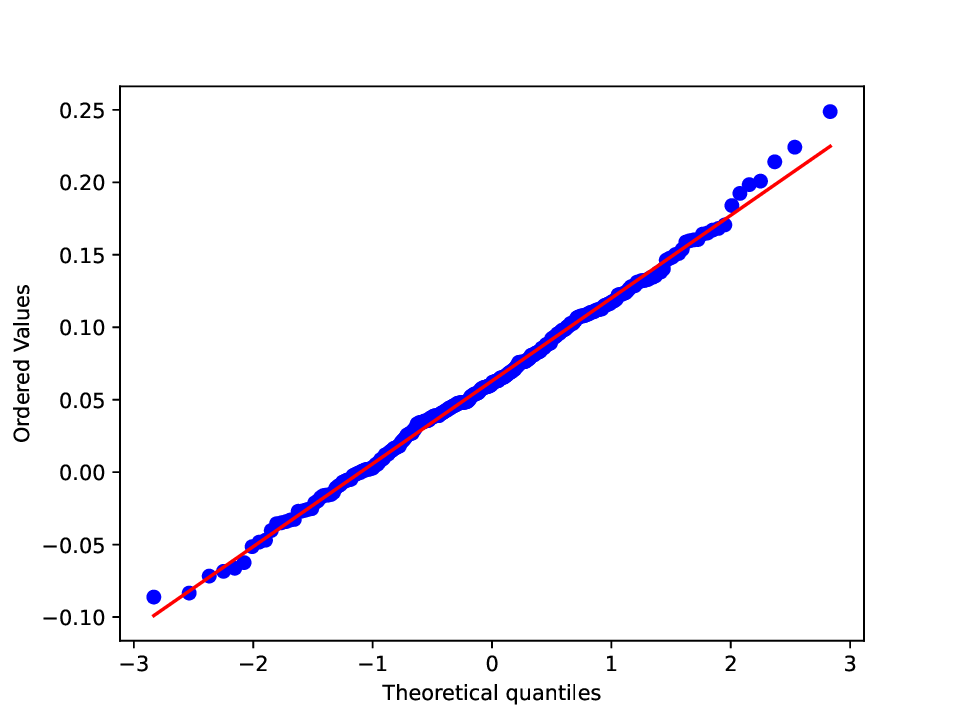}
  \end{subfigure}
  \caption{QQ-plots of the bootstrapped and simulation distributions
    against the standard normal distribution for $(X_t^{(3)}, Y_t^{(3)})$. {\bf Up-left}: The
    simulated distribution for $\widehat{\delta}(0)$. {\bf Up-right}: The
    simulated distribution for $\widehat{\delta}(3)$. {\bf Low-left}: The
    bootstrapped distribution of $\delta^{\star}(0)$. {\bf Low-right}: The
    bootstrapped distribution of $\delta^{\star}(3)$. }
  \label{fig:garchtqq}
\end{figure}

\section{Real data analyses}\label{sec:realdataanalyses}
Recall that the centered TMES at lag $h\ge 0$, \[\delta_0(h) = \lim_{u\to \infty} \E[X_t \mid Y_{t-h} >u] - \E[X_t]\,.\] If $\delta_0 (h)=0$
it signifies that the systemic event $\{ a_{m_n}^{-1} \mathbf{Y}_{t-h} \in A \}$ does not influence the random variable $X_t$. To test whether $\delta_0 (h) = 0$, we employ confidence intervals derived from the stationary bootstrap method outlined in Section~\ref{sec:bootstrap}. For the empirical TMES, we set $m_n=20$ using $\theta =1/10$ as the parameter in the bootstrap algorithm and generate 300 bootstrapped samples throughout this section. 

\subsection{Squared log-returns of COMEX gold continuous contracts'
  closing prices and US EPU index}

According to the arguments in \cite{baker2016measuring}, the US EPU
index denoted by $(Y_t^{(4)})$ reflects economic uncertainty caused by
policies in the United States, and thus, we will use 
the {\em US EPU index} (USEPU) to indicate the level of the systemic risk. Gold is a
popular hedging instrument against inflation or other economic
disruptions. The {\em squared log-returns of the COMEX gold continuous
contracts' closing prices} (SLCGP), denoted by $(X_t^{(4)})$, act as a
risk measure of gold prices. SLCGP has, in general, no relation with
USEPU unless economic uncertainty arises significantly. Therefore, the
centered empirical TMES $\widehat{\delta}_0 (h)= \widehat{\delta}(h) - n^{-1} \sum_{t=1}^n X_t$, $h\ge 0$ are zero unless a
systemic event happens.

The data includes $9,786$ daily observations of
SLCGP and USEPU from January 4, 2000 to April 30, 2024 covering the
2007-2009 financial crisis and the COVID-19 pandemic. In Figure~\ref{fig:financesample}, we present the values of USEPU and
SLCGP. By applying a moving window method, we calculate the centered empirical TMES's $\widehat{\delta}_0 (h)$ for
the segment of the data in each window of size $200$ and
$\widehat{\delta}_0 (h)$ is indexed by the end date of the corresponding
window. In Figure~\ref{fig:tmesfinance}, the centered empirical TMES's
$\widehat{\delta}_0 (h)$, $h=0,1,3,7$, for all the segments are shown and
 $\widehat{\delta}_0 (h)$ takes values close to zero during most of the whole
 period. When a systemic event happens, the behaviors of $\widehat{\delta}_0 (h)$ at different lags vary. 

 We take the data from April 7, 2008 to October 21, 2009 to examine
 the performance of $\widehat{\delta}_0 (h)$ under the $2007-2009$ financial
 crisis. In Figure~\ref{fig:finance2008}, we show the curve of
 $\widehat{\delta}_0 (h)$, $h=0,1,3,7$ with the corresponding
 $90\%$-confidence intervals. The zero value is included in the
 $90\%$-confidence intervals of $\widehat{\delta}_0 (h)$, $h=0,3,7$, but the
 zero value is excluded from the pointwise confidence intervals of
 $\widehat{\delta}_0 (1)$ around September, 2008 when the crisis happened.
 Unlike the centered empirical TMES around the $2007-2009$ financial
 crisis, the pointwise $90\%$-confidence intervals of the centered
 empirical TMES's $\widehat{\delta}_0 (h)$, $h=0,1,3,7$ in Figure~\ref{fig:finance2020} exclude the zero value in March 2019 when the temporary closures of high-contact businesses and the
 encouragement of remote work by other businesses started in the US.  
\begin{figure}[htbp]
  \centering
  \begin{subfigure}[b]{0.48\textwidth}
  \includegraphics[width=\textwidth]{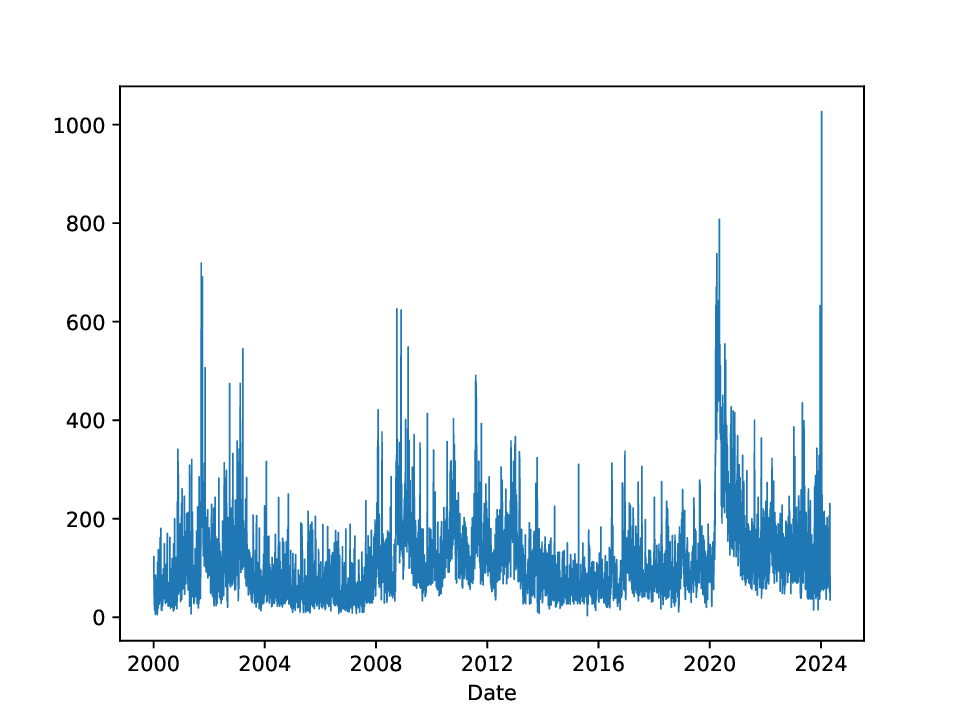}
  \end{subfigure}
  \hfill
  \begin{subfigure}[b]{0.48\textwidth}
  \includegraphics[width=\textwidth]{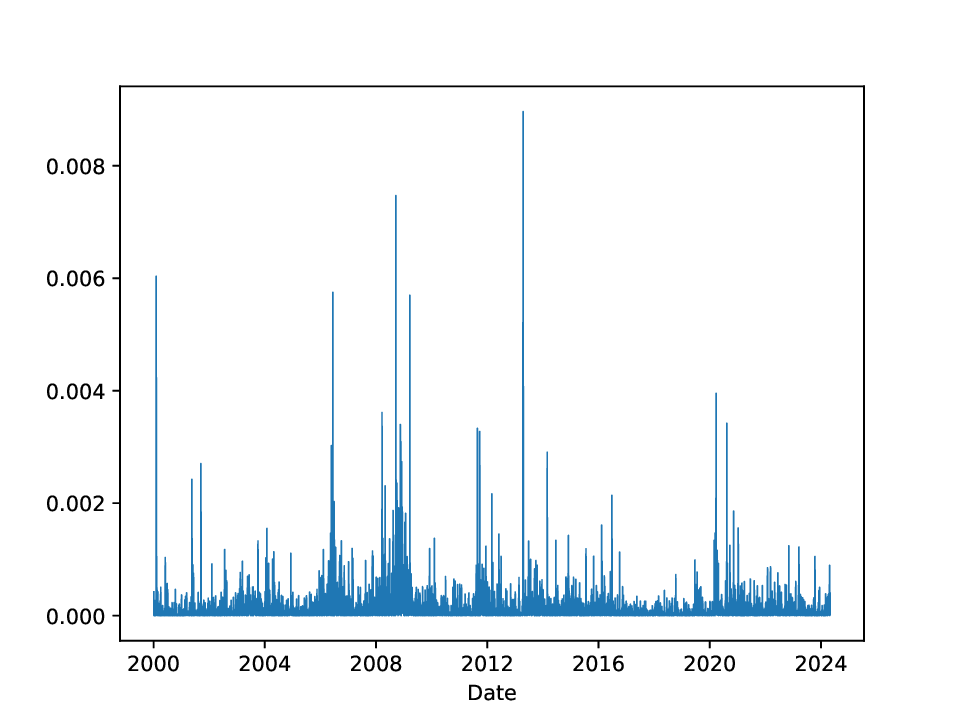}
  \end{subfigure}
  \caption{{\bf Left}: The US EPU index (USEPU) from January 4,
    2000 to April 30, 2024. {\bf Right}: Daily squared log-returns of
    the COMEX gold continuous contracts' closing prices (SLCGP) from January 4,
    2000 to April 30, 2024. }
  \label{fig:financesample}
\end{figure}

\begin{figure}[htbp]
  \centering
   \includegraphics[width=0.8 \textwidth]{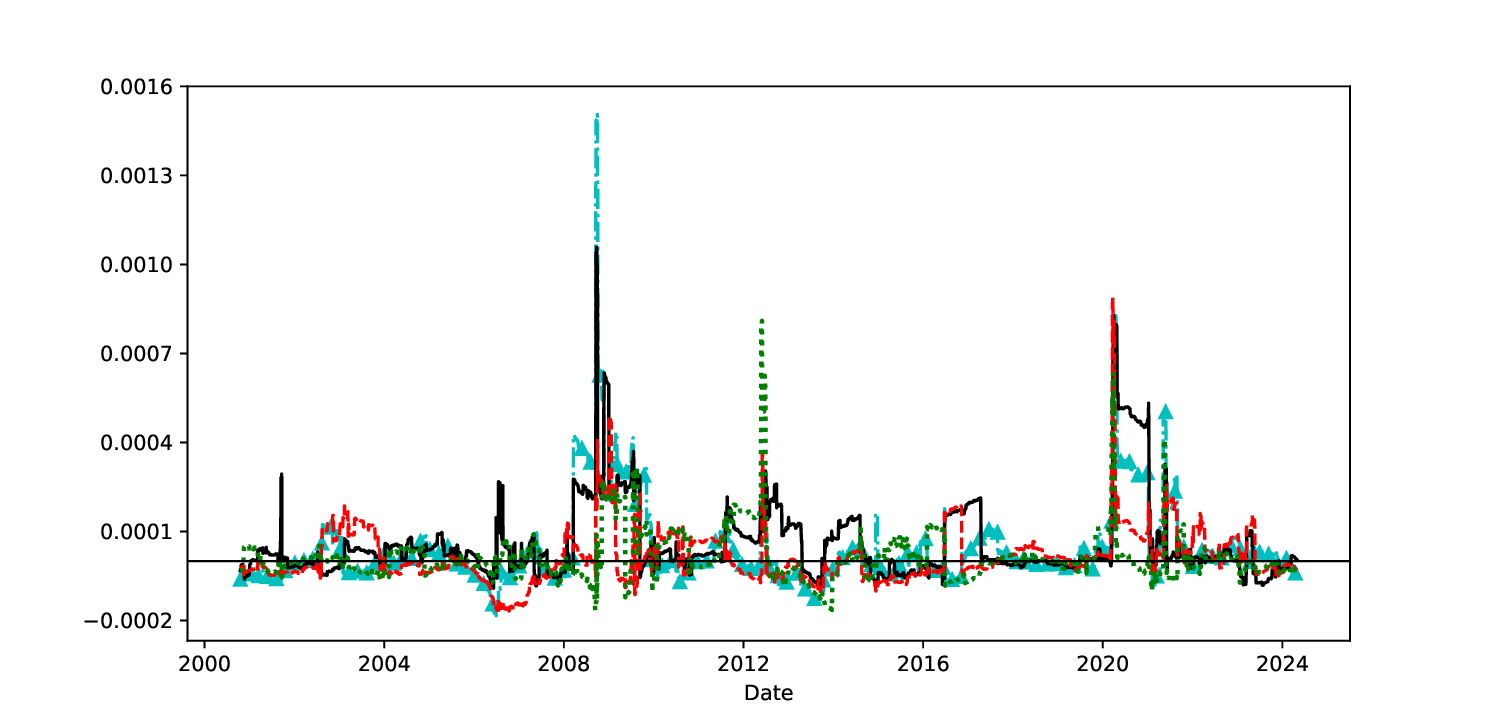}
  \caption{The centered empirical TMESs' $ \widehat{\delta}_0 (0)$ (black
    solid line), $\widehat{\delta}_0 (1)$ (cyan dash-dot line with triangular
    markers), $\widehat{\delta}_0 (3)$ (red dashed line), and
    $\widehat{\delta}_0 (7)$ (green dotted line) for $(X_t^{(4)}, Y_t^{(4)})$,
    where $X_t^{(4)}$ represents USEPU and $Y_t^{(4)}$ represents SLCGP.   }
  \label{fig:tmesfinance}
\end{figure}

\begin{figure}[htbp]
  \centering
  \begin{subfigure}[b]{0.49\textwidth}
  \includegraphics[width=\textwidth]{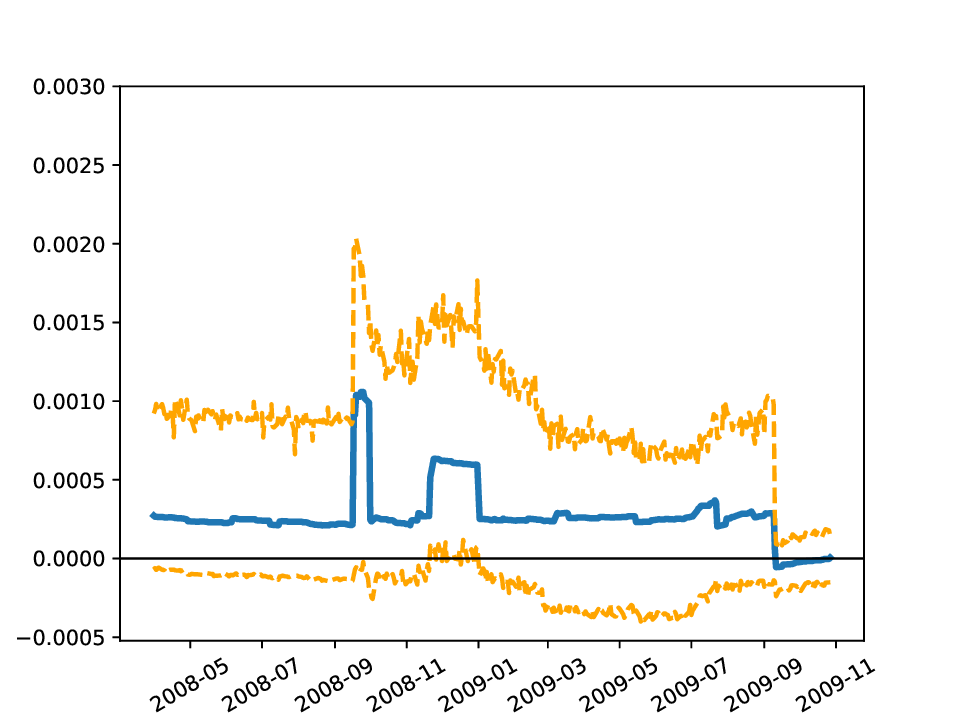}
  \end{subfigure}
  \hfill
  \begin{subfigure}[b]{0.49\textwidth}
  \includegraphics[width=\textwidth]{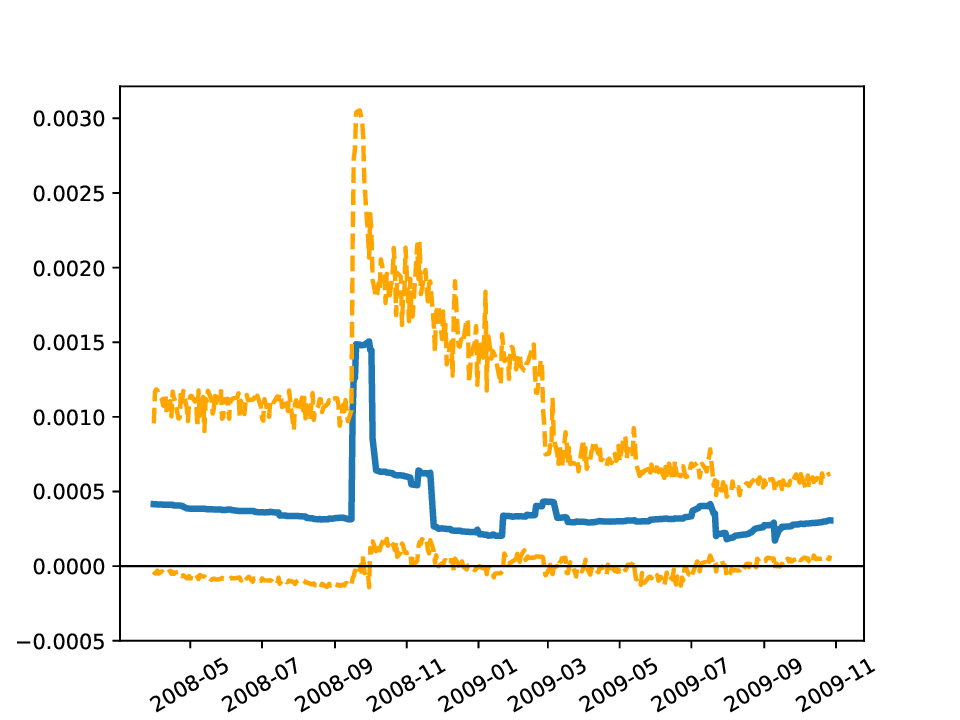}
  \end{subfigure}
\vskip\baselineskip
  \begin{subfigure}[b]{0.49\textwidth}
  \includegraphics[width=\textwidth]{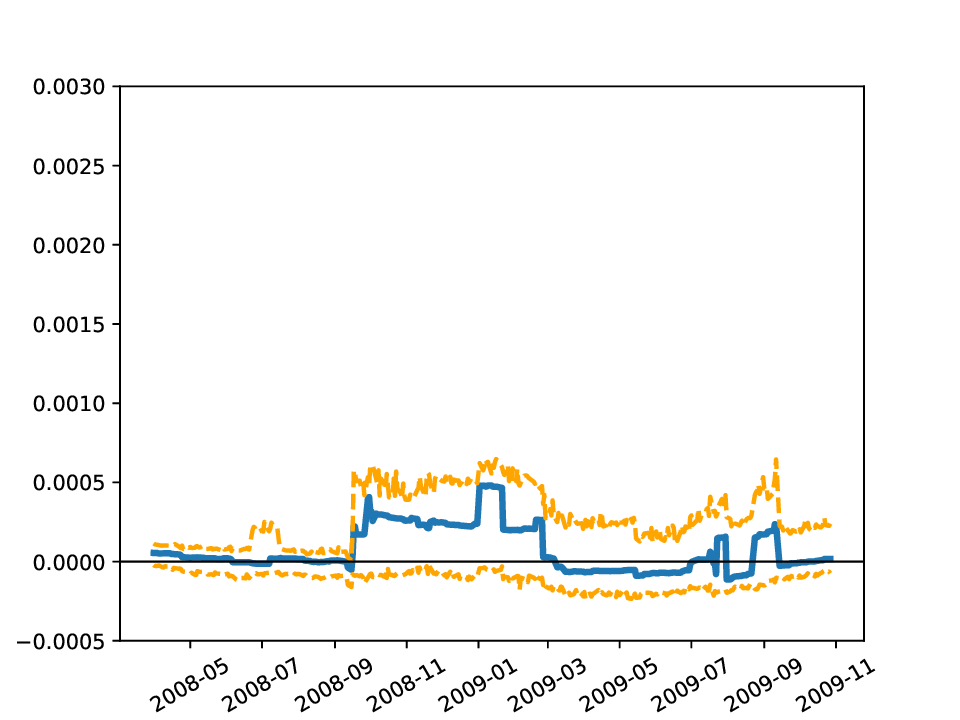}
  \end{subfigure}
  \hfill
  \begin{subfigure}[b]{0.49\textwidth}
  \includegraphics[width=\textwidth]{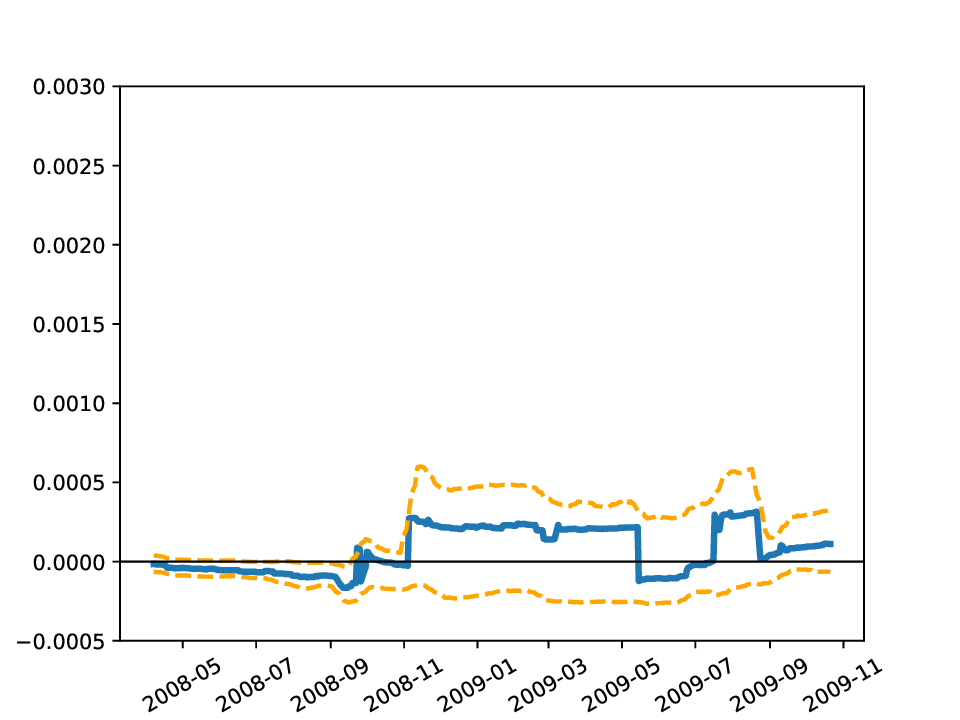}
  \end{subfigure}
\caption{The centered empirical TMESs' $ \widehat{\delta}_0 (0)$ (up-left), $\widehat{\delta}_0 (1)$ (up-right),
  $\widehat{\delta}_0 (3)$ (low-left), and
    $\widehat{\delta}_0 (7)$ (low-right) and their corresponding
    bootstrapped $90\%$-confidence intervals for $(X_t^{(4)}, Y_t^{(4)})$ around the
    2008 financial crisis (from April 7, 2008 to October 21, 2009),
    where $X_t$ represents USEPU and $Y_t$ represents SLCGP.}
 \label{fig:finance2008}
 \end{figure}
 
 \begin{figure}[htbp] 
  \centering
  \begin{subfigure}[b]{0.49\textwidth}
  \includegraphics[width=\textwidth]{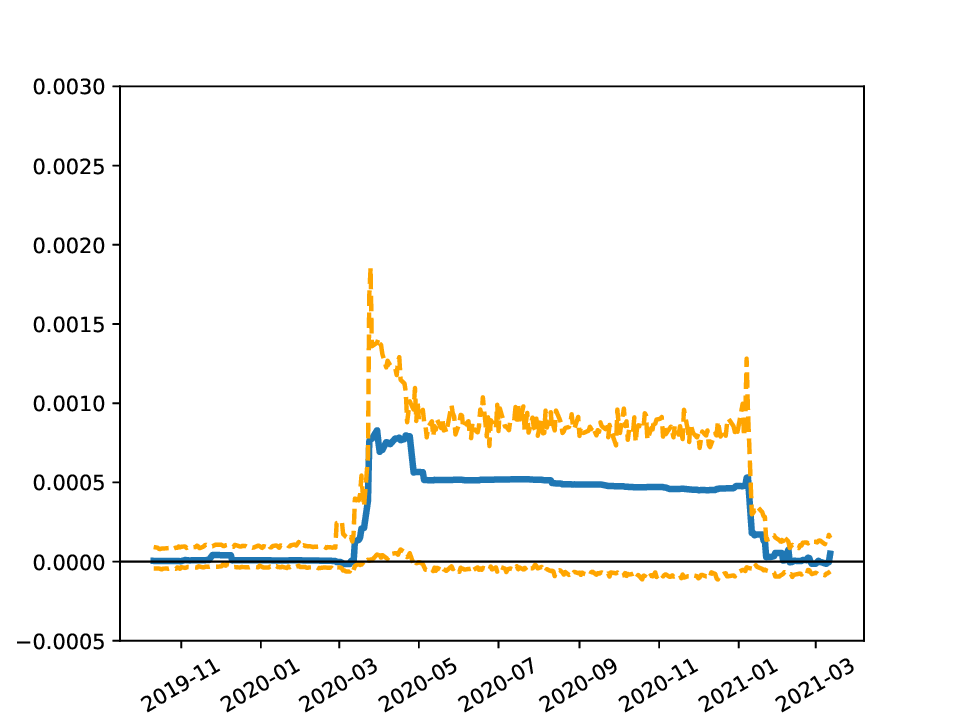}
  \end{subfigure}
  \hfill
  \begin{subfigure}[b]{0.49\textwidth}
  \includegraphics[width=\textwidth]{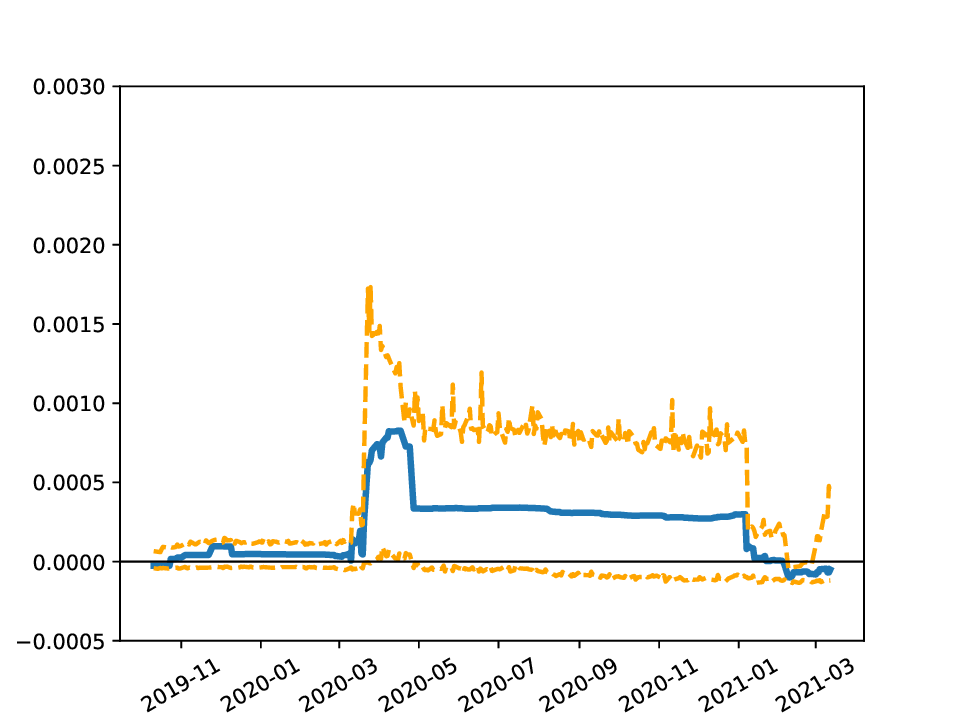}
\end{subfigure}
\vskip\baselineskip
  \begin{subfigure}[b]{0.49\textwidth}
  \includegraphics[width=\textwidth]{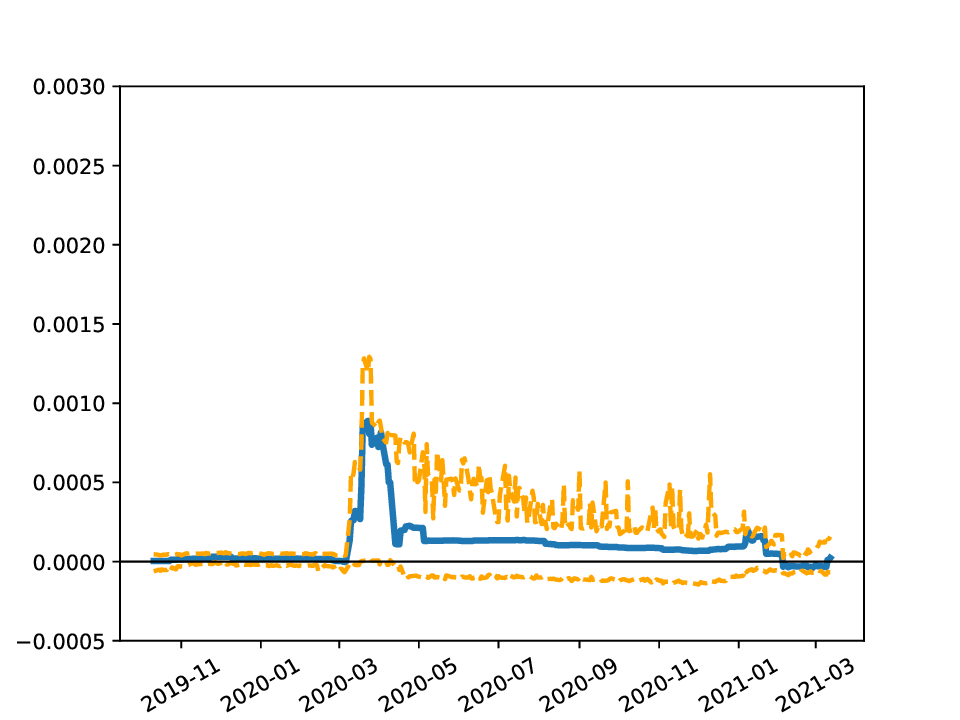}
  \end{subfigure}
  \hfill
  \begin{subfigure}[b]{0.49\textwidth}
  \includegraphics[width=\textwidth]{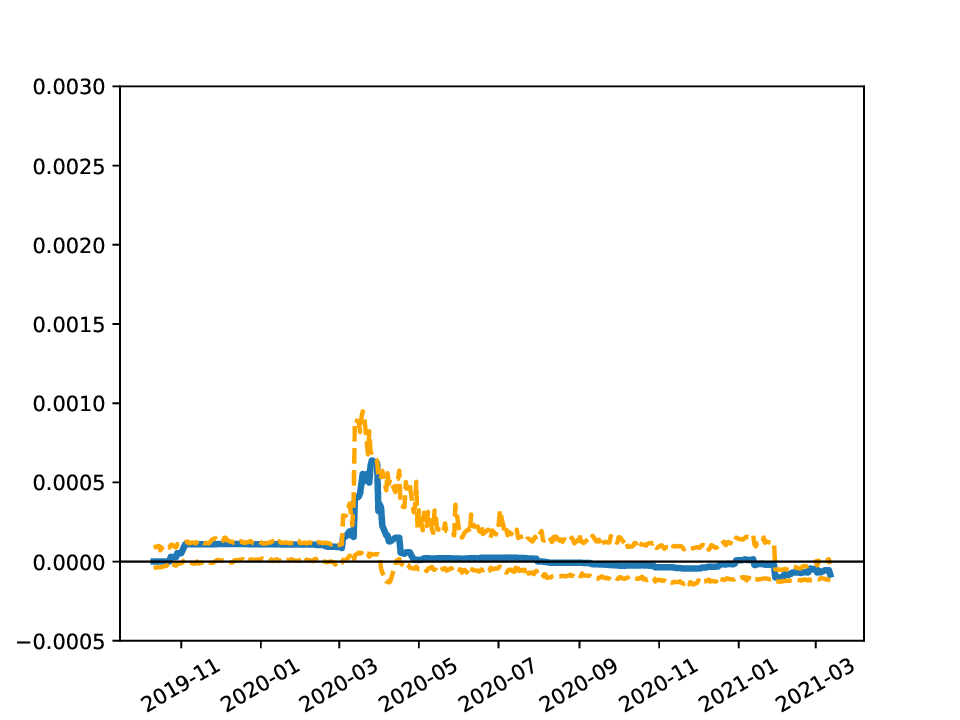}
  \end{subfigure}
\caption{The centered empirical TMESs' $ \widehat{\delta}_0 (0)$ (up-left), $\widehat{\delta}_0 (1)$ (up-right),
  $\widehat{\delta}_0 (3)$ (low-left), and
    $\widehat{\delta}_0 (7)$ (low-right) and their corresponding
    bootstrapped $90\%$-confidence intervals of $(X_t^{(4)}, Y_t^{(4)})$ around the
    COVID-19 pandemic (from October 18, 2019 to March 8, 2021),
    where $X_t^{(4)}$ represents USEPU and $Y_t^{(4)}$ represents SLCGP.}
\label{fig:finance2020}
  \end{figure}
 
\subsection{Relationship between water level and precipitation in a
  city}

We consider the data of water levels and precipitation in Jinan,
Shandong Province, China; see Figure~\ref{fig:map} for the locations
of the water level station and the precipitation observation station. We take the
daily precipitation values in the summer months (June, July and August) from
$2006-2021$ as $(Y_t^{(5)})$ and the daily water levels in the summer months from
$2006-2021$ as $(X_t^{(5)})$. We also take the
daily precipitation values in the autumn months (September, October and November) from
$2006-2021$ as $(Y_t^{(6)})$ and the daily water levels in the autumn months from
$2006-2021$ as $(X_t^{(6)})$. The summer of Jinan has a high
temperature with a lot of rainfalls, while the autumn of Jinan is
cooler but also humid.

The precipitation and the evaporation determine the
water level in the same area. In Figure~\ref{fig:summer}, we present
$(X_t^{(5)})$, $(Y_t^{(5)})$ and the centered empirical TMES's $\widehat{\delta}_0 (h)$
for $(X_t^{(5)}, Y_t^{(5)})$ at lags $h=0,\ldots,30$ with the bootstrapped
$90\%$-confidence intervals. Due to the high temperature in summer, the
evaporation effect overtakes the precipitation effect and
$\widehat{\delta}_0 (h)$ decreases below zero. In Figure~\ref{fig:autumn}, we present
$(X_t^{(6)})$, $(Y_t^{(6)})$ and the centered empirical TMES's $\widehat{\delta}_0 (h)$
for $(X_t^{(6)}, Y_t^{(6)})$ at lags $h=0,\ldots,30$ with the bootstrapped
$90\%$-confidence intervals. As the temperature decreases with fewer rainfalls
 in autumn, the precipitation effect on the water level
becomes more significant in a relatively short period. The centered empirical
TMES's $\widehat{\delta}_0 (h)$ increases when $h\le 7$ and decrease when $h>7$.

\begin{figure}[htbp]
\centering  
\includegraphics[scale=0.5]{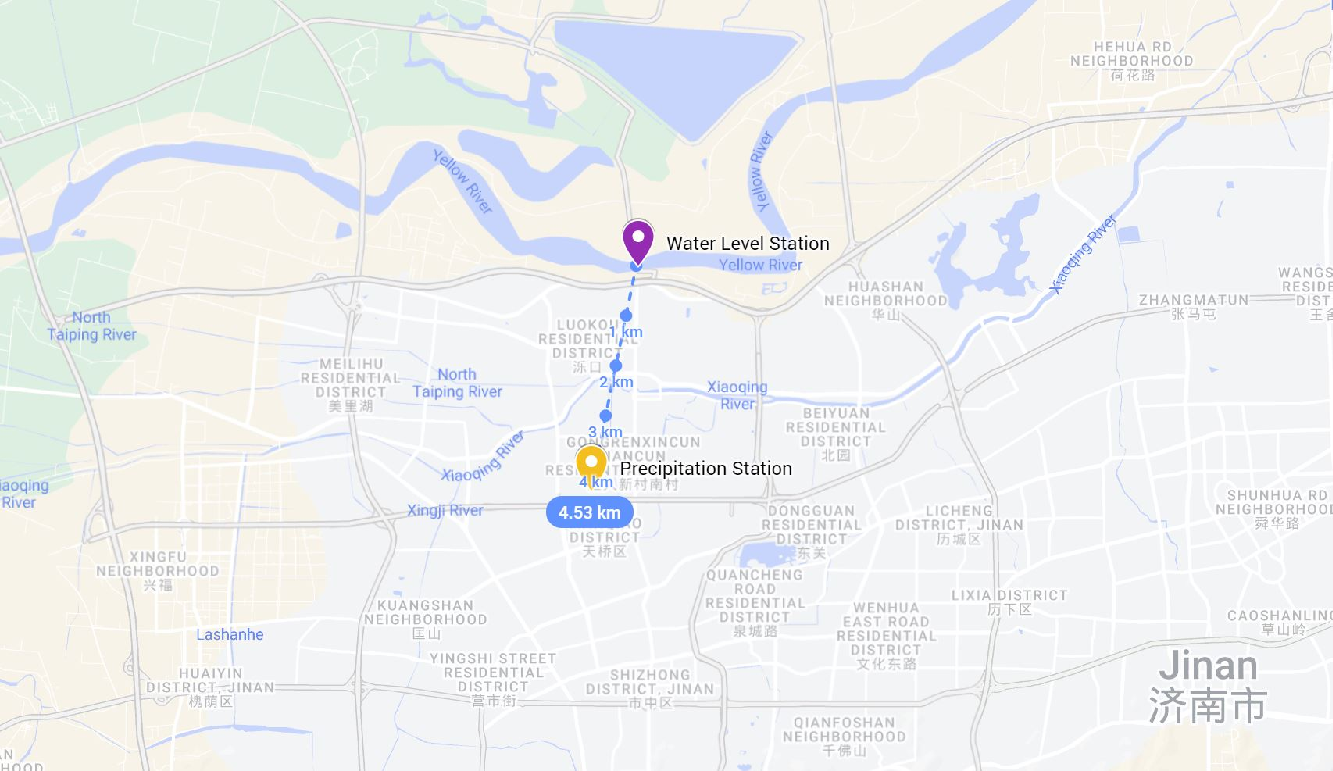}
\caption{Locations of the water level observation station and the
  precipitation observation station in Jinan, Shandong Province, China}
\label{fig:map}
\end{figure}

\begin{figure}[htbp]
  \centering
  \begin{subfigure}[b]{0.45\textwidth}
  \includegraphics[width=\textwidth]{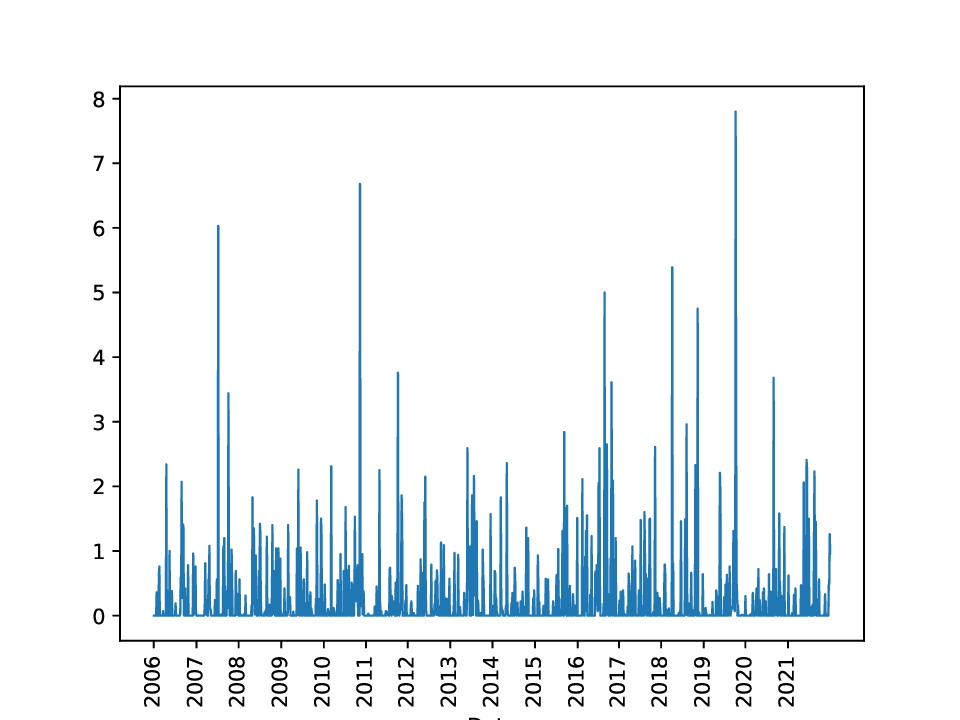}
  \end{subfigure}
  \hfill
  \begin{subfigure}[b]{0.45\textwidth}
  \includegraphics[width=\textwidth]{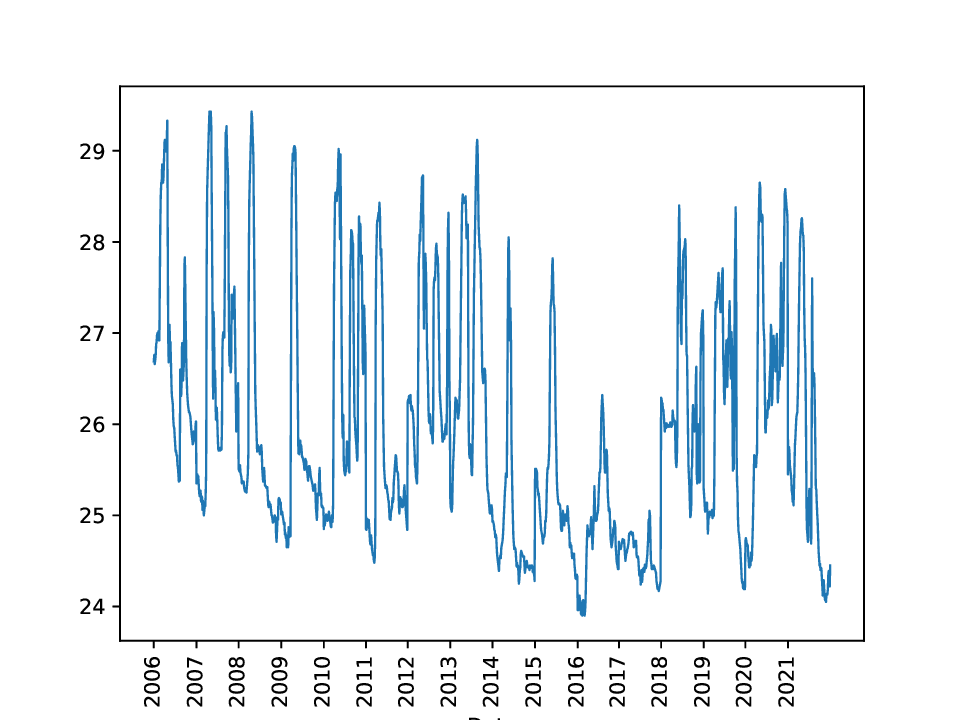}
  \end{subfigure}
\\
  \vfill
  \begin{subfigure}[b]{0.5\textwidth}
  \includegraphics[width=\textwidth]{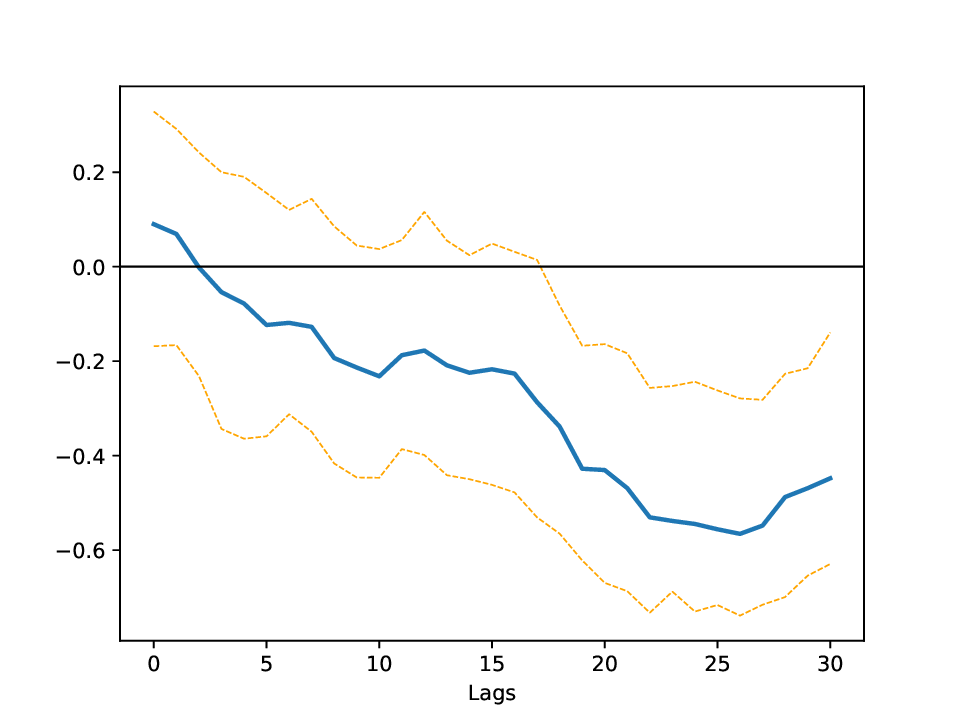}
  \end{subfigure}
  \caption{{\bf Up-left}: The daily water levels in the summer months from
$2006-2021$, $(X_t^{(5)})$. {\bf Up-Right}: The daily precipitation values in the summer months from
$2006-2021$, $(Y_t^{(5)})$. {\bf Bottom}: The centered empirical TMES
$\widehat{\delta}_0 (h)$ at lags $h=0,\ldots, 30$ for $(X_t^{(5)}, Y_t^{(5)})$ with the
bootstrapped $90\%$-confidence intervals. }
  \label{fig:summer}
  \end{figure}
  
\begin{figure}[htbp]
  \centering
  \begin{subfigure}[b]{0.45\textwidth}
  \includegraphics[width=\textwidth]{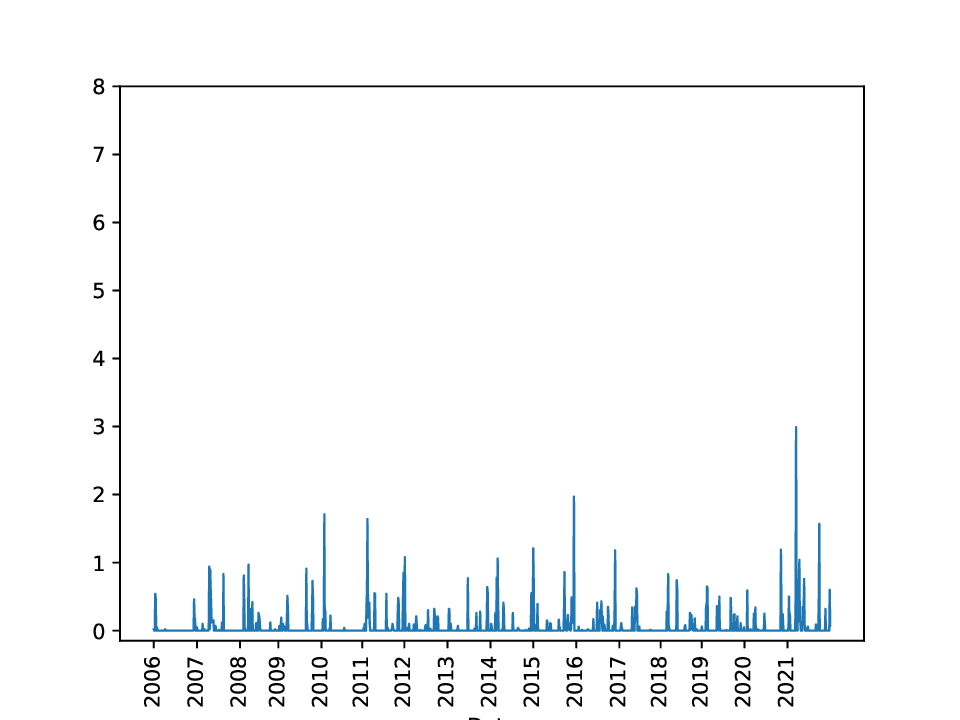}
  \end{subfigure}
  \hfill
  \begin{subfigure}[b]{0.45\textwidth}
  \includegraphics[width=\textwidth]{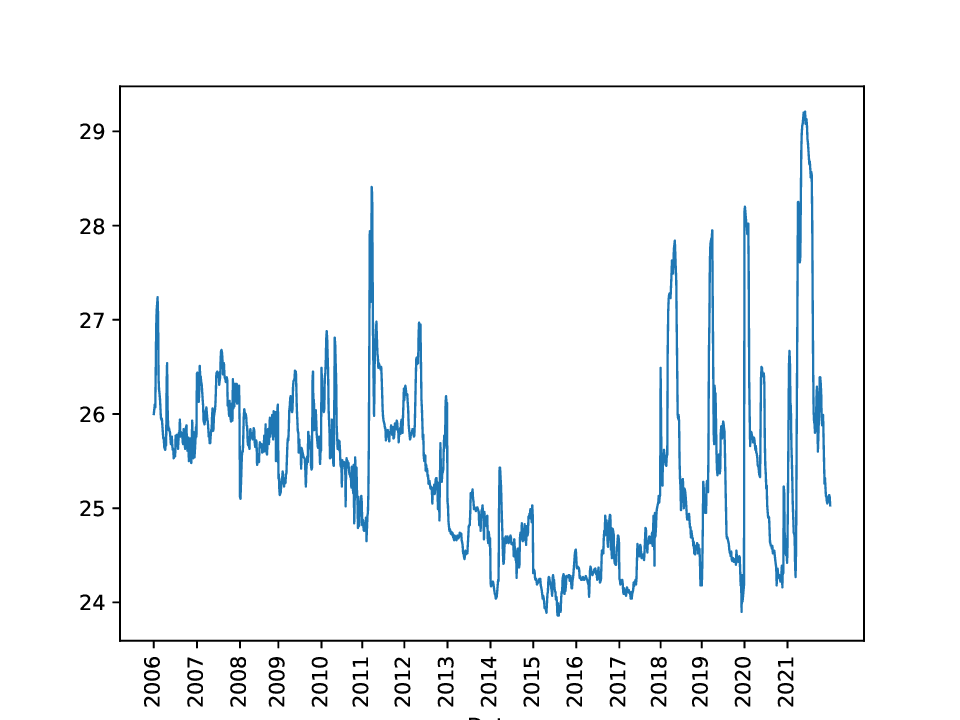}
  \end{subfigure}
\\
  \vfill
  \begin{subfigure}[b]{0.5\textwidth}
  \includegraphics[width=\textwidth]{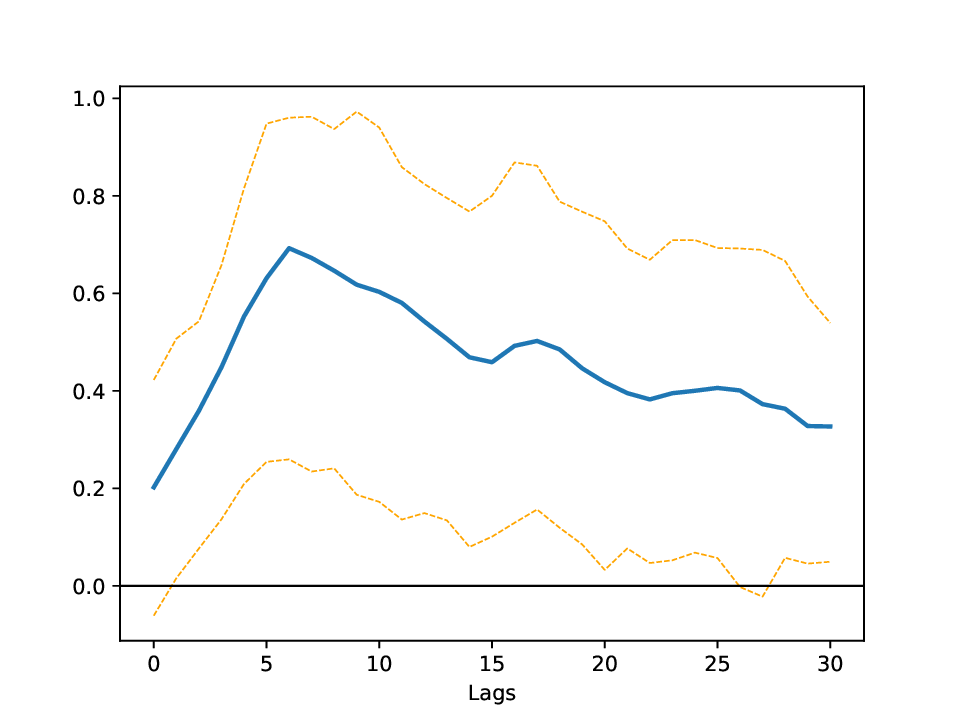}
  \end{subfigure}
  \caption{{\bf Up-left}: The daily water levels in the autumn months from
$2006-2021$, $(X_t^{(6)})$. {\bf Up-Right}: The daily precipitation
values in the autumn months from
$2006-2021$, $(Y_t^{(6)})$. {\bf Bottom}: The centered empirical TMES
$\widehat{\delta}_0 (h)$ at lags $h=0,\ldots, 30$ for $(X_t^{(6)}, Y_t^{(6)})$ with the
bootstrapped $90\%$-confidence intervals. }
  \label{fig:autumn}
  \end{figure}

\begin{appendix}
\section{Proof of Theorem \ref{thm:morvdef}}
We provide a theorem for a criterion for the convergence in
$M_O(S)$, which is similar in spirit to Theorem A.1 in
\cite{segers2017} or Theorem 2.2 and 2.3 in
\cite{billingsley1999}. The proof of Theorem~\ref{thm:mocriterion} is
a modification of Theorem A.1 in \cite{segers2017} and thus it is
omitted. 
\begin{theorem} \label{thm:mocriterion}
 Suppose that $\mathcal{A}$ is a $\pi$-system on $S$ satisfying the following two
 conditions: 
 \begin{itemize}
 \item[(C1)] There exists a decreasing sequence $(r_i)_{i\in \mathbb{N}}$ of
   positive scalars with $r_i \to 0$ as $i\to \infty$ such that for each
   $i$, there exists a neighborhood of the cone $\mathbb{C}$, say
   $N_i$, such that $N_i \in \mathbb{C}_{r_i}$ and $A\setminus N_i \in \mathcal{A}$ for
   all $A\in \mathcal{A}$. 
   \item[(C2)] Each open subset $G$ of $S$ with $\mathbb{C}$ not
     included in the closure of $G$ is a countable union of $\mathcal{A}$-sets.   
 \end{itemize}
 If $\mu_n (A) \to \mu(A)$ as $n\to \infty$ for all $A$ in $\mathcal{A}$, then $\mu_n \to
 \mu$ in $M_O(S)$ as $n\to \infty$. 
\end{theorem}

By defining $Z_t = (X_t, \mathbf{Y}_t)$ for $t\in
\mathbb{N}$, we will use a similar idea as in the proof of Theorem 4.1 in
\cite{segers2017}. If 
$(Z_t)_{t \in \mathbb{N}}$ is strictly stationary, we have $\mu^{(\infty)}
(\{(z_t)_{t\in \mathbb{N}} \in S^{(\infty)}: z_s \notin \mathbb{C}\})$ is a constant for all $s\in
\mathbb{N}$. Since $S^{(\infty)} \setminus \mathbb{C}^{(\infty)} = \cup_{s\in \mathbb{Z}} \{ (z_t)_{t\in \mathbb{N}}: z_s 
\notin \mathbb{C}\}$, we must have $\mu^{(\infty)} ((z_t)_{t\in \mathbb{N}}: z_s \notin
\mathbb{C}\}) >0$ if and only if $\mu^{(\infty)}$ is non-zero.

For integers $n\ge m \ge 0$, define the projections $Q_m: S^{(\infty)} \to
S^{(m)}$ and $Q_{n,m}: S^{(n)} \to S^{(m)}$ by  
\begin{align*}
Q_m ((z_t)_{t\in \mathbb{N}}) & = (z_1, \ldots, z_m)\,, \quad (z_t)_{t \in \mathbb{N}} \in S^{(\infty)}\,, \\
Q_{n,m} (z_{-n}, \ldots, z_n) & = (z_{-m}, \ldots, z_m)\,, \quad (z_{-n}, \ldots, z_n) \in S^{(n)}\,.
\end{align*}
The inverse images are denoted by $Q_m^{-1}$ and $Q_{n,m}^{-1} $,
inducing maps from the Borel $\sigma$-filed $\mathcal{B}(S^{(m)})$ to the Borel
$\sigma$-fileds $\mathcal{B}(S^{(\infty)})$ and $\mathcal{B}(S^{(n)})$, respectively. The projections
are continuous and $Q_m(\mathbb{C}^{(\infty)}) = Q_{n,m}(\mathbb{C}^{(n)})
= \mathbb{C}^{(m)}$. Define the measures 
\begin{align*}
\nu_u^{(m)} (\cdot)  = \frac{1}{V(u)} \p(u^{-1}(Z_t)_{t=1,\ldots,m} \in \cdot )\,,
 \quad \nu_u^{(\infty)} (\cdot )  = \frac{1}{V(u)} \p(u^{-1} (Z_t)_{t\in \mathbb{N}} \in \cdot )\,,  
\end{align*}
where $u>0$ and integer $m\ge 0$. Trivially, $(Z_t)_{t=1,\ldots, m} =
Q_m((Z_t)_{t\in \mathbb{N}}) = Q_{n,m} ((Z_t)_{t=1,\ldots,n})$ for integers $n\ge m \ge
0$. We obtain  
\begin{align*}
    \nu_u^{(n)} \circ Q_{n,m}^{-1} = \nu_u^{(\infty)} \circ Q_m^{-1} = \nu_u^{(m)}\,.
\end{align*}

The proof of Theorem~\ref{thm:morvdef} consists of two lemmas. 
\begin{lemma} \label{lem:morvdefif}
 Assume that the conditions of Theorem~\ref{thm:morvdef} are
 satisfied. If a random element $(X_t,
  \mathbf{Y}_t)_{t\in \mathbb{N}}$ is partially regularly varying with the cone
  $\mathbb{C}^{(\infty)}$ and index $\xi$, then for all $m\in \mathbb{N}$,
  the random vector $(X_t, \mathbf{Y}_t)_{t=1,\ldots,m}$ is partially
  regularly varying with the cone $\mathbb{C}^{(m)}$ and index
  $\xi$.
\end{lemma}
\begin{proof}[Proof of Lemma~\ref{lem:morvdefif}]
An application of the continuous mapping theorem (Theorem 2.3 in \cite{lindskog2014}) implies that 
\begin{align*}
\nu_u^{(m)} = \nu_u^{(\infty )} \circ Q_m^{-1} \to \mu^{(\infty)} \circ Q_m^{-1} = \mu^{(m)}
\end{align*}
given that the limit $\nu_u^{(\infty)} \to \mu^{(\infty )}$ as $u\to \infty$ in $M_O(S^{(\infty)})$ holds by assumption. 
\end{proof}

\begin{lemma} \label{lem:morvdefonlyif}
 Assume that the conditions of Theorem~\ref{thm:morvdef} are
 satisfied. If for all $m\in \mathbb{N}$,
  random vector $(X_t, \mathbf{Y}_t)_{t=1,\ldots,m}$ is partially
  regularly varying with the cone $\mathbb{C}^{(m)}$ and index
  $\xi$, then the random element $(X_t,
  \mathbf{Y}_t)_{t\in \mathbb{N}}$ is partially regularly varying with the cone
  $\mathbb{C}^{(\infty)}$ and index $\xi$. 
\end{lemma}
\begin{proof}[Proof of Lemma~\ref{lem:morvdefonlyif}]
  The relation $\nu_u^{(n)} \circ Q_m^{-1} = \nu_u^{(m)}$ for integers $n\ge m \ge 0$ implies that 
\begin{align}
\label{eq:fakeconsistent}
\mu^{(n)} \circ Q_{n,m}^{-1} = \mu^{(m)}\,.
\end{align}
Our task is to find $\mu^{(\infty)}$. Although \eqref{eq:fakeconsistent}
shows that $(\mu^{(m)})_{m\in \mathbb{N}}$ forms a self-consistent family of
measures, the Daniell-Kolmogorov extension theorem fails to construct
$\mu^{(\infty)}$ since the spaces $S^{(m)}\setminus \mathbb{C}^{(m)} $ are not
product spaces. We will follow the steps in the proof of Theorem 4.1
in \cite{segers2017} to construct $\mu^{(\infty)}$ by writing $\mu^{(\infty)}$
as a sum of finite measures in $M_O(S^{(\infty)})$.

We will denote $M_b(A)$ as the class of finite measures on the Borel
$\sigma$-field $\mathcal{B}(A)$. For integer $m\ge 0$ and $r>0$, the neighborhood of $x
\in S^{(m)}$, $N_{m,r}(x)$, is given by 
\begin{align*}
N_{m,r}((z_t)_{t=1,\ldots,m}) = \{ (\widetilde{z}_t)_{t=1,\ldots,m} \in
  S^{(m)}: d_m\big( (\widetilde{z}_t)_{t=1,\ldots,m}, (z_t)_{t=1,\ldots,m} \big) <r\}\,. 
\end{align*}
Let
\begin{align*}
\mathcal{B}_f(S^{(\infty)}) & = \bigcup_{m=0}^{\infty} \big\{ Q_m^{-1}(A): A \in \mathcal{B}(S^{(m)})
             \big\}\,, \\
  \mathcal{B}_s(S^{(m)}) & = \{ B\in \mathcal{B}(S^{(m)}): \mu^{(m)}( \partial B) =0\}\,.
\end{align*}
As discussed in the proof of Theorem 4.1 in \cite{segers2017},
$\mathcal{B}_f(S^{(\infty)})$ forms a $\pi$-system and $\mathcal{B}(S^{(\infty)})$ is  
generated by $\mathcal{B}_f(S^{(\infty)})$. Moreover, $\mathcal{B}(S^{(m)})$ is equal to the
$\sigma$-field generated by $\mathcal{B}_s(S^{(m)})$. Due to the separability of
$S^{(m)}$ and $S^{(\infty)}$, the Lindel\"{o}f property gives us that 
\begin{itemize}
\item each open subset of $S^{(\infty)}$ can be covered by a countable unions
  of open subsets in $\mathcal{B}_f(S^{(\infty)})$, 
  \item each open subsets of $S^{(m)}$ can be covered by a countable
    collection of $\mu^{(m)}$-smooth open balls. 
\end{itemize}

For integers $m\ge k \ge 1$, define the subsets $A_k^{(m)}$ of $S^{(m)}$
by
\begin{align*}
A_k^{(m)} = \big\{ ( z_1, \ldots, z_m) \in S^{(m)}: \max_{1 \le
  j \le k} d (z_j , \mathbb{C}) \le k^{-1}\big\}\,.
\end{align*}
The homogeneity of $\mu^{(m)}$ under the partial scalar mapping implies that
\begin{equation*}
\mu^{(m)}\big(\{ z_1 , \ldots, z_m) \in S^{(m)}: d (z_j ,
\mathbb{C}) =c \} \big) = 0
\end{equation*}
for all $m\ge 1$, $1 \le j \le m$ and $c>0$. Thus, $\mu^{(m)}(\partial A_k^{(m)} )=0$ for all $m\ge k \ge 1$ and
$S^{(m)} \setminus A_k^{(m)} \in \mathcal{B}_s(S^{(m)})$ is open. Moreover, $d_m((z_t)_{t=1,\ldots,m},
\mathbb{C}^{(m)}) \ge 2^{-k}/k$ for $(z_t)_{t=1,\ldots,m} \in S^{(m)} \setminus A_k^{(m)}$ and
$d_m((z_t)_{t=1,\ldots,m}, \mathbb{C}^{(m)}) \le 3k^{-1}$ for $(z_t)_{t=1,\ldots,m} \in A_k^{(m)}$. Since
$\mu^{(m)} \in M_O(S^{2m +1})$, we have $\mu^{(m)} (S^{(m)} \setminus A_k^{(m)}) <
\infty$.

For the application of the Daniell-Kolmogorov extension theorem, we
will construct a consistent family of measures embedded with product
topologies. Let $\mu_k^{(m)} (\cdot ) = \mu^{(m)} \big( \cdot \setminus A_k^{(m)} \big)$
for all $m\ge k \ge 1$. For $n\ge m$, we have $Q_{n,m}^{-1}(A_k^{(m)}) =
A_k^{(n)} $ and $\mu_k^{(n)} \circ Q_{n,m}^{-1} = \mu_k^{(m)}$ by
\eqref{eq:fakeconsistent}. Define $R_k = \mu^{(k)}(S^{(k)} \setminus
A_k^{(k)})$. We must have $0< R_k < \infty$ according to the fact that
$\mu^{(k)} \in M_O(S^{(k)})$ is nonzero and homogeneous with respect to
the partial scalar mapping. Write the probability measure $P_k^{(m)} =
R_k^{-1} \mu_k^{(m)}$ on $S^{(m)}$. Therefore, we have
\begin{align*}
P_k^{(n)} \circ Q_{n,m}^{-1} = R_k^{-1}\big( \mu_k^{(n)} \circ Q_{n,m}^{-1}
  \big) = R_k^{-1} \mu_k^{(m)} = P_k^{(m)}\,, \quad m\ge n\ge k\,.
\end{align*}
The family $(P_k^{(m)})_{m\ge k}$ for each $k$ is consistent in sense of
Chapter 4, Section 8 in \cite{pollard2002}. Moreover, $P_k^{(m)}$ is
tight due to the fact that $(S^{(m)},d_m)$ is complete and separable;
see e.g. Theorem 1.3, \cite{billingsley1999}. According to the
Daniell-Kolmogorov extension theorem (Theorem 53, \cite{pollard2002}),
there exists a tight probability measure $P_k^{(\infty)}$ on $S^{(\infty)}$ such
that
\begin{align}
\label{eq:probconsistent}
P_k^{(\infty)} \circ Q_m^{-1} = P_k^{(m)}\,, \quad m\ge k\,.
\end{align}
Write $\mu_k^{(\infty)} = R_k P_k^{(\infty)}$ and 
\begin{align*}
A_k^{(\infty)} = \big\{ (x_t)_{t\in \mathbb{Z}}: \max_{-k \le j \le k} d(x_j , \mathbb{C})
  \le k^{-1} \big\}\,, \quad k\ge 1\,.
\end{align*} 
It is easy to show that $A_k^{(\infty)} \subset \{ \mathbf{x} = (x_t)_{t\in \mathbb{Z}}:
d_{\infty}(\mathbf{x}, \mathbb{C}^{(\infty )}) \le 3k^{-1}\}$. For $B \in \mathcal{B}(S^{(\infty)})$
bounded away from $\mathbb{C}^{( \infty )}$, there exists $k_0 \ge 1$ such that
$A_k^{(\infty)} \cap B = \emptyset$ for all $k\ge k_0$. We further have
\begin{align*}
  \mu_k^{(\infty)}(A_k^{(\infty)})
  & = \mu_k^{(\infty)} (Q_k^{-1}(A_k^{(k)})) = \mu^{(k)}(A_k^{(k)} \setminus A_k^{(k)})\,,\\
\mu_k^{(\infty)} (S^{(\infty)} \setminus A_k^{(\infty)}) &= \mu^{(k)}(S^{2k+1} \setminus A_k^{(k)}) < \infty\,,
\end{align*}
that is, $\mu_k^{(\infty)}$ is finite and vanishes on $A_k^{(\infty)}$. Moreover, for
$m\ge  k \ge 1$ and $B\in \mathcal{B}(S^{(m)})$,
\begin{align*}
\mu_k^{(\infty)} ( Q_m^{-1}(B) \setminus A_k^{(\infty)}) & = \mu_k^{(\infty)}( Q_m^{-1} (B\setminus
  A_k^{(\infty)})) = \mu_k^{(m)}(B \setminus A_k^{(m)})\\ & = \mu_k^{(m)} (B) = \mu_k^{(\infty)} (Q_m^{-1}(B))\,,
  \quad l\ge k \ge 1\,. 
\end{align*}
Recall that $\mathcal{B}_f(S^{(\infty)})$ forms a $\pi$-system generating $\mathcal{B}(S^{(\infty)})$. The
$\pi-\lambda$ Theorem (see e.g. Theorem 3.3 in \cite{billingsley1995}) implies
that 
\begin{align*}
\mu_m^{(\infty)}(\cdot \setminus A_k^{(\infty)}) = \mu_k^{(\infty)} ( \cdot \setminus A_k^{(\infty)}) = \mu_k^{(\infty)}(\cdot)
  \,, \quad m\ge k \ge 0\,.
\end{align*}

We now shall define a measure $\mu^{(\infty)}$ with desired properties. 
Notice that $A_k^{(\infty)} \subset A_m^{(\infty)}$ for $m\ge k \ge 0$.
For $B\in \mathcal{B}(S^{(\infty)})$, let 
\begin{align*}
\mu^{(\infty)}(B) = \mu_1^{(\infty)} (B \setminus A_1^{(\infty)}) + \sum_{k=2}^{\infty} \mu_k^{(\infty)} \big(
  (B \cap A_{k-1}^{(\infty)}) \setminus A_k^{(\infty)} \big)\,.
\end{align*}
Trivially, we must have $\mu^{(\infty)}(\cdot \setminus A_k^{(\infty)}) = \mu_k^{(\infty)}(\cdot)$
and $\mu^{(\infty)} \circ Q_m^{-1} = \mu^{(m)}$, $m\ge 0$.

Let $B \in
\mathcal{B}(S^{(m)})$ be bounded away from $\mathbb{C}^{(m)}$. We can find $l \ge
\max(m,1)$ such that $\max_{j=1,\ldots, m} d(z_j, \mathbb{C}) > l^{-1}$
for all $(z_j)_{1\le j \le m} \in B$ and thus $A_l^{(\infty)} \cap Q_m^{-1}(B) = \emptyset$
and $A_l^{(l)} \cap Q_{l,m}^{-1}(B) = \emptyset$. Finally, since $Q_m^{-1}(B) =
Q_l^{-1}(Q_{l,m}^{-1}(B))$, we have
\begin{align*}
\mu^{(\infty)} (Q_m^{-1}(B)) = \mu_l^{(\infty)}(Q_m^{-1}(B)) =
  \mu_l^{(l)}(Q_{l,m}^{-1}(B)) = \mu^{(l)}(Q_{l,m}^{-1}(B)) = \mu^{(m)}(B)\,.
\end{align*}
Furthermore, $\mu^{(\infty)}(\{\mathbf{x}: x_0 \neq 0\}) =\sup_{k} \mu^{(k)}(S \setminus
\mathbb{C}) >0$.

We will complete the proof by find the collection $\mathcal{A}$ satisfying the conditions (C1) and (C2)
in Theorem~\ref{thm:mocriterion}. Define the collection 
\begin{align*}
\mathcal{A} = \bigcup_{m=0}^{\infty} \Big\{ Q_m^{-1}(B): B \in \mathcal{B}_s(S^{(m)})\,, d(B,
  \mathbb{C}^{(m)})>0 \Big\}\,. 
\end{align*}
By following similar arguments in the proof of Theorem 4.1 in
\cite{segers2017}, it is easy to verify that conditions (C1) and (C2)
hold for $\mathcal{A}$. The Portmanteau theorem for the $M_O$-convergence
(\cite{lindskog2014}, Theorem 2.1) ensures that for every $A=
Q_m^{-1}(B) \in \mathcal{A}$ 
with $B\in \mathcal{B}_s(S^{(m)})$ and $d(B, \mathbb{C}^{(m)})>0$, 
\begin{align*}
\frac{1}{V(u)} \p(u^{-1} (Z_t)_{t\in \mathbb{N}} \in A)= \frac{1}{V(u)} \p
  (u^{-1} (Z_t)_{t=1,\ldots,m}\in B) \to \mu^{(m)}(B) = \mu^{(\infty)}(A)\,,
\end{align*}
An application of Theorem~\ref{thm:mocriterion} completes the proof. 
\end{proof}

\section{Proof of Proposition~\ref{prop:consistency}}
We have 
\begin{align*}
  \E \big[\widehat{\delta}(h)]
  & = \frac{m_n}{n} \sum_{t=1}^n \E[X_t I_{t-h}] =
  \big( m_n \p(a_{m_n}^{-1}\mathbf{Y}_{t-h} \in A) \big) \E\big[ X_t \mid
  a_{m_n}^{-1}\mathbf{Y}_{t-h} \in A \big]\\ 
& \to \rho(0)\delta(h) = \delta(h)\,, \quad n\to \infty\,,
\end{align*}
since $\rho(0)=1$ according to the definition of $a_{m_n}$. This proves \eqref{eq:consistency1}.

\begin{align*}
 & \frac{n}{m_n}\var \big( \widehat{\delta}(h) \big) \\
  &= \frac{m_n}{n} \sum_{t=1}^n
  \E \big[ (X_t I_{t-h} -\E[X_tI_{t-h}])^2 \big]\\ 
  & \quad + 2 \frac{m_n}{n}
  \sum_{l=1}^{n-1} \sum_{t=1}^{n-l} \E \big[ (X_tI_{t-h} -
  \E[X_tI_{t-h}])(X_{t+l}I_{t+l -h} - \E[X_{t+l} I_{t+ l -h }]) \big]\\ 
  & = Q_1 +Q_2\,. 
\end{align*}
For $Q_1$,
\begin{align*}
Q_1 &= m_n \E[X_t^2 I_{t-h}] - m_n^{-1} \big( m_n \E[X_t I_{t-h}]
      \big)^2 \\
& = \E[X_t^2 \mid a_{m_n}^{-1} \mathbf{Y}_{t-h} \in A] -
      \frac{1}{m_n} \big( \E[X_t^2 \mid a_{m_n}^{-1} \mathbf{Y}_{t-h} \in A]
      \big)^2\\ 
    & \to \tau_h(0)\,.
\end{align*}
Write
\begin{align*}
  & \frac{Q_2}{2} \\
  & = \frac{m_n}{n}
  \Big( \sum_{l=1}^k + \sum_{l=k+1}^{r_n+h} + \sum_{l=r_n+h+1}^n \Big) \sum_{t=1}^{n-l} \E \big[ (X_tI_{t-h} -
  \E[X_tI_{t-h}])(X_{t+l}I_{t+l -h} - \E[X_{t+l} I_{t+ l -h }]) \big]\\ 
  & = Q_{21} + Q_{22} + Q_{23}\,.
\end{align*}
Similar as $Q_1$, we have
\begin{align*}
\limsup_{n\to \infty} Q_{21} = \sum_{l=1}^k\rho(l) \tau_h(l)\,.
\end{align*}
Due to \eqref{eq:largeneg}, without loss of generality, we assume that
$|X_t|\le m_n$. We have
\begin{align*}
|Q_{23}| \le 4m_n^3 \sum_{l=r_n +1}^{\infty} \alpha(l) \to 0\,, \quad n\to \infty\,.
\end{align*}
Here we use Theorem 17.2.1 in \cite{ibragimov1971}.

We have $\sup_l |\tau_h(l)|< \infty$ since $(\tau_h(l))_l$ is absolutely summable. For $Q_{22}$,
\begin{align*}
 &\limsup_{n\to \infty} |Q_{22}| \\
& \le  \limsup_{n\to \infty} m_n \sum_{l=k+1}^{r_n + h} \Big( \big| \E \big[ X_h X_{h+l}\mid
  a_{m_n}^{-1} (\mathbf{Y}_0, \mathbf{Y}_l) \in A\times A \big] \big| \p(a_{m_n}^{-1}(\mathbf{Y}_0, \mathbf{Y}_l) \in
  A\times A)  + (\E[X_h I_0])^2 \Big)\\ 
& \le \big( \sup_l |\tau_h(l)| \big) \limsup_{n\to \infty}
  \sum_{l=k+1}^{r_n} m_n \p \big( a_{m_n}^{-1} (\mathbf{Y}_0, \mathbf{Y}_l) \in A\times A \big)
  + O(r_n/m_n) \to 0\,, \quad k\to \infty\,.  
\end{align*}
To sum up, we have 
\begin{align} \label{eq:normvariance}
\sigma_h^2 = \lim_{k\to \infty} \limsup_{n\to \infty} \frac{n}{m_n} \var
  (\widehat{\delta}(h)) = \tau_h(0) + 2 \sum_{l=1}^{\infty} \rho(l) \tau_h(l)\,.
\end{align}

\section{Proof of Theorem~\ref{thm:clt}}
According to Lemma~\ref{lem:Xtsmallmn}, 
\begin{align*}
\sqrt{\frac{m_n}{n}} \big( \sum_{t=h+1}^n X_tI_{t-h} \1(|X_t| >m_n) - \E[
 X_tI_{t-h} \1(|X_t| >m_n) ]   \big) \overset{\p}{\to } 0\,, \quad n\to \infty\,.
\end{align*}
Therefore, without loss of generality, we assume that $|X_t|\le m_n$
for $t\in \mathbb{N}$. Write 
\begin{align*}
Z_t =  X_tI_{t-h} - \E[X_tI_{t-h}] \,, \quad t\in \mathbb{N}\,.
\end{align*} 
Then we have 
\begin{align*}
\sqrt{\frac{n}{m_n}} \big( \widehat{\delta}(h) -
  \E[\widehat{\delta}(h)] \big) = \sqrt{\frac{m_n}{n}} \sum_{t=1}^n (X_tI_{t-h} -
  \E[X_t I_{t-h}]) = \sqrt{\frac{m_n}{n}} \sum_{t=1}^n Z_t\,.  
\end{align*}
Trivially, we have $\E[Z_t] =0$, $t \in \mathbb{N}$.

We apply a big-block-small-block argument to $(m_n/n)^{0.5} \sum_{t=1}^n
Z_t$. Let the length of big blocks be $m_n$ and the length of small
blocks be $r_n$. We always assume $k_n = n/m_n$ is an integer for the
ease of presentation; the general case can be treated similarly.

The big blocks $K_{ni}$ are given by 
\begin{align*}
K_{ni} = \{(i-1) m_n+1, \ldots, i m_n \}\,, \quad i=1,\ldots, k_n\,,
\end{align*}
and $\widetilde{K}_{ni}$ is the index set consisting of all but the
first $r_n$ elements of $K_{ni}$. The small blocks $J_{ni}=K_{ni}
\setminus \widetilde{K}_{ni} $. Write $S_n(B) =
(m_n/n)^{0.5} \sum_{t\in B} Z_t$ for a set $B \subset \{1, \ldots, n\}$.

The following lemmas show that the sums over the small blocks is
asymptotically negligible. 
\begin{lemma}\label{lem:smallneg}
  Under the conditions of Theorem~\ref{thm:clt}, 
  \begin{align}
  \label{eq:smallneg}
 \var \Big( \sum_{i=1}^{k_n} S_n(J_{ni}) \Big) \to 0\,, \quad n\to \infty\,.  
  \end{align}
\end{lemma}

\begin{proof}
  We have
  \begin{align*}
    \var \Big( \sum_{i=1}^{k_n} S_n(J_{ni}) \Big)
    & \le \sum_{i=1}^{k_n} \var \big(
    S_n(J_{ni}) \big) + 2 \sum_{1\le i_1 < i_2 \le k_n} \big|
    \cov \big( S_n(J_{ni_1}), S_n(J_{ni_2}) \big) \big|\\ 
    &  = P_1 + P_2\,.
  \end{align*}
We start with
\begin{align*}
  &\var(S_n(J_{ni})) \\
  & = \frac{m_n}{n} \left( \sum_{t=(i-1)m_n +1}^{(i-1) m_n + r_n } \var
    (X_t I_{t-h}) + 2 \sum_{l=1}^{r_n -1}  \sum_{t=(i-1) m_n +1}^{(i-1) m_n
    +r_n -l} \cov( X_t I_{t-h}, X_{t+l} I_{t+l - h} )\right) \\ 
  & \le 2 \frac{m_n}{n} \sum_{t=(i-1)m_n +1}^{(i-1)m_n + r_n} \E[(X_tI_{t-h})^2] 
+ 4 \frac{r_n m_n}{n} \Big(\sum_{t=(i-1)m_n +1}^{(i-1)m_n + r_n}
    \E[(X_tI_{t-h})^2]  \Big)^2 \,. 
\end{align*}
Notice that 
\begin{align*}
& m_n\E[(X_t I_{t-h})^2] = m_n \E[X_t^2 I_{t-h}] = \E[X_t^2 \mid
  a_{m_n}^{-1} \mathbf{Y}_{t-h} \in A ] < \infty\,.
\end{align*}
Thus $\var(S_n(J_{ni})) = O(r_n/n)$. Due to the
stationarity of $(X_t, \mathbf{Y}_t)$, $P_1 = \sum_{i=1}^{k_n} \var(S_n(J_{ni})) =
O(r_n k_n/n) = O(r_n/ m_n) \to 0$ as $n\to \infty$.

For $i_1 < i_2$, 
\begin{align*}
  \big| \cov \big( S_n(J_{ni_1}), S_n(J_{ni_2}) \big) \big|
  &\le \frac{m_n}{n} \sum_{t_1
  =(i_1-1)m_n +1}^{(i_1-1)m_n + r_n} \sum_{t_2 = (i_2-1)m_n + 1}^{(i_2
  -1) m_n + r_n} \big| \cov \big(X_{t_1} I_{ t_1-h}, X_{t_2} I_{t_2 -h}
  \big) \big| \\ 
&\le \frac{m_nr_n}{n} \sum_{l=1}^{r_n} 4 m_n^2\alpha((i_2-i_1)m_n-r_n + l) \,.
\end{align*}
In the last step, We will use Theorem 17.2.1 in
\cite{ibragimov1971}. Therefore, 
\begin{align*}
  P_2
& \le  4 \sum_{i_1=1}^{k_n} \sum_{l=1}^{k_n - i_1} |\cov(S_n(J_{ni_1}),
  S_n(J_{n,i_1+l}))|\\ 
  & \le 16 \frac{k_nm_n r_n}{n} \sum_{i=1}^{k_n-1} m_n^2
  \sum_{l=1}^{r_n}  \alpha(im_n-r_n  + l) \le 16 m_n^2r_n
  \sum_{l=r_n+1}^{\infty} \alpha(l) \to 0\,.
\end{align*}
Thus, \eqref{eq:smallneg} holds.
\end{proof}

Relation \eqref{eq:smallneg} implies that $\sum_{i=1}^{k_n} S_n(K_{ni})$
has the same distribution of $\sum_{i=1}^{k_n}
S_n(\widetilde{K}_{ni})$. Assume that 
$(\widetilde{S}_n(\widetilde{K}_{ni}))_{i=1,\ldots, k_n}$ is an independent
and identically distributed sequence with $S_n(\widetilde{K}_{ni})
\overset{d}{=} \widetilde{S}_n (\widetilde{K}_{ni})$. For any $t\in \mathbb{R}$,
\begin{align*}
& \Big| \E\Big[ \prod_{l=1}^{k_n} e^{\imath tS_n(\widetilde{K}_{nl})}\Big] -
  \E\Big[ \prod_{s=1}^{k_n} e^{\imath t \widetilde{S}_n(\widetilde{K}_{ns})}
  \Big] \Big|\\
& = \Bigg| \sum_{l=1}^{k_n} \E \Big[ \big(e^{\imath t S_n(\widetilde{K}_{nl})}
  - e^{\imath t \widetilde{S}_n(\widetilde{K}_{nl})}\big) \prod_{s=1}^{l-1} e^{\imath t
  \widetilde{S}_n (\widetilde{K}_{ns}) }\prod_{s=l+1}^{k_n} e^{\imath t
  S_n(\widetilde{K}_{ns})}\Big] \Bigg|\\ 
& \le  \sum_{l=1}^{k_n} \Bigg|\E \Big[ \prod_{s=1}^{l-1} e^{\imath t
  \widetilde{S}_n (\widetilde{K}_{ns}) } \big(e^{\imath t S_n(\widetilde{K}_{nl})}
  - e^{\imath t \widetilde{S}_n(\widetilde{K}_{nl})} \big)\prod_{s=l+1}^{k_n} e^{\imath t
  S_n(\widetilde{K}_{ns})} \Big] \Bigg|\\
& \le 4 k_n \alpha(r_n) \to 0\,. 
\end{align*}
In the last step, we used Theorem 17.2.1 in
\cite{ibragimov1971}. The sums
$\sum_{i=1}^{k_n} S_n(\widetilde{K}_{ni})$ and $\sum_{i=1}^{k_n}
\widetilde{S}_n(\widetilde{K}_{ni})$ have the same limits in
distribution if existed. Moreover,
$\widetilde{S}_n(\widetilde{K}_{ni})$ and
$\widetilde{S}_n(K_{ni})$ have the same asymptotic distribution due to
Lemma~\ref{lem:smallneg}. According to Theorem 4.1 in \cite{Petrov1995},  
the central limit theorem 
\begin{align*}
\sum_{i=1}^{k_n} \widetilde{S}_n(K_{ni}) \overset{d}{\to } N \big( 0,
  \sigma_h^2\big)\,,
\end{align*}
holds with $\sigma^2_h$ defined in \eqref{eq:normvariance} if and only if the following three conditions: 
\begin{align}
  \label{eq:clt001}
  \E\Big[ \sum_{i=1}^{k_n } \widetilde{S}_n (K_{ni}) \Big] & =0\,,\\ 
  \label{eq:clt002}
 \var\Big( \sum_{i=1}^{k_n } \widetilde{S}_n (K_{ni}) \Big)  &\to \sigma^2_h\,,\quad n\to \infty\,,
\end{align}
and for every $\varepsilon>0$, 
\begin{align}
\label{eq:clt003}
\sum_{i=1}^{k_n} \E\big[ (\widetilde{S}_n(K_{ni}))^2 \1(|\widetilde{S}_n(K_{ni})| > \varepsilon) \big] \to
  0\,, \quad n\to \infty\,.
\end{align}
Combining with
Lemma~\ref{lem:smallneg}, Proposition~\ref{prop:consistency} implies that \eqref{eq:clt001}
and \eqref{eq:clt002} hold. We will prove \eqref{eq:clt003} in the
following lemma. 

\begin{lemma}\label{lem:infismall}
 Under the conditions of Theorem~\ref{thm:clt}, the limit \eqref{eq:clt003}
 holds.  
\end{lemma}
\begin{proof}[Proof of Lemma~\ref{lem:infismall}]

  According to the Cauchy-Schwarz inequality, we have  
  \begin{align*}
&\Big( \sum_{i=1}^{k_n}\E \Big[ (\widetilde{S}_n(K_{ni}))^2 \1 (|\widetilde{S}_n(K_{ni})
    |> \varepsilon)  \Big]\Big)^2 \\
  & \le c k_n^2 \E\big[ (\widetilde{S}_n  (K_{n1}))^4\big] \var(\widetilde{S}_n(K_{ni}))= \E \Big[ \Big(
    \sum_{t=1}^{m_n} Z_t \Big)^4 \Big] O(m_n/n)\,. 
  \end{align*}
  Define the set $B=\{(s_1, s_2, s_3, s_4) \in \mathbb{N}^4: 1 \le s_1 \le s_2 \le s_3 \le s_4 \le m_n\}
$ and the subsets  
\begin{align*}
B_1 &=\{(s_1, s_2, s_3, s_4) \in B:  s_2-s_1 >r_n \}\,,\\ 
 B_2 &=\{(s_1, s_2, s_3, s_4) \in B\setminus B_1: s_4-s_3 >r_n \}\,,\\  
 B_3 &=\{(s_1, s_2, s_3, s_4) \in B: s_3-s_2 >r_n , \max(s_2-s_1, s_4-s_3)\le r_n\}\,,\\
 B_4 &=\{(s_1, s_2, s_3, s_4) \in B: \max(s_2-s_1,s_3-s_2, s_4-s_3)\le r_n\}\,. 
\end{align*}
We have
\begin{align*}
& \E\Big[ \Big( \sum_{t=1}^{m_n} Z_t \Big)^4 \Big] \leq \E \big[c
  \sum_{(s_1,s_2, s_3, s_4) \in B} \prod_{j=1}^4Z_{s_j} \big]\\ 
& = c\Big( \sum_{(s_1,s_2, s_3, s_4) \in B_1} +  \sum_{(s_1,s_2, s_3, s_4) \in
  B_2} +  \sum_{(s_1,s_2, s_3, s_4) \in B_3} + \sum_{(s_1,s_2, s_3, s_4) \in
  B_4}   \Big) \E \Big[ \prod_{j=1}^4 Z_{s_j} \Big]\\ 
& = c(P_1 + P_2 +P_3 + P_4)\,.
\end{align*}

Let $c>0$ be a constant whose value may vary between lines. An application of Theorem 17.2.1 in \cite{ibragimov1971} yields that 
\begin{align*}
  |P_1| + |P_2|
& = \sum_{(s_1, s_2, s_3 ,s_4) \in B_1 }\big|  \cov(Z_{s_1}, Z_{s_2}
  Z_{s_3} Z_{s_4}) \big| +\sum_{(s_1, s_2, s_3 ,s_4) \in B_2 } \big| \cov(Z_{s_4}, Z_{s_1}
  Z_{s_2} Z_{s_3}) \big|\\
& \le c m_n^7 \sum_{h=r_n+1}^{\infty} \alpha(h)\,.
\end{align*} 
For $P_3$,
\begin{align*}
  |P_3|
& = \sum_{(s_1, s_2, s_3 ,s_4) \in B_3 }\big|  \cov(Z_{s_1}, Z_{s_2}
  Z_{s_3} Z_{s_4}) \big| + \big|\cov(Z_{s_1}, Z_{s_2})\big|\,
  \big| \cov(Z_{s_3}, Z_{s_4})\big|\\ 
& \le c r_n^2 m_n^5 \sum_{h=r_n +1}^{\infty} \alpha(h) + c r_n^2 \big( \E[X_h^2\mid
  a_{m_n}^{-1}\mathbf{Y}_0  \in A] \big)^2 \,.
\end{align*}
For $P_4$,
\begin{align*}
|P_4| \le c r_n^3\E[X_h^4 \mid a_{m_n}^{-1} \mathbf{Y}_0 \in A]\,.
\end{align*}
Therefore, we have
\begin{align*}
\E\Big[ \Big( \sum_{t=1}^{m_n} Z_{t,h} \Big)^4 \Big] \le c (r_n^2 m_n^7 +
  r_n^2 m_n^5) \sum_{l=r_n +1}^{\infty} \alpha(l) + c r_n^3  \E[X_{h+1}^4 \mid
  a_{m_n}^{-1} \mathbf{Y}_1 \in A]\,,  
\end{align*}
which implies that 
\begin{align*}
\frac{m_n}{n} \E \Big[ \Big( \sum_{t=1}^{m_n}  Z_{t}
  \Big)^4 \Big] \le c \frac{m_n^8r_n^2}{n} \sum_{h=r_n +1}^{\infty} \alpha(h) + c
  \frac{m_nr_n^3}{n} \to 0\,.
\end{align*}
\end{proof}

\section{Proof of Theorem~\ref{thm:bootclt}}
Based on Lemma~\ref{lem:Xtsmallmn}, we always assume that
  $|Z| \le m_n$ in the proof. As implied by \eqref{eq:consistency2},  
\begin{align*}
\frac{m_n}{n} \sum_{t=1}^n \big( X_t I_{t-h} - \E[X_t I_{t-h}] \big) \overset{\p}{\to
  }0\,, \quad n\to \infty\,.
\end{align*} 

Based on the bootstrap algorithm, we have
\begin{align*}
 &  \frac{m_n}{n} \var^{\star}  \Big( \sum_{t=1}^n Z_{t^{\star}}  \Big)\\ 
   & = \frac{m_n}{n} \sum_{t=1}^n \var^{\star}\big(
     Z_{t^{\star}} \big) +
 \frac{2m_n}{n} \sum_{l=1}^{n-1} \sum_{t=1}^{n-l} \cov^{\star}\big(
     Z_{t^{\star}}, Z_{(t+l)^{\star}} \big)\\ 
   & = \frac{m_n}{n} \sum_{t=1}^n (Z_t - \overline{Z}_n)^2
     + \sum_{l=1}^{n-1} \frac{2m_n(n-l) }{n^2 } \p (L_1 \ge l) \sum_{t=1}^n 
      (Z_t - \overline{Z}_n) (Z_{t+l} -
     \overline{Z}_n)   \\
   & \quad + \sum_{l=1}^{n-1} \frac{2m_n(n-l) }{n^2 }  \p (L_1 < l) \Big(\sum_{t=1}^n  
     (Z_t - \overline{Z}_n )\Big)^2\\ 
   & = Q_1 + Q_2 + Q_3\,.
\end{align*}
Since $\sum_{t=1}^n Z_t - \overline{Z}_n =0$, we have $Q_3
=0$. Due to \eqref{eq:consistency2}, we have 
\begin{align*}
\frac{m_n}{n}(\overline{Z}_n)^2 \overset{\p}{\to } 0\,, \quad 
\E \big[\frac{m_n}{n} \sum_{t=1}^n \E^{\star} [Z^2_{t^{\star}} ] \big] =\tau_h(0)\,,
\end{align*}
and thus, $\E[Q_1] =  \tau_h (0)$. 

We will show \eqref{eq:bootclt} by following the proof of
Theorem 2.1 in \cite{davis2012142}. Recall $N=N_n = \inf \{l\ge 1: \sum_{i=1}^l L_i \ge n\}$ and 
\begin{align*}
S_N = \sum_{i=1}^N S_{Ni} = \sum_{i=1}^N  \sum_{t=H_i}^{H_i+L_i
  -1} Z_t\,.
\end{align*}
When performing the stationary bootstrap, the length of the
bootstrapped sequence may exceed $n$ and we will show that the
bootstrapped sequence truncated at length $n$ has the same
asymptotic distribution as the full bootstrapped sequence.
\begin{lemma}\label{lem:bootneg}
Under the conditions of Theorem~\ref{thm:bootclt}, 
\begin{align}\label{eq:bootneg}
\p^{\star} \Big( (n/m)^{0.5} |\delta^{\star}(h) - \frac{m_n}{n}S_N)|> \varepsilon \Big) \overset{\p}{\to
  } 0\,.
\end{align}
\end{lemma}
\begin{proof}[Proof of Lemma~\ref{lem:bootneg}]
  By following the arguments in \cite{politis1994}, the distribution
  of $(n/m)^{0.5} (\delta^{\star}(h) - (m_n/n)S_N)$ is the same as 
  \begin{align*}
\Big(\frac{m_n}{n}\Big)^{0.5} \sum_{t=H_1}^{H_1+L_1-1} Z_t \text{ or }
    \Big(\frac{m_n}{n}\Big)^{0.5}\sum_{t=1}^{L_1 -1} Z_t\,, 
  \end{align*}
due to the memoryless property of $L_i$'s geometric distribution and
the stationarity of $(Z_t)$. It is sufficient to prove 
\begin{align*}
\Big(\frac{m}{n}\Big)^{0.5} \sum_{t=1}^{L_1 -1} Z_t \overset{\p}{\to } 0\,.
\end{align*}
We have
\begin{align*}
 \frac{m_n}{n} \E\Big[ \Big( \sum_{t=1}^{L_1 -1} Z_t  \Big)^2 \Big] 
& = \sum_{l=1}^{\infty} \theta(1- \theta)^{l-1} \E \Big[\frac{m_n}{n}\Big(
  \sum_{t=1}^l Z_t \Big)^2 \Big]\\ 
& = \sum_{l=1}^{\infty} \theta(1- \theta)^{l-1}  \Big(\frac{m_n}{n} \sum_{t=1}^l
  \E \big[Z_t^2\big] + \frac{2m_n}{n}\sum_{s=1}^{l-1} \sum_{t=1}^{l-s} \E \big[Z_t
 Z_{t+s} \big]\Big)\,. 
\end{align*}
We first focus on 
\begin{align*}
& \frac{m_n}{n} \sum_{t=1}^l \E \big[Z_t^2
  \big] \le \frac{2m_nl}{n}  \E \big[Z_1^2\big] < cl/n\,,
\end{align*}
where $c>4 m_n \E[Z_1^4]$ is a finite constant. The last inequality is obtained by
Proposition~\ref{prop:consistency}. Therefore, 
\begin{align*}
\sum_{l=1}^{\infty} \theta(1- \theta)^{l-1}  \frac{m_n}{n} \sum_{t=1}^l
  \E \Big[(Z_t - \overline{Z}_n)^2 \Big] \le \frac{c}{n} \sum_{l=1}^{\infty} l \theta(1- \theta)^{l-1} = c/(n\theta)\to 0\,.
\end{align*}
By using the inequality$
\Big| \E \big[Z_s Z_{s+t} \big] \Big| \le \E \big[Z_t^2 \big]$,
 we have 
\begin{align*}
\Big|\sum_{l=1}^{\infty} \theta(1- \theta)^{l-1}   \frac{2m_n}{n}\sum_{s=1}^{l-1} \sum_{t=1}^{l-s} \E \big[Z_t
 Z_{t+s} \big] \Big|\,\le \sum_{l=1}^{\infty}\theta(1 - \theta)^{l-1} (cl^2/n )\le c/(n\theta^2) \to 0\,.
\end{align*}
Thus, we have 
\begin{align*}
\frac{m_n}{n} \E\Big[ \big( \sum_{t=1}^{L_1 -1} Z_t  \big)^2 \Big] \to
  0\,, \quad n\to \infty\,,
\end{align*}
which implies \eqref{eq:bootneg}.
\end{proof}
It is sufficient to prove that 
\begin{align}\label{eq:bootneg01}
\sqrt{\frac{n}{m_n}} \Big( \frac{m_n}{n}S_N -m_n \overline{Z}_n\Big) \overset{d}{\to
  } N(0, \sigma_h^2)\,, \quad \text{conditional on } (Z_t)\,.
\end{align}
 We will replace the
left-hand side of \eqref{eq:bootneg01} with a sum of an iid sequence
of random variables. According to the stationary bootstrap algorithm,
$\E^{\star}[S_{Ni}] = \theta^{-1}\overline{Z}_n$ for $i=1,2, \ldots, N$.
Notice that 
\begin{align*}
 & (n/m_n)^{0.5} \Big(\frac{m_n}{n} S_N - m_n \overline{Z}_n \Big) -
  (m_n/n)^{0.5} \sum_{i=1}^N (S_{Ni} - \theta^{-1} \overline{Z}_n)\\ 
 & = - m_n^{0.5} (n^{0.5}- N n^{-0.5} \theta^{-1}) \overline{Z}_n = (n
   \theta)^{-0.5} \frac{N- n\theta}{\sqrt{n \theta}} \big( (m_n/n)^{0.5}
   \sum_{t=1}^nZ_t \big)\,. 
\end{align*}
As discussed in \cite{davis2012142}, $(N- n \theta)/\sqrt{n \theta}$ converges
in distribution to a normal random variable as $n \theta \to \infty$ and its
variance is bounded for all the $n$. Since $n \theta \to \infty$ according to (\ref{bootcondition3}), we have 
\begin{align*}
(n \theta)^{-0.5} \big( (m_n/n)^{0.5} \sum_{t=1}^nZ_t \big) \overset{\p}{\to
  } 0\,.
\end{align*}
Therefore, if 
\begin{align} \label{eq:bootneg02}
\sqrt{\frac{m_n}{n}} \sum_{i=1}^N (S_{Ni} - \theta^{-1} \overline{Z}_n)
  \overset{d}{\to } N(0, \sigma_n^2)\,, \quad \text{ conditional on } (Z_t)\,,
\end{align}
the limit \eqref{eq:bootclt} holds. The sequence $(S_{Ni} - \theta^{-1}
\overline{Z}_n)_{i=1,2,\ldots, N}$ forms an iid sequence conditional on $(Z_t)$ for fixed
$N$. Given the normality of $(N-n \theta)/\sqrt{n \theta}$ as $n\theta \to \infty$, we
can apply an Anscombe type argument(see e.g. Lemma 2.5.8 in
\cite{embrechts1997}) to the left-hand side of \eqref{eq:bootneg02}
and can replace $n \theta$ by a sequence of positive integers $(K_n)$
satisfying $n \theta/K_n \to 1$. Note that 
\begin{align*}
\sqrt{\frac{m_n}{n}} (S_{Ni} - \theta^{-1} \overline{Z}_n)\,, \quad i=1,\ldots,
  N\,; n=1,2,\ldots\,,
\end{align*} 
is a triangular array. Similar to the proof of Theorem 2.1 in
\cite{davis2012142}, we will apply a classical central limit theorem
for triangular arrays of independent random variables conditional on
$(Z_t)$ by verifying the Lyapunov condition given $(Z_t)$, 
\begin{align*} 
(m_n/n)^{3/2}K_n \E^{\star} \big[ \big| S_{N1} - \theta^{-1}\overline{Z}_n
  \big|^3 \big] \sim \frac{m_n^{3/2}}{n^{1/2}} \theta \E^{\star}\big[ \big| S_{N1} - \theta^{-1}\overline{Z}_n
  \big|^3 \big] \overset{\p}{\to }0\,.
\end{align*}

\begin{lemma} \label{lem:lyapunov}
 Under the conditions of Theorem~\ref{thm:bootclt}, 
 \begin{align} \label{eq:lyapunov}
 \frac{m_n^{3/2}}{n^{1/2}} \theta \E^{\star}\big[ \big| S_{N1} -
   \theta^{-1}\overline{Z}_n \big|^3 \big] \overset{\p}{\to }0 \,.
 \end{align}
\end{lemma}
\begin{proof}[Proof of Lemma~\ref{lem:lyapunov}]
\begin{align*}
\E\Big[ \E^{\star}\big[ \big| S_{N1} -
  \theta^{-1}\overline{Z}_n \big|^3 \big] \Big]
& = \sum_{l=1}^{\infty} \theta(1-\theta)^{l-1} \E \Big[  \Big| \sum_{t=1}^l(Z_t - \overline{Z}_n)\Big|^3
  \Big]\\
& \le c\sum_{l=1}^{\infty} \theta(1 - \theta)^{l-1} \E \Big[ \Big| \sum_{t=1}^l
  Z_t \Big|^3 \Big]  + c \sum_{l=1}^{\infty} l^3\theta(1 - \theta)^{l-1}
  \E[|\overline{Z}_n|^3]\\
& = Q_1 +Q_2\,.  
\end{align*}
As shown in Theorem~\ref{thm:clt}, $(m_nn)^{1/2}
\overline{Z}_n \overset{d}{\to } N(0, \sigma_h^2)$ as $n\to \infty$. We have
\begin{align*}
\frac{m_n^{3/2} \theta}{ n^{1/2}} Q_2 \le c \frac{m_n^{3/2} \theta}{n^{1/2}}
  (\theta (m_nn)^{1/2})^{-3} = c (\theta^2 n^2)^{-1} \to 0\,. 
\end{align*}
Let $A = \{(t_1, t_2, t_3): 1\le t_i \le l\,, i=1,2,3\}$. We are
interested in the subsets of $A$. 
\begin{align*}
  A_1
& = \{(t_1, t_2, t_3 )\in A: |t_i - t_j|> r_n +2h\,, i\neq j\,,
  i,j=1,2,3\}\,,\\  
  A_2
& = \{ (t_1, t_2, t_3) \in A: |t_i -t_1|>r_n + 2h\,, |t_2-t_3|\le
  r_n+2h\,, i=2,3\}\,,\\ 
  A_3
& = \{ (t_1, t_2, t_3) \in A: |t_1-t_2| \le r_n +2h, \min(|t_1 - t_3|,
  |t_2 -t_3|) >r_n + 2h\}\,,\\ 
  A_4
& = \{ (t_1, t_2, t_3) \in A: |t_1-t_3| \le r_n +2h, \min(|t_1 - t_2|,
  |t_2 -t_3|) >r_n + 2h\}\,,\\ 
  A_5
& = \{ (t_1, t_2, t_3) \in A: \max(|t_1 -t_2|, |t_2-t_3|, |t_3-t_1|) \le
  2r_n + 4h\}\,.
\end{align*}
Trivially, $A \subset \bigcup_{i=1}^5 A_i$. We will use Theorem 17.2.1 in
\cite{ibragimov1971} to provide upper bounds for 
\begin{align*}
\sum_{(t_1, t_2, t_3) \in A_i} \big| \E[|Z_{t_1}|Z_{t_2} Z_{t_3}] \big|\,, \quad
  i=1,2,\ldots, 5\,.
\end{align*}
Recall that $\E[Z_t] =0$. We take $\sum_{(t_1, t_2, t_3) \in A_2} \big|
\E[|Z_{t_1}|Z_{t_2} Z_3] \big|$ as example.
\begin{align*}
& \sum_{(t_1, t_2, t_3) \in A_2} \big| \E[|Z_{t_1}|Z_{t_2} Z_{t_3}] \big|\\
& \leq \sum_{(t_1, t_2, t_3) \in A_2} \big| \E[|Z_{t_1}|Z_{t_2} Z_{t_3}] -
  \E[|Z_{t_1}|] \E[Z_{t_2}Z_{t_3}] \big| + \big|  \E[|Z_{t_1}|]
  \E[Z_{t_2}Z_{t_3}] \big| \\ 
& \le \sum_{(t_1, t_2, t_3) \in A_2} (cm_n^3 \alpha\big( \min(|t_1-t_3|,
  |t_1-t_2|) \big) + cm_n^{-2})\\ 
& \le cm_n^3 r_nl \sum_{h_1=r_n+1}^{\infty} \alpha(h_1) + cm_n^{-2}r_nl^2\,. 
\end{align*}
By using similar arguments, we have
\begin{align*}
 \sum_{(t_1, t_2, t_3) \in A_1} \big| \E[|Z_{t_1}|Z_{t_2} Z_{t_3}] \big|& \le c m_n^3l^2\sum_{h_1=r_n+1}^{\infty} \alpha(h_1)\,,\\ 
 \sum_{(t_1, t_2, t_3) \in A_3} \big| \E[|Z_{t_1}|Z_{t_2} Z_{t_3}] \big|
& \le c m_n^3r_n l\sum_{h_1=r_n+1}^{\infty} \alpha(h_1)\,,\\ 
 \sum_{(t_1, t_2, t_3) \in A_4} \big| \E[|Z_{t_1}|Z_{t_2} Z_{t_3}] \big|
& \le c m_n^3r_n l\sum_{h_1=r_n+1}^{\infty} \alpha(h_1)\,,\\ 
 \sum_{(t_1, t_2, t_3) \in A_5} \big| \E[|Z_{t_1}|Z_{t_2} Z_{t_3}] \big|
& \le cm_n^{1/2}r_n^2l\,.
\end{align*}
For $Q_1$, we have
\begin{align*}
 Q_1
& \le  \sum_{l=1}^{\infty} \theta( 1- \theta)^{l-1} \Big( \sum_{A_1} + \sum_{A_2} +  \sum_{A_3} + \sum_{A_4} + \sum_{A_5}  \Big) \big|  \E\big[ |Z_{t_1}| Z_{t_2} Z_{t_3} \big]
  \big|\\ 
& \le  c (m_n^3 \theta^{-2}+ m_n^3r_n \theta^{-1})  \sum_{h_1 =r_n +1}^{\infty}
  \alpha(h_1) + c r_n m_n^{-2} \theta^{-2} + c m_n^{1/2} r_n^2 \theta^{-1} \,.
\end{align*}
Therefore, we have 
\begin{align*}
\frac{m_n^{3/2} \theta}{n^{1/2}} Q_1 \le c \Big( \frac{m_n^{9/2}}{n^{1/2}
  \theta} + \frac{m_n^{9/2}r_n}{n^{1/2}}  \Big)\sum_{h_1 =r_n +1}^{\infty}
\alpha(h_1) + c \frac{r_n}{m_n^{1/2}n^{1/2}\theta} + c \frac{m_n^{2}r_n}{n^{1/2}}\to 0\,, \quad n\to \infty\,. 
\end{align*}

\end{proof}

\end{appendix}


\bibliographystyle{imsart-number}
\bibliography{all}
\end{document}